\documentclass[11pt]{article}

\usepackage[margin=1.3in]{geometry}

\usepackage{natbib}

\usepackage{graphicx}%
\usepackage{multirow}%
\usepackage{amsmath,amssymb,amsfonts}%
\usepackage{amsthm}%
\usepackage{mathrsfs}%
\usepackage[title]{appendix}%
\usepackage{xcolor}%
\usepackage{textcomp}%
\usepackage{manyfoot}%
\usepackage{booktabs}%
\usepackage{algorithm}%
\usepackage{algorithmicx}%
\usepackage{algpseudocode}%
\usepackage{listings}%

 \usepackage{natbib}

\usepackage{authblk}

\usepackage{hyperref}

\theoremstyle{plain}
\newtheorem{theorem}{Theorem}[section]
\newtheorem{lemma}{Lemma}[section]
\newtheorem{corollary}{Corollary}[section]
\newtheorem{proposition}{Proposition}[section]

\theoremstyle{definition}
\newtheorem{definition}{Definition}[section]

\theoremstyle{remark}
\newtheorem{remark}{Remark}[section]

 \usepackage{bm}
\usepackage{dsfont}
\usepackage{mathabx}

\DeclareMathOperator{\gap}{gap}

\DeclareMathOperator*{\argmax}{argmax}

\DeclareMathOperator*{\supp}{supp}
\DeclareMathOperator*{\Ker}{Ker}
\DeclareMathOperator*{\rank}{rank}

\DeclareMathOperator{\relint}{relint}

 \DeclareMathOperator{\id}{Id}
 \DeclareMathOperator{\Span}{span}
 \DeclareMathOperator{\diag}{diag}
 \DeclareMathOperator{\Tr}{Trace}

\usepackage[mathscr]{eucal}
\DeclareMathOperator{\cA}{\mathscr{A}}

\DeclareMathOperator{\cC}{\mathscr{C}}
\DeclareMathOperator{\cD}{\mathscr{D}}

\DeclareMathOperator{\cF}{\mathscr{F}}

\DeclareMathOperator{\cO}{\mathscr{O}}

\DeclareMathOperator{\cS}{\mathscr{S}}

\DeclareMathOperator{\cZ}{\mathscr{Z}}

\DeclareMathOperator{\fS}{\mathfrak{S}}

\newcommand{\Exp}{\mathrm{Exp}}    
\newcommand{\Log}{\mathrm{Log}}

\DeclareMathOperator{\KL}{KL}

\renewcommand{\leq}{\leqslant}
\renewcommand{\le}{\leqslant}
\renewcommand{\geq}{\geqslant}
\renewcommand{\ge}{\geqslant}
\renewcommand{\bar}{\overline}

\newcommand{\norm}[1]{\left\| #1 \right \|} 

\newcommand{\absv}[1]{\left| #1  \right|} 

\newcommand{\dprod}[1]{\left\langle #1 \right\rangle}

\newcommand{\etax}{\eta_{x}}
\newcommand{\etay}{\eta_{y}}

\newcommand{\SR}[1]{\rho\left(#1\right)}
\newcommand{\tSR}[1]{\rho(#1)}
\newcommand{\eig}[1]{\mathrm{eig}\left( #1\right)}
\newcommand{\teig}[1]{\mathrm{eig}\ ( #1 )}
\newcommand{\eigTilde}[1]{\widetilde{\mathrm{eig}}\left( #1\right)}

 \DeclareMathOperator{\RR}{\mathbb{R}}
\DeclareMathOperator{\NN}{\mathbb{N}}

\DeclareMathOperator{\CC}{\mathbb{C}}
\DeclareMathOperator{\DD}{\bar{\mathbb{D}}}

\newcommand{\ones}{\mathbf 1}

 \newcommand{\vx}{v^x}
 \newcommand{\vy}{v^y}
 
\newcommand{\V}{\mathbf{w}}

 \newcommand{\oneinf}{1 \rightarrow \infty} 

 \newcommand{\prodSimplex}{\smash{\widehat{\Delta}}} 
\newcommand\restr[2]{{
  \left.\kern-\nulldelimiterspace 
  #1 
  \vphantom{|} 
  \right|_{#2} 
  }}

 \newcommand{\Phim}{\Phi_m}

\newcommand{\px}{x}
\newcommand{\py}{y}
\newcommand{\ppx}{x'}
\newcommand{\ppy}{y'}
\newcommand{\pxp}{x^{+}}
\newcommand{\pyp}{y^{+}}

\newcommand{\gz}{\bar z}

\newcommand{\FP}{\widetilde \cZ(\restr{\Phi_m}{\prodSimplex})}
\newcommand{\FPone}{\widetilde \cZ(\restr{\Phi_1}{\prodSimplex})}
\newcommand{\tZ}{\widetilde Z}
\newcommand{\tx}{\widetilde x}
\newcommand{\ty}{\widetilde y}
\newcommand{\Zfp}{[\widetilde x,\widetilde y,\widetilde x,\widetilde y]}
 \newcommand{\NE}{\cZ^{\star}}
\newcommand{\Zs}{Z^\star}
\newcommand{\Zne}{[x^\star,y^\star,x^\star,y^\star]}
\newcommand{\Lpar}{{L^{\parallel}_{\widetilde Z}}}

\title{Stability and Convergence of Optimistic Exponential Weights with Asymmetric Step Sizes in Bimatrix Games}

\author[1]{Hédi Hadiji}
\author[2]{Sarah Sachs}

\affil[1]{Laboratoire des Signaux et Systèmes, CentraleSupélec, Paris, France}
\affil[2]{School of Mathematics, University of Bristol, Bristol, United Kingdom}

\date{}

\begin{document}

 \maketitle

  \abstract{We study two-player bimatrix games under the optimistic exponential weights method, allowing the players to use different step sizes, $\etax$ and $\etay$.  Our main finding is that the product $\etax\etay$, rather than the individual step sizes, determines both equilibrium stability and last-iterate convergence to an equilibrium.
For zero-sum games, assuming finitely many fixed points, we give a sufficient condition on $\etax\etay$ for global last-iterate convergence. The condition is practically relevant and helps explain empirical behavior beyond previously known guarantees. For general bimatrix games, we derive an almost-tight threshold, again in terms of $\etax\etay$, separating asymptotic stability from instability. We recover several known results and useful special-case step-size bounds, and illustrate our findings numerically.
  }

    \maketitle

\section{Introduction}
 We consider bimatrix games $\Gamma(A,B)$ with payoff matrices $A,B^\top \in \mathbb{R}^{d_x \times d_y}$. 
 Denote by $\Delta_{d_x}$ and $\Delta_{d_y}$  the  probability simplices in $\RR^{d_x}$ (respectively $\RR^{d_y}$). 
The $x$-player chooses $x \in \Delta_{d_x}$ to maximize $x^\top A y$, while the $y$-player chooses $y \in \Delta_{d_y}$ to maximize $  y^\top B x$.  
A Nash equilibrium $[x^\star,y^\star]$ satisfies
\[
x^\star \in \arg\max_{x \in \Delta_{d_x}} x^\top A y^\star,
\qquad
y^\star \in \arg\max_{y \in \Delta_{d_y}} y^\top B x^\star .
\]
 
 A central question in game theory and learning in games is whether simple, efficient iterative algorithms converge to such equilibrium points. Classical dynamics such as Gradient Descent-Ascent and Multiplicative Weights Update are known to exhibit cycling and fail to converge even in bilinear games (see, e.g., \cite{bailey2018multiplicative,cheung2019vortices,mertikopoulos2018cycles}). 
This has motivated significant interest in modifications of these algorithms that mitigate this issue.  A prominent approach is the \emph{optimism} framework \citep{chiang2012online,Rakhlin:2013aa,syrgkanis2015fast}, which underlies algorithms such as Optimistic Gradient Descent–Ascent and optimistic exponential weights (optEW). 
In this paper, we study the stability of Nash equilibria under optimistic exponential weights dynamics (and variants thereof) in two-player bimatrix games. That is, all coordinates $i \in \{1, \dots, d_x\}$ and $j \in \{ 1, \dots, d_y\}$ are updated as
\begin{equation}\label{eq:def_ew}
x_{t+1, i} \propto \, x_{t,i} \exp( \etax [A(2 y_{t} - y_{t-1} )]_i)\quad \text{ and } \quad y_{t+1,j} \propto\,  y_{t,j} \exp(\etay [B(2 x_{t} -  x_{t-1})]_j)\, 
,
\end{equation}
 where $\etax, \etay >0$ are the step sizes. There are two fundamentally different settings for the step size choices: (1) constant step sizes and (2) time-dependent step sizes. 
We study the convergence behavior of $(z_t)_{t \in \NN}=((x_t,y_t))_{t \in \NN}$ under constant, potentially unequal  step sizes. We refer to this setting as using \emph{asymmetric step sizes}, noting that related literature sometimes describes it as  ‘two-time-scale’ step sizes \citep{lin2025two}. We adopt the former terminology to clearly distinguish our setting from time-varying step size schemes for saddle-point problems with stochastic feedback. In the classical two-time-scale framework, the step size sequences $(\eta_{x,t})$ and $(\eta_{y,t})$ are assumed to be non-summable but square summable, with $\eta_{x,t} / \eta_{y,t} \rightarrow 0$ as $t \rightarrow \infty$ \citep{Borkar:1997aa,Borkar:2024aa}.  
  
   Existing analyses often rely on simultaneously controlling both step sizes, which results in a step size requirement on $\etax = \etay$, or controlling the individual step sizes, which results in a step size requirement on $\max\{\etax,\etay\}$.  Such conditions do not capture the product dependence as suggested by empirical observations (see Figure~\ref{fig:ExamplenRPS}).  Our results align with these experiments. 
    
\paragraph{Contributions: }  \begin{enumerate}
 \item \textbf{Global convergence results under sufficient step size conditions: }
   We provide a sufficient product-type condition on the step sizes under which optimistic exponential weights dynamics in zero-sum games, initialized in the relative interior, converge globally to the set of fixed points. Moreover, if the fixed-point set is finite, then the last iterate converges to a Nash equilibrium.

  \item \textbf{Local convergence for general bimatrix games:}
        We study the stability of equilibria of the optEW dynamics. We give a characterization of the stability in terms of the product of the step sizes $\etax\etay$ and the Jacobian at an equilibrium. The step size criterion is almost tight, meaning that if $\etax\etay \in (0,c)$ stability follows, and if $\etax\etay \in (c, \infty)$ instability follows. However, the case $\etax \etay = c$ is indeterminate.   
        \item \textbf{Results for special classes of games:} Building on our asymptotic stability result, we derive several known and new results for special classes of games. In all these cases, the resulting threshold on the step size product $\etax\etay$ can be expressed in closed form.
  \item \textbf{Empirical Studies and Open Research Directions: } We provide experiments that (a) show that our theoretical results closely match observed behavior and (b) highlight open research questions that cannot be explained by our results.
\end{enumerate}
Our global convergence result gives a practical step-size condition and covers a broader range of step sizes than previously established. Our result on asymptotic stability is mainly of theoretical interest. However, it provides a foundation for explaining the frequently observed phenomenon that algorithms can perform well even beyond the step-size regimes covered by existing guarantees.
 \paragraph{Motivation and Broader Context: }
While our results are primarily of fundamental theoretical interest, our findings are also motivated by several practical considerations.
For example, for non-convex-concave optimization, asymmetric step sizes are often used to guarantee convergence of the concave maximization subproblem while maintaining overall stability (see, e.g., \cite{lin2025two} and references therein). Similar asymmetries arise in bilevel optimization, where the inner problem is typically updated more aggressively than the outer problem (e.g., \cite{Hong:2023aa}), and in stochastic games, where players may experience different feedback variances (e.g., \cite{Sayin:2022aa}).  
  \paragraph{Related Literature: }
 Our work is most closely related to the work by \citet{de2025optimistic}, who analyze optimistic gradient descent in unconstrained general-sum bilinear games and characterize sharp stability
thresholds via the eigenvalues of the induced linear system. Although their focus
is not on asymmetric step sizes, their analysis implies a similar
product-dependence of the step sizes for optimistic gradient descent in unconstrained bilinear games. For more details and an in-depth discussion, see Section~\ref{sec:CompareResultsJacobian}. 

  \citet{fiez2021local}  also study the role of asymmetric step sizes in game dynamics, showing that gradient descent-ascent with a sufficiently large but finite timescale-separation ratio converges locally to strict local minmax/Stackelberg equilibria in smooth nonconvex-nonconcave zero-sum games, while non-equilibrium critical points become unstable. 

Conceptually, our work is also closely related to the research on replicator dynamics and discretized variants thereof. The exponential weights method can be viewed as an Euler discretization of the replicator dynamics. This is a well-known connection; for details, see, e.g., \citet{Falniowski2025}. The discrete-time nature of multiplicative weights is
known to produce recurrence, cycling, and chaotic behavior in games
\citep{mertikopoulos2018cycles,Falniowski2025}.  Examples of stability analysis for replicator dynamics include \citet{WeibullJorgen1995EGT}, \citet{Hofbauer_Sigmund_1998}, \citet{Sandholm:2010aa}, and references therein.

For the equal step size regime, a large body of work establishes convergence guarantees of optimistic or
extra-gradient-type dynamics under suitable step size conditions. These methods
have a long history: extra-gradient and related prediction-correction schemes go
back at least to \citet{korpelevich1976extragradient} and  
 \citet{Popov:1980aa}. In online learning, optimistic mirror descent,
optimistic follow-the-regularized-leader, and optimistic exponential weights arise from the
predictable-sequences framework of
\citet{chiang2012online,Rakhlin:2013aa}. They have become central in the study
of last-iterate convergence in games, since methods such
as gradient descent-ascent and multiplicative weights may cycle or exhibit
unstable behavior even for zero-sum games 
\citep{bailey2018multiplicative,mertikopoulos2018cycles,cheung2019vortices}.

In convex-concave and monotone settings, this optimistic structure yields
positive stability results. Asymptotic convergence guarantees were proved for
optimistic mirror descent and stochastic extra-gradient variants by
\citet{mertikopoulos2018optimistic} and \citet{hsieh2019convergence}.
\citet{Daskalakis2018TheLP} studied the limit points and local stability of optimistic gradient dynamics, while  \citet{daskalakis2019lastiterate}  proved last-iterate convergence for optimistic multiplicative-weights dynamics. \citet{lei2021last} extended
this to local last-iterate convergence with constant step size using a spectral analysis approach.  See Section~\ref{sec:CompareResultsJacobian} for a detailed discussion. We note that the focus of this line of literature is on improving the convergence rates; thus, comparisons should be made with caution due to the conceptual differences.

%
\subsection{Notation}

We use calligraphic letters for sets, e.g., $\cZ$. Sets of equilibria are
denoted by a superscript star, e.g., $\cZ^\star$, and, when a distinction
is needed, sets of fixed points of an operator $T$ by tildes, e.g.,
$\widetilde{\cZ}(T)$. For a set $\cA$, we write $\relint \cA$ for its
relative interior. We denote by $\RR_+$ the set of non-negative real
numbers and use the convention that $0\notin\NN$. For $d\in\NN$,
$
\Delta_d := \{x\in\RR_+^d : \ones_d^\top x = 1\}
$
denotes the probability simplex in $\RR^d$.

We denote by $\ones_d$ the all-ones vector in $\RR^d$, omitting the
subscript whenever the dimension is clear from context, and by
$e_1,\dots,e_d$ the canonical basis vectors of $\RR^d$. For matrix 
$M\in\RR^{n\times d}$, we define
$
\norm{M}_{\oneinf} := \max_{i,j}|M_{i,j}|,
$
and denote by $\norm{M}_{\ell_2}$ its $\ell_2$ operator norm. Unless
otherwise specified, $\norm{\, \cdot\, }$ denotes the Euclidean norm. For
$v\in\RR^d$ and $\delta>0$,
$
B_\delta(v):=\{z\in\RR^d:\norm{z-v}<\delta\}
$
denotes the open ball of radius $\delta$ centered at $v$.

For $v\in\RR^d$, $\diag(v)\in\RR^{d\times d}$ denotes the diagonal matrix
with diagonal $v$. We use $\odot$ for the Hadamard (entry-wise) product
and $\Exp$ for the component-wise exponential map.
  
\section{Algorithm and General Results}\label{sec:algo}
Our analysis builds on techniques from discrete dynamical systems. Thus, we define an
operator ${\Phi_m}$ corresponding to the optimistic exponential weights method in a
two-player game.
The optimistic updates depend on the last two gradients seen, so the corresponding
mapping operates on the product of simplices 
$\prodSimplex := (\Delta_{d_x} \times \Delta_{d_y}) \times (\Delta_{d_x} \times \Delta_{d_y})
\subset \RR^{2(d_x + d_y)}  $. 

We also introduce a parameter $m \in (0,1]$, which controls the impact of the optimistic 
updates. We recover the
usual optEW (as stated in \eqref{eq:def_ew}) when $m = 1$.  In the limit $m\downarrow0$, the update approaches standard exponential weights. For any vector $v \in \RR^d$ with $v^\top \ones \neq 0$, define  
\[
  P_d(v) = \bigg(\sum_{i=1}^d v_i \bigg)^{-1} v \, ,
\]
and   
\begin{multline*}
  \cD_{{\Phi_m}}  = 
  \bigg\{
    [\px, \py, \ppx, \ppy] \in  \RR^{d_x} \times \RR^{d_y}  \times \RR^{d_x} \times \RR^{d_y}
    \quad \Big| \\
    \sum_{i = 1}^{d_x} \px_{i} \exp(\etax [A((1 + m)\py - m \ppy)]_i)  \neq 0
    \quad \text{and} \quad
    \sum_{j = 1}^{d_y} \py_{j} \exp(\etay [B((1 + m)\px - m \ppx)]_j)  \neq 0
  \bigg\}
\end{multline*}
The map we consider is
\begin{align}\label{def:Phi}
  {\Phi_m} :
  \left\{
  \begin{aligned}
  &  \cD_{\Phi_m}
  &
  \to 
  &  \RR^{2(d_x + d_y)}\\
  &
  \begin{pmatrix} 
    \px\\ \py\\ \ppx\\ \ppy
  \end{pmatrix} 
  &
  \mapsto
  &
  \begin{pmatrix}
  P_{d_x}\big( \px \odot \Exp(\etax A((1+m)\py - m\ppy))\big)\\[4pt]
  P_{d_y}\big( \py \odot \Exp(\etay B((1+m)\px - m\ppx))\big)\\[4pt]
  \px\\[4pt]
  \py
  \end{pmatrix} 
  \end{aligned} 
  \right.\,.
\end{align}
 We note that on the domain $\cD_{{\Phi_m}}$,  $P_d$ is always well defined and   ${\Phi_m}( \prodSimplex) \subseteq \prodSimplex $. Moreover, the sequences $(x_t)_{t \in \NN}$ and $(y_t)_{t \in \NN}$ are iterates of the optimistic exponential weights method, cf. \eqref{eq:def_ew}, if and only if 
$Z_t = (x_t, y_t, x_{t-1}, y_{t-1})$ satisfies
$ Z_{t+1} = {\Phi_1}(Z_t) $ for all $t \in \NN$ and $Z_0 \in \prodSimplex$.
Ultimately, we are interested only in the iterates of ${\Phi_m}|_{\prodSimplex}$. Crucially, we note that the repeated iterations in the ambient space $\cD_{\Phi_m}$ are not well defined since $\Phi_m(\cD_{\Phi_m})$ is not necessarily a subset of $\cD_{\Phi_m}$. We emphasize that this is not used, and the ambient extension is used solely to compute the differential.
\paragraph{Fixed points of ${\Phi_m}$}
We start with a characterization of the fixed points for
$\restr{{\Phi_m}}{\prodSimplex}$.

\begin{theorem}\label{thm:FPPhi}
For any $m \in (0,1]$ and $\etax,\etay > 0$, the set of fixed points of $\restr{{\Phi_m}}{\prodSimplex}$ is:
\[
  \FP
  = \left\{ [\px, \py, \ppx, \ppy] \in \prodSimplex \;\middle|\; 
  \begin{aligned}
  &\px = \ppx \in \Delta_{d_x}, \; \py=\ppy \in \Delta_{d_y}, \\
  &\forall i,j \in \supp(\px), \; (A\py)_i = (A\py)_j, \\
  &\forall k,l \in \supp(\py), \; (B\px)_k = (B\px)_l
  \end{aligned}
  \right\}
\]
\end{theorem}
  The proof is deferred to Appendix~\ref{appendix:FPPhi}. 
  We define the set of points corresponding to the set of Nash equilibria lifted in dimension as  \[\NE = \{[x^\star,y^\star,x^\star,y^\star] :
[x^\star,y^\star]  \text{ is a Nash equilibrium} \} \subseteq \prodSimplex\,.\] 
With a slight abuse of terminology, we also refer to the higher-dimensional set $\NE$ as  ‘the set of Nash equilibria’. 
\begin{remark}\label{rem:strictInclusion}
  All $\Zs \in \NE$ are fixed points of ${\Phi_m}$ 
 but a fixed point $\tZ \in \FP$ is not necessarily  in $\NE$. As an example, consider a game with a unique fully mixed Nash equilibrium. Then $(e_i,e_j,e_i, e_j)$, $i \in \{1,\dots,d_x\}, j \in \{1,\dots,d_y\}$ are fixed points of $\restr{{\Phi_m}}{\prodSimplex}$ (and of ${\Phi_m}$), but are not contained in $\NE$.
 \end{remark}
 To define stability and instability of a set $\cS \subset \cZ$, we denote the distance $d(Z,\cS) := \inf_{Z' \in \cS}\norm{Z-Z'}$. 
  \begin{definition}[Invariant, Stable, Asymptotically Stable and Globally Attracting Sets]
Let $F:\cZ \rightarrow \cZ$ and let $\cC \subseteq \cZ$. 
We say $ \cC$ is a (forward) \textbf{invariant set} if $F(\cC) \subseteq \cC$. A set $\cS \subset \cZ$ is
\begin{enumerate}
  \item (Lyapunov)  \textbf{stable}, if for all $\epsilon > 0$ there exists a $\delta>0$, such that for any $n \in \NN$ and any $Z \in \cZ$ with
  $ d(Z,\cS) < \delta$, we have $d(F^n(Z), \cS) < \epsilon$;
  \item an \textbf{attracting set} if 
  there exists $\delta > 0$ such that for any $Z \in \cZ$ with $  d(Z,\cS)<\delta $, 
  we have $d(F^n(Z),\cS) \xrightarrow[n \to \infty]{} 0$;
  \item \textbf{asymptotically stable} if $ \cS$ is a stable and attracting set. 
  \end{enumerate}
  We call $\cS$ \textbf{unstable} if it is not stable. We call $\cS$ \textbf{globally attracting on $ \cZ$} if $d(F^n(Z),\cS) \xrightarrow[n \to \infty]{} 0$ for any $Z \in \cZ$.
 \end{definition}
In particular, if $\cS$ is a singleton, we call it a \emph{stable/unstable, asymptotically stable, or globally attracting fixed point}. 
Ultimately, we are interested in the attractiveness and stability of the set
corresponding to the Nash equilibria $\NE$.

We note that the definition of Lyapunov stability requires the sequence to stay in $\cZ$. Again, note that this does not necessarily hold for the ambient space $\cD_{\Phi_m}$. Thus, we restrict to $\prodSimplex$ for our global convergence result. 
 \section{Global Convergence for Zero-Sum Games}\label{sec:global}
In this section, we restrict $\restr{\Phi_m}{\prodSimplex}$ to $m=1$, that is $ \restr{\Phi_1}{\prodSimplex}$. Furthermore, we add the standing assumption 
 that the initial $Z_0$ is from the relative interior of $\prodSimplex$.
  We show that for the special case of zero-sum games, the dynamics converge to $\NE$ if  $\etax \etay \,\|A\|_{\oneinf}^2\leq \frac16$.
Our argument consists of two steps:
\begin{enumerate}
\item In Theorem~\ref{thm:convFP}, we use a Lyapunov-style argument to show that $\FPone$ is globally attracting;
 \item We show that, under the assumption that $\FPone$ is finite, every accumulation point of the dynamics is an equilibrium in Theorem~\ref{thm:NEconvergenceZSG}. 
 \end{enumerate}
   \begin{theorem}\label{thm:convFP}
    Consider a zero-sum game $\Gamma(A,-A^\top)$. Assume  $Z_0 \in \relint \prodSimplex $  and     \[
      0<\etax \etay \,\|A\|_{\oneinf}^2 \leq  \frac16\,.
    \]
    Then the set of fixed points $ \FPone$  is globally attracting with respect to $\relint \prodSimplex$.
    
  \end{theorem}
  For a formal proof, see Appendix~\ref{apx:proof_convFP}.
   We note that Theorem~\ref{thm:convFP} does not show convergence to the set of Nash equilibria. Recall the example from Remark~\ref{rem:strictInclusion}.  
 In this example, the pure actions are repelling, but $\FPone$ is globally attracting. 
 Intuitively, this implies that the dynamics will not converge to the repelling fixed points, but only to the Nash equilibria, which we show formally in the following theorem. 
   \begin{theorem}\label{thm:NEconvergenceZSG}
    Consider $\Gamma(A,-A^\top)$ and assume  $Z_0   \in \relint \prodSimplex$ and $\etax, \etay$ satisfy the step size condition of Theorem~\ref{thm:convFP}. Further, assume that $\FPone $ is finite. Let $(Z_t)_{t\in\NN}$ be the sequence of iterates defined by $Z_t = {\Phi_1}(Z_{t-1})$. Then,  $(Z_t)_{t\in\NN}$ converges and the limit $Z_{\infty}$ is a Nash equilibrium; that is, $Z_\infty \in \NE$.
\end{theorem}
For a proof, see Appendix~\ref{apx:proofglobal}. While our results cover a wide range of step sizes, note that even for the toy examples in Figure~\ref{fig:ExamplenRPS}, we observe that the step size bound is not tight. Furthermore, as the simulations in Figure~\ref{fig:ExamplenRPS} (b) show, we observe similar behaviour for $1$-optEW and non-zero-sum games. Theorem~\ref{thm:NEconvergenceZSG} does not apply here since it requires zero-sum games. This motivates the results of Section~\ref{sec:local}.
 \begin{remark} Consider a zero-sum game $\Gamma(A,-A^\top)$ where $A$ has at least one non-zero entry. 
For  step sizes $\etax = \etay $ with $0<\etax\etay \leq1/({8^2 \norm{A}^2_{\ell_2}})$, the results in \citet{Wei:2021aa} imply last iterate convergence of $1$-optEW. Observe that $ 1/({8^2 \norm{A}^2_{\ell_2}}) <   {1}/({6 \norm{A}_{\oneinf}^2})$. However, we emphasize that their focus is on convergence rates.  \end{remark}
 \begin{figure}[t]
\centering

 \begin{minipage}[t]{.4829\linewidth}
\centering
\includegraphics[width=\linewidth]{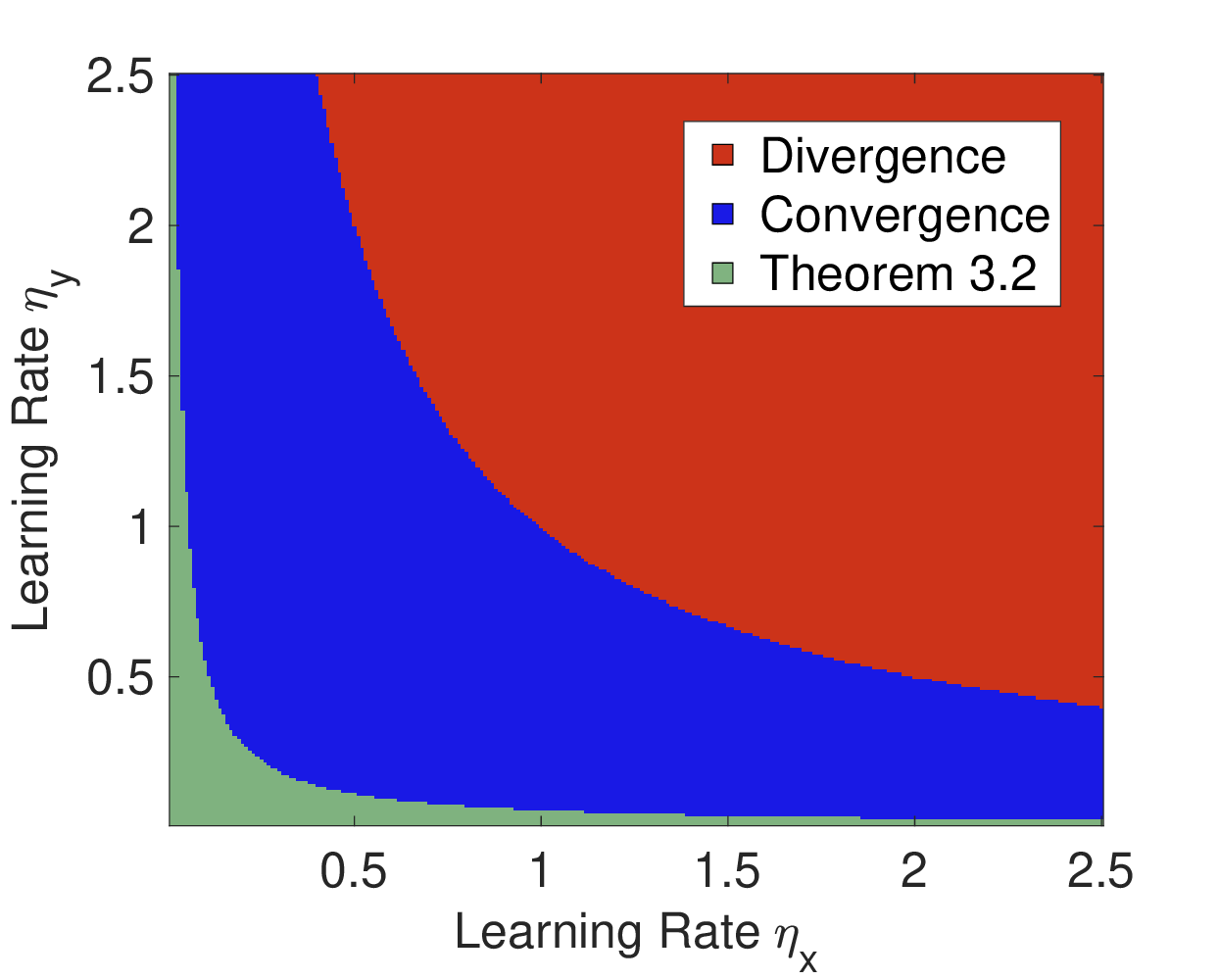}

\smallskip
\footnotesize{\textbf{(a)} Rock-Paper-Scissors}
\end{minipage}
\hfill
\begin{minipage}[t]{.4829\linewidth}
\centering
\includegraphics[width=\linewidth]{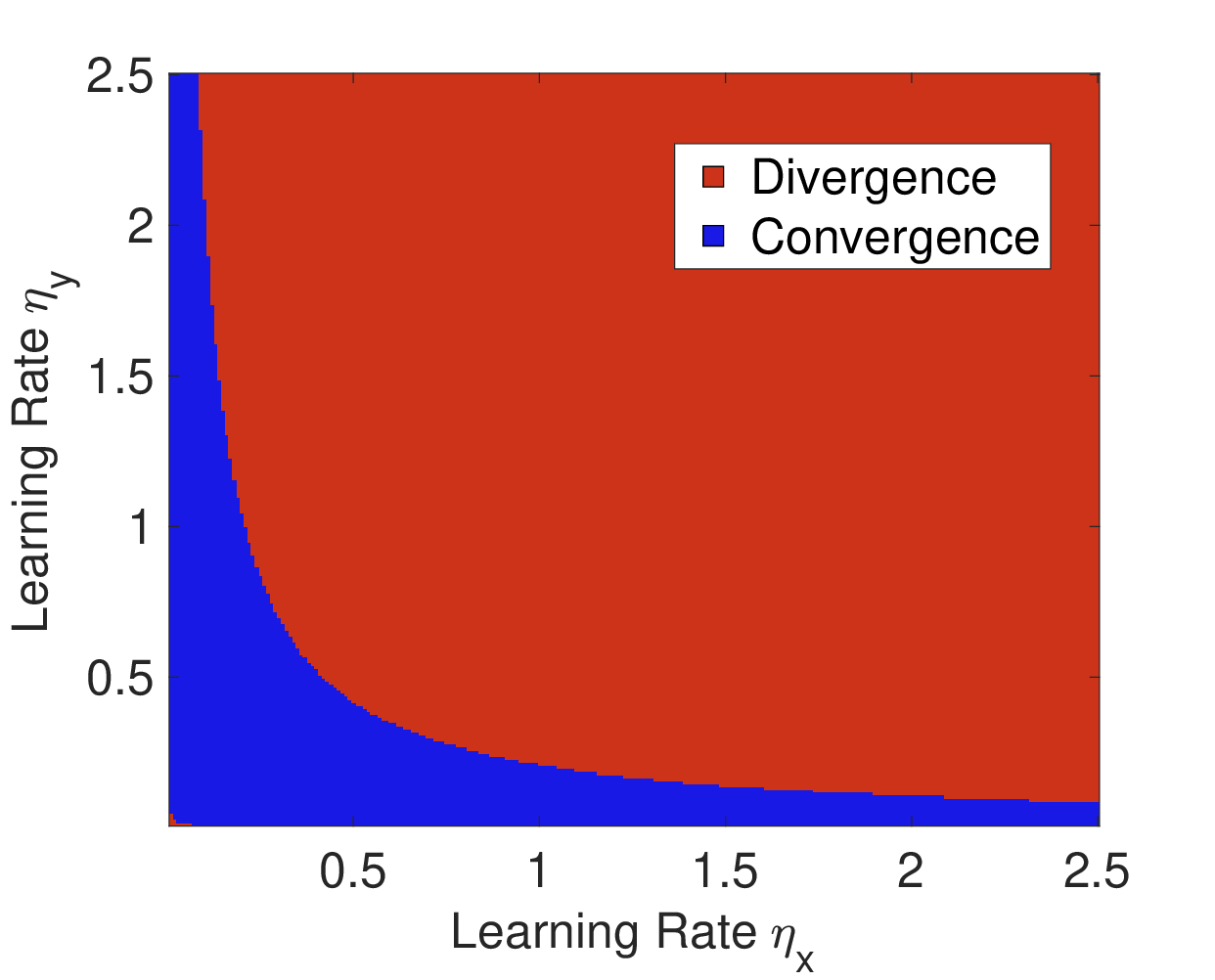}

\smallskip
\footnotesize{\textbf{(b)} Non-Zero-Sum Matching Pennies}
\end{minipage}

\caption{Estimated convergence and divergence of $1$-optEW as a function of $\etax\etay$. The dynamics are labeled as convergent if they reach an $\epsilon$-neighbourhood of the Nash equilibrium within $10\,000$ iterations. The non-convergence for $\etax\etay< 0.01$ is due to the finite-time cut-off. For zero-sum games, the step size range shown in Theorem~\ref{thm:NEconvergenceZSG} is highlighted. For both games, see Appendix~\ref{apx:games} for details.}
\label{fig:ExamplenRPS}
\end{figure}

 \subsection{Why $m = 1$? }
 
 In our result, we focused on vanilla optimistic exponential weights, namely $1$-optEW. 
 The local stability results in Section~\ref{sec:ZSG_local_Stab} show that, for $m\le 1/2$, fully mixed equilibria in the stated zero-sum setting are locally unstable. Hence, no result asserting global asymptotic stability can hold in that regime.
 The case $m \in (\frac{1}{2}, 1)$, however, remains unclear. In Section~\ref{sec:open_questions}, we provide numerical evidence suggesting that $m$-optEW  global convergence may fail for $m \in (\frac{1}{2}, 1)$ with  step sizes that still guarantee asymptotic stability (cf. Section~\ref{sec:local}). Establishing this rigorously remains an interesting open problem.

\section{Local Convergence}\label{sec:local}

We now characterize the asymptotic stability and instability of fixed points and equilibria under the $m$-optEW dynamics as a function of the product of the step sizes.
 Our main local result is that stability can be read off from the Jacobian of the
dynamics at a fixed point, once it is restricted to the hyperplane containing the simplex. We then show that the corresponding spectrum admits an explicit
description in terms of a smaller local matrix $M_x(\tZ)$ (cf.  \eqref{eq:def_Mx}), from which one obtains that,
for equilibria, stability depends only on the product $\etax\etay$.
\paragraph{Notation: }
 For any finite-dimensional real linear map $T:E\to E$ and any real $T$-invariant subspace $V\subseteq E$, the notation
$\eig{T|_V}, \eigTilde{T|_V},$ and $\SR{T|_V}$
denotes the spectrum, the spectrum without multiplicities, and the spectral radius of the complexified restriction
$(T|_V)_\mathbb{C}:V_\mathbb{C}\to V_\mathbb{C}.$
When a real matrix or real linear map is applied to a vector in a complexified space, we use its complex-linear extension and suppress the subscript $\mathbb{C}$ whenever it is clear from the context.
The modulus of a complex number $c \in \CC$ is denoted by $\absv{c}$. 

\subsection{Stability and the Jacobian}
 Following a standard approach, we analyze the eigenvalues of the differential of the dynamics map $\Phi_m$ at the fixed point. There are, however, some
technical hurdles we need to handle
\begin{itemize}
  \item We are interested in ${\Phi_m}|_{\prodSimplex}$ and $\prodSimplex$ has an empty interior, so
        the computation of its differential is non-trivial. Recall, we defined ${\Phi_m}$
        as a function $\cD_{\Phi_m} \to \RR^d $, where $\cD_{\Phi_m}$ is open, so the
        computation of its differential is straightforward via standard rules. 
        However, the mapping of interest is the restriction of ${\Phi_m}$ to $\prodSimplex$.
   
  \item Equilibria may occur on the boundary of $\prodSimplex$. There, the differential is not
        \emph{a priori} enough to characterize stability.
\end{itemize}
 Theorem~\ref{thm:stab_from_jac} is a key technical result to handle these technical challenges.
Define the linear part of the hyperplane containing $\prodSimplex$ as, 
\[
  L =  \{ [h_1, h_2, h_3, h_4] \in \RR^{d_x} \times \RR^{d_y}\times \RR^{d_x} \times \RR^{d_y}
    \mid h_k^\top \mathbf 1 = 0  \text{ for all $k \in\{1,2,3,4\}$} 
  \}
  \, .
\]
Recall that to discuss the eigenvalues of the restricted Jacobian
$J_{\Phi_m}(\tZ)|_L$, one first has to show that $L$ is invariant under $J_{\Phi_m}(\tZ)$.

\begin{theorem}\label{thm:stab_from_jac}
  Let $\tZ \in \prodSimplex $ denote a fixed point of ${\Phi_m}$. That is, $\tZ \in \FP  $.
  Then $J_{\Phi_m}(\tZ) L \subset L$. 
  Moreover, $ \tZ$ is an asymptotically stable fixed point for $\restr{\Phi_m}{\prodSimplex}$ if the spectral radius satisfies 
  $\SR{\restr{J_{\Phi_m}(\tZ)}{L}} < 1$, and $ \tZ$ is unstable if  $\SR{\restr{J_{\Phi_m}(\tZ)}{L}} > 1$.
\end{theorem}

The proof is deferred to Appendix~\ref{app:jac_stab}. It combines four main
ingredients. First, one computes the Jacobian of ${\Phi_m}$ at a fixed point. Second, one
shows that the relevant local dynamics is obtained by restricting ${\Phi_m}$ to the affine
space $\tZ+L$, and that the corresponding linearization is given by the
restricted Jacobian $J_{\Phi_m}(\tZ)|_L$. Third, one applies standard local dynamical systems
arguments to this restricted dynamics. Finally, because fixed points may lie on the
boundary of $\prodSimplex$, the instability statement requires an additional argument adapted
to the constrained setting.

The borderline case $\tSR{ J_{\Phi_m}(\tZ)|_L}=1$ is not covered by this criterion:
The linearization alone does not suffice in general to determine stability.

\subsection{The Spectrum of the Jacobian}

\paragraph{Notation and definitions}
Let $\tZ$ be a fixed point of ${\Phi_m}$. By Theorem~\ref{thm:FPPhi}, $\tZ$ must have the form $\tZ = \Zfp$.
We provide a characterization of the stability of  $\tZ$ in terms of a local
matrix $M_x(\tZ)$. We denote by
  $S_x := \{i \in \{1,\dots,d_x\} \mid \tx_i>0\}$ and
  $S_y := \{j \in \{1,\dots,d_y\} \mid \ty_j>0\}$ the supports of the two strategies, with
  cardinalities $n_x:=|S_x|$ and $n_y:=|S_y|$.
Since $\tZ$ is a fixed point, the payoffs are constant on the supports, so we may define
  $\vx := (A\ty)_i$ for any $i\in S_x$ and $\vy := (B\tx)_j$ for any $j\in S_y$.
We also write
  $L_x := \{h \in \RR^{n_x} \mid h^\top \mathbf 1 = 0\}$  
and denote by $\Pi_{S_x}\in\RR^{n_x\times d_x}$ the
coordinate-selection matrix associated with the supports.
For $i\in\{1,\dots,d_x\}$ and $j\in\{1,\dots,d_y\}$, let  $\V_{1, i} := \exp\!\left( \etax \big( (A \ty)_i - \vx  \big) \right)$ and
$\V_{2, j} := \exp\!\left( \etay \big( (B \tx)_j - \vy  \big) \right)$,
and collect the off-support terms in
\begin{equation}\label{eq:def_W}
  W = \{ \V_{1, i} \mid i \in \{1, \dots, d_x\} \setminus S_x \} 
  \cup 
  \{ \V_{2, j } \mid j \in \{1, \dots, d_y\} \setminus S_y \} \,.
\end{equation}
 For vector $v \in \RR^d$, write $H(v) := \diag(v) - v v^\top$.
The reduced matrix governing the nontrivial part of the spectrum is then
\begin{equation}
  \label{eq:def_Mx}
  M_x(\tZ) = \Pi_{S_x} H(\tx) A H(\ty) B \Pi_{S_x}^\top \in \mathbb{R}^{n_x\times n_x}
  \,.
\end{equation}
Finally, for any $\lambda \in \mathbb C \setminus \{ m / (m+1)\}$, define the rational
map
\[
  Q_m(\lambda) = \frac{\lambda^2(\lambda - 1)^2}
  {((m+1)\lambda - m)^2}
  \,.
\]

\paragraph{Main result}
We are now ready to state the main theorem of this section, which characterizes the
spectrum of $J_{{\Phi_m}}(\tZ)|_L$ via $M_x(\tZ)$. We note that the following theorem does not account for multiplicities of
eigenvalues; however, for the stability analysis, multiplicities are irrelevant. 
\begin{theorem}
  \label{thm:jac_spectrum} Let $\tZ \in \prodSimplex$ be a fixed point of ${\Phi_m}$ and assume $n_x \geq n_y$.
   The spectrum of the Jacobian without multiplicities is
  \[
    \eigTilde{J_{\Phi_m}(\tZ)|_{L}} \setminus \{0\}
    = \left(W \cup Q_m^{-1} \Big( \eigTilde{\etax \etay M_x(\tZ)|_{L_x}} \Big)\right) \setminus \{0\}\, .
  \]
   \end{theorem}
   The proof consists of carefully following the eigenvalues and eigenvectors, and 
is detailed in Appendix~\ref{app:jac_spectrum}.
We exclude the eigenvalue $\lambda=0$ for technical reasons (the reduction to the
equation involving $Q_m$ requires dividing by $\lambda$). The value $0$ plays no role
in the stability analysis, since it is strictly less than one.

 An important feature of Theorem~\ref{thm:jac_spectrum} is that it separates the
spectrum into two parts: the off-support eigenvalues are collected in $W$, and the
support-restricted eigenvalues are governed by
$Q_m^{-1}\!\left(\eig{\etax\etay \restr{M_x(\tZ)}{L_x}}\right).$
The off-support eigenvalues may depend on $(\etax,\etay)$ separately, but they
do not affect the stability analysis of Nash equilibria. For example, if the fixed
point $\tZ$ corresponds to a quasi-strict 
Nash equilibrium\footnote{A  Nash equilibrium is quasi-strict if all best response pure actions have positive support.}, 
then every eigenvalue in $W$ has
modulus strictly smaller than $1$; whereas if $\tZ$ does not correspond to a Nash
equilibrium, then at least one eigenvalue in $W$ has modulus strictly larger
than $1$. Thus, the off-support spectrum distinguishes Nash equilibria from
non-equilibrium fixed points. Conditional on $\tZ$ corresponding to a quasi-strict Nash
equilibrium, all off-support eigenvalues are already stable, and the remaining
stability information is determined solely by the product $\etax\etay$ through
the spectrum of $\etax\etay M_x(\tZ)|_{L_x}$.  Building on this property, we derive results for special cases in Section~\ref{sec:consequences}.  

\begin{remark}\label{rem:yIdentity}
The assumption that the support of the $x$-player is greater than or equal to the $y$-player's support can be eliminated by noting that a similar result holds when $n_y \geq n_x$. Define $\Pi_{S_y}$ , $L_y$ and $M_y(\tZ)$ analogously. Then \[\eigTilde{\restr{J_{\Phi_m}(\tZ)}{L}} \setminus \{ 0\}= \left(W \cup Q_m^{-1} \Big( \eigTilde{\etax \etay \restr{M_y(\tZ)}{L_y}} \Big)\right) \setminus \{0\}\,.\]   
 Since the matrices $\restr{M_y(\tZ)}{L_y}$ and  $\restr{M_x(\tZ)}{L_x}$ share their non-zero complex spectrum,  
\[\eigTilde{\restr{J_{\Phi_m}(\tZ)}{L}}\setminus\{ 0\} = \left(W \cup Q_m^{-1} \Big( \eigTilde{\etax \etay \restr{M_y(\tZ)}{L_y}} \Big) \cup Q_m^{-1} \Big( \eigTilde{\etax \etay \restr{M_x(\tZ)}{L_x}} \Big)\right) \setminus \{0\}\, .\]   
This identity holds without any assumptions on the support size.
Note, however, that for many special cases, for example, unique pure or unique fully mixed Nash equilibria, we have $n_x = n_y$. Thus, the characterization with respect to either $\restr{M_x(\tZ)}{L_x}$ or $\restr{M_y(\tZ)}{L_y}$ is sufficient. 
\end{remark}
 
\subsection{Interpretation and Illustrations}

 Let
\[
  \Omega_m = 
  \Big\{
    \mu \in \CC \; | \;
    \forall \, \lambda  \in \CC \quad\text{ with }\quad
    Q_m(\lambda) = \mu 
    \Rightarrow 
    |\lambda| < 1
  \Big\}
\]
A consequence of Theorem~\ref{thm:jac_spectrum} is that a quasi-strict Nash equilibrium is an asymptotically stable fixed point of $m$-optEW if all eigenvalues 
of $\etax \etay M_x(\Zs)$ restricted to $L_x$ lie in $\Omega_m$, that is $\eig{\restr{\etax \etay M_x(\Zs)}{L_x}}\subseteq \Omega_m$. Note that $\Omega_m = \varnothing$ for $m \in (0,\frac{1}{2}]$, which provides the intuition behind Corollary~\ref{cor:instability}.
See Figure~\ref{fig:qm_regions} for an illustration of $\Omega_m$.
 \begin{figure}[h]
  \center
 \includegraphics[width=0.9\linewidth]
 {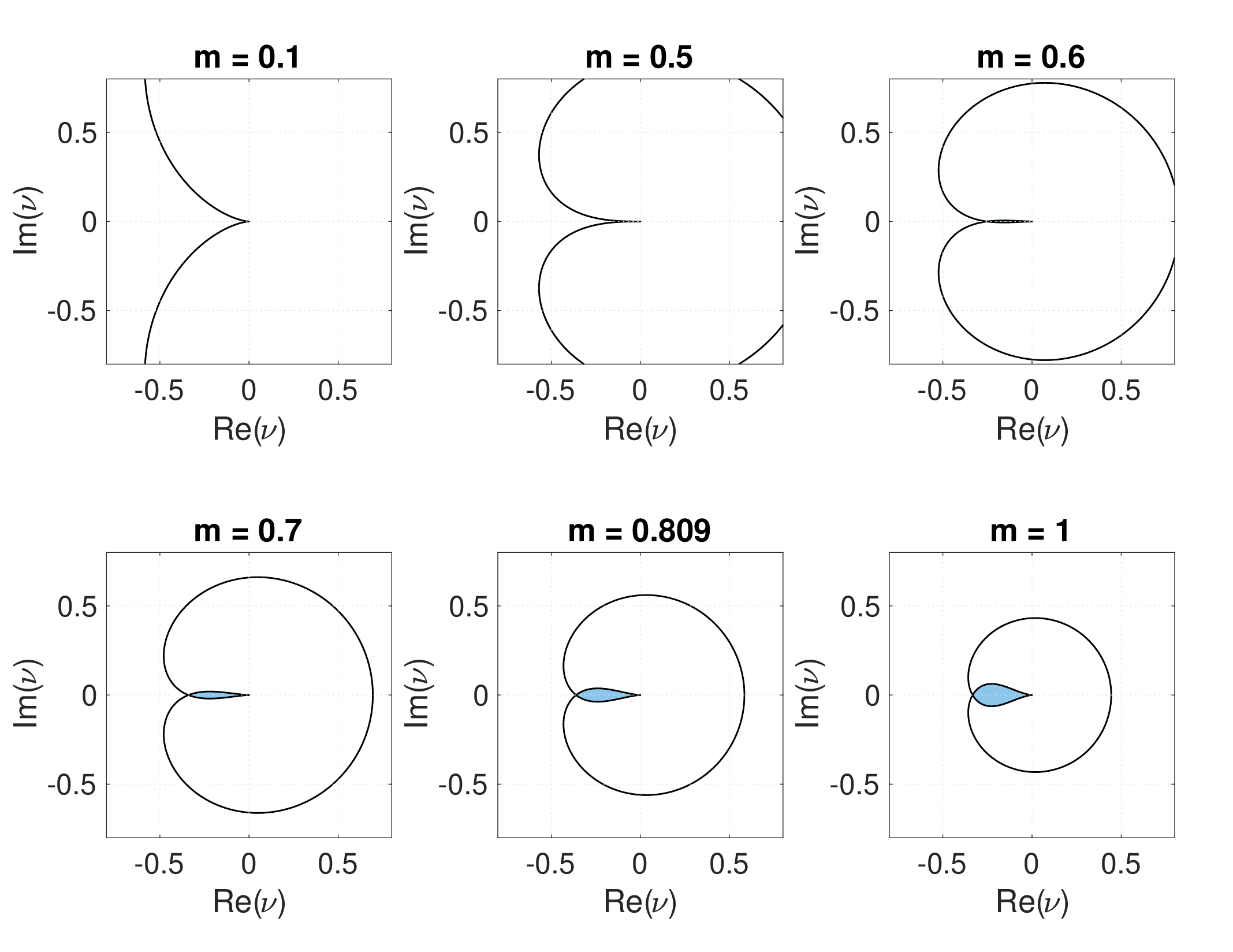} 
  \caption{
  The blue region indicates all solutions of $ Q_m(\lambda) = z$ that have modulus less than one. }\label{fig:qm_regions}
\end{figure}

\subsection{Consequences}\label{sec:consequences}
The following results are direct consequences of Theorem~\ref{thm:jac_spectrum} and Theorem~\ref{thm:stab_from_jac}. The first corollary relies on the observation that $\Omega_m = \varnothing$ for $m \in (0,1/2]$. For a formal proof, see Appendix~\ref{apx:corMsmall}.
\begin{corollary}\label{cor:instability}
Consider the game $\Gamma(A,B)$ with $d_x\geq d_y \geq 2$ and assume that there exists a fully mixed Nash equilibrium $Z^\star \in \NE$. 
Suppose $m \in (0,1/2]$.  If $\restr{M_x(Z^\star)}{L_x}$ is non-singular, then $Z^\star$ is an unstable fixed point for the $\Phi_m$ dynamics for any step size choices $\etax, \etay  >0$.
\end{corollary}
 The next corollary follows from the definition of the set $W$. We noted before that for an equilibrium, all values in $W$ are strictly less than one. However, this does not hold for fixed points that are not equilibria. 
 \begin{corollary}
 Suppose $\tZ \in \FP \setminus \NE$, that is, $\tZ$ is a fixed point but not a Nash equilibrium. Then $\tZ$ is unstable under the $\Phi_m$ dynamics for any $\etax, \etay >0$.
\end{corollary}
\begin{proof}
  By Theorem~\ref{thm:jac_spectrum}, the nonzero eigenvalues of $J_{\Phi_m}(\tZ)|_L$
  contain the set $W$. If $[x,y]$ is not a Nash equilibrium, then by the
  characterization of $W$ above, at least one element of $W$ has modulus strictly
  larger than $1$. Hence $\SR{J_{\Phi_m}(\tZ)|_L}>1$, and Theorem~\ref{thm:stab_from_jac} gives the claim.
\end{proof}

Recall that an equilibrium is strict if it is pure (that is, $S_x$ and $S_y$ 
are singletons) and $(A y^\star)_i < (x^\star)^\top A y^\star$, and 
$ (B x^\star)_j < (y^\star)^\top B x^\star$ for all $i \notin S_x$ and $j \notin S_y$. The following corollary is a well-known result (see, e.g., \citet{Mertikopoulos:2016aa}, \citet{Giannou:2021aa}).
\begin{corollary}
  If $Z^\star \in \NE$ corresponds to a strict Nash equilibrium, then it is an asymptotically stable fixed point for $\Phi_m$ for any $m \in (0,1]$ and $\etax, \etay >0$.
\end{corollary}
\begin{proof}
  Since $ Z^\star$ is a strict Nash equilibrium, it is pure, so
  $S_x$ and $S_y$ are singletons, hence $n_x = n_y = 1$. Hence $L_x=\{0\}$ and therefore
  $\eig{\restr{M_x(Z^\star)}{L_x}}=\varnothing$.
  Moreover, for every $i \notin S_x$ and $j\notin S_y$, strictness implies elements of $W$ are strictly smaller than $1$.
  By Theorem~\ref{thm:jac_spectrum}, every nonzero eigenvalue of
  $J_{\Phi_m}(Z^\star)|_L$ therefore has modulus strictly smaller than $1$, and thus
  $\SR{J_{\Phi_m}(Z^\star)|_L}<1$. Theorem~\ref{thm:stab_from_jac} implies that
  $ Z^\star$ is asymptotically stable.
\end{proof}

\begin{corollary}\label{cor:unstable_positive_mu}
 Let $\tZ \in \FP$. 
  If there exists $\mu \in \teig{\restr{M_x(\tZ)}{L_x}} \cup \teig{\restr{M_y(\tZ)}{L_y}} $ such that $\mu \in \RR$ and $\mu > 0$, then 
  $ \tZ $ is  unstable for any value of $m$ and any choices of $\etax,\etay \in (0,\infty)$.
\end{corollary}
\begin{proof}
  Due to Theorem~\ref{thm:jac_spectrum} and Remark~\ref{rem:yIdentity}, it suffices to show that for any $\etax ,\etay \in (0,+\infty)$ there exists $\lambda \in \RR$ with $\absv{\lambda }> 1$ such that 
  $Q_m(\lambda)= \etax \etay\mu$. 
 This holds since $Q_m(1) = 0$ and $Q_m(\lambda) \to + \infty$ as $\lambda \to +\infty$
  and $Q_m$ is continuous on $[1, +\infty)$.
\end{proof}
 An example of such an unstable equilibrium is the fully mixed Nash equilibrium $x^\star = y ^\star = [\frac12,\frac12]$ in the coordination game $\Gamma(A,B)$ where \[ A = B = \begin{bmatrix} 1&0\\0&1 \end{bmatrix}\, .\]
 In this case $\frac14 \in \teig{\restr{M_y(\Zs)}{L_y}} $ which implies that $\Zs = \Zne$ is unstable for any value of $m$ and choices of $\etax,\etay \in (0,\infty)$.
 
\subsubsection{Zero-Sum Games}\label{sec:ZSG_local_Stab}
For the class of zero-sum games, we derive the following result from Theorem~\ref{thm:stab_from_jac} and Theorem~\ref{thm:jac_spectrum}.
 \begin{theorem}\label{thm:stabilityZSG}
   Consider a zero-sum game with a unique fully mixed Nash equilibrium $ Z^\star \in \NE = \{ Z^\star\}$.
  Then $ Z^\star$ is asymptotically stable for  $\restr{\Phi_m}{\prodSimplex}$ if
  \[
     0 < \etax \etay \SR{\restr{M_x(Z^\star)}{L_x}}<  \frac{2m -1}{m^2(2 m+ 1)}  \,, 
  \]
  and unstable if 
  \[
     \etax \etay \SR{\restr{M_x(Z^\star)}{L_x}} >  \frac{2m -1}{m^2(2 m+ 1)}  \,.
  \]
  In particular, if $m \leq 1/ 2$ and $d_x, d_y \geq 2$, then $Z^\star$ is an unstable fixed point of $\Phi_m$.
\end{theorem}
For a proof, see Appendix~\ref{apx:consequences_ZSG}.  We note that the result is primarily of theoretical interest since the step size bounds depend on $Z^\star = [x^\star,y^\star,x^\star,y^\star]$. However, we observe that $H(x^\star)$ and $H(y^\star)$ are covariance matrices and therefore $\norm{H(x^\star)}_{{\ell_2}}$ and $\norm{H(y^\star)}_{{\ell_2}}$ are upper bounded by $\frac12$. Hence
\[ \SR{\restr{M_x(Z^\star)}{L_x}} \leq \norm{H( x^\star)}_{\ell_2} \norm{H(y^\star)}_{\ell_2} \norm{A}_{\ell_2}^2 \leq \frac{1}{4}\norm{A}_{\ell_2}^2\, , \]
 which gives us the simple and potentially loose criterion for the stability of a Nash equilibrium:
 \[ \etax \etay   \norm{A}_{\ell_2}^2 <  \frac{4(2m -1)}{m^2(2 m+ 1) }\, . \]
 Comparing this result with the global bound in Section~\ref{sec:global}, we observe that this bound allows for larger $\etax \etay$ and for $m \in (\frac{1}{2}, 1]$ while guaranteeing asymptotic stability. However, Theorem~\ref{thm:stabilityZSG} does not guarantee global convergence, as the results in Section~\ref{sec:global} do.
  
\subsubsection{The class of $2\times2$ games}
The class of games with $2$ players and $2$ actions is also amenable to a detailed 
study. Consider
\[
  A=
  \begin{bmatrix}
    a & b\\
    c & d
  \end{bmatrix},
  \qquad
  B=
  \begin{bmatrix}
    e & f\\
    g & h
  \end{bmatrix}\, .
\]
 The Nash equilibria are of the form $ x^\star=(p,1-p)^\top,\quad y^\star=(q,1-q)^\top$ where $p,q \in [0,1]$. 
Further, $x^\star$ is fully mixed if $(e-g)(h-f) > 0$, and $y^\star$ is fully mixed if $(d-b)(a-c) > 0$.
 The complements/substitutes split is classical in game theory; see, e.g.,  \citet{Rapoport:1966aa}. Learning dynamics in $2\times 2$ games, including Experience Weighted Attraction and replicator limits, are studied, for instance, by
 \citet{Pangallo:2022aa}.
 For this class of games, the following result follows from Theorem~\ref{thm:stab_from_jac} and Theorem~\ref{thm:jac_spectrum}.
\begin{theorem}\label{thm:2x2games}
    Suppose $p,q \in (0,1)$. Define
\[
  \Delta_A:=a-b-c+d\,,
  \qquad
  \Delta_B:=e-f-g+h\, ,
\]
and threshold \[ E :=  \frac{1}{p(1-p)q(1-q)|\Delta_A \Delta_B|} 
     \frac{2m-1}{m^2(2m+1)} \, .\]
    Then $[x^\star, y^\star]$ is a locally asymptotically stable fixed point for ${\Phi_m}$ if  
    \[
   \Delta_A \Delta_B < 0 \qquad \text{ and } \qquad 0 < \etax \etay < 
      E,
    \] 
    and an unstable fixed point if
    \[
     \Delta_A \Delta_B < 0 \qquad \text{ and } \qquad  \etax \etay > E\, .
    \] 
     In particular,  $[x^\star, y^\star]$ is unstable
     for every 
    $\etax, \etay$ if $\Delta_A \Delta_B > 0$.
\end{theorem}
For a proof, see Appendix~\ref{apx:theorem2x2}.

 \subsubsection{Low Dimensional Games with Unique Fully Mixed Equilibria}
 Beyond two-action games, games with $d_x=d_y\leq 5$ and a unique fully mixed Nash equilibrium admit a closed-form expression for the spectrum, obtained by radicals from rational functions of the entries of the game matrices $A$ and $B$.
 \begin{theorem}\label{thm:lowDim}
Let $\Gamma(A,B)$ be the game, and assume the Nash equilibrium is unique and fully mixed and that the dimensions of the game matrices satisfy $d_x=d_y \leq 5$. 
Then the non-zero spectrum of  $\restr{J_{{\Phi_m}}(Z^\star)}{L}$ can be expressed explicitly in terms of the entries of $A$ and $B$.
 \end{theorem}
A proof can be found in Appendix~\ref{apx:lowDim}. The proof relies on a closed-form expression for the Nash equilibrium, which holds for any dimension but requires uniqueness and full support.  A closed-form expression for the spectrum follows from the roots of the characteristic polynomial of $\restr{M_x(\Zs)}{L_x}$. Since general polynomials of degree five or higher are not solvable by radicals, it suffices that this characteristic polynomial have degree at most four, which is ensured by our assumptions on the Nash equilibrium and the dimensions.

 \subsection{Remarks on related results}\label{sec:CompareResultsJacobian}
 In this section, we compare our results in-depth to two results from the literature. 
 \paragraph{Optimistic Gradient Method:}
 Consider a game $\Gamma(A,B)$ with $A, B^\top \in \RR^{d_x \times d_y}$, played over the unconstrained strategy
spaces $\RR^{d_x}$ and $\RR^{d_y}$; that is, the $x$-player maximizes
$x^\top A y$ over $\RR^{d_x}$ and the $y$-player maximizes $y^\top B x$
over $\RR^{d_y}$. The optimistic gradient method (OGM) is
 \begin{align*}
 \begin{cases}
 x_{t+1} = x_t + \etax A\left(2y_t -  y_{t-1}\right) \\
 y_{t+1} = y_t +  \etay B\left(   2x_t -   x_{t-1}\right)\, .
 \end{cases}
 \end{align*}
Note that this algorithm applies only to \emph{unconstrained} bimatrix
games. This is the setting of
\citet{de2025optimistic}, whose general-sum bilinear games
$(x^\top \tilde A y, x^\top \tilde B y)$ with
$\tilde A, \tilde B \in \RR^{d_x \times d_y}$ correspond to
$\Gamma(A,B)$ via $\tilde A = A$ and $\tilde B = B^\top$. In particular,
the matrices $\tilde B^\top \tilde A$ and $\tilde A \tilde B^\top$
appearing in their analysis are $BA$ and $AB$ in our notation. From
their analysis, the following result follows as a corollary.
\begin{corollary}\label{cor:deMontbrun}
Consider the game $\Gamma(A,B)$ with $d_x = d_y$ and assume
$\teig{BA}$ are real and $\teig{BA} \subset (-\infty , 0) $. Then
$[x^\star,y^\star] = [0,0]$ is the unique Nash equilibrium of the
unconstrained game. Consider OGM with asymmetric step sizes
$\etax, \etay>0$. The corresponding fixed point
  is an asymptotically stable
fixed point of the OGM dynamics if
\[0< \etax\etay \SR{-BA}< \frac{1}{3}\, ,\]
and unstable if
\[ \etax \etay \SR{-BA} > \frac{1}{3 }\, .\]
\end{corollary}
The result follows from a small modification of Proposition 3.7 in
\citet{de2025optimistic}. For the convenience of the reader, we add a
note on the necessary modification in Appendix~\ref{apx:deMontbrun}.
 Analogously, the same modification
applied to Theorem 3.8 in \citet{de2025optimistic} yields global
exponential convergence of the OGM iterates to a Nash equilibrium under
the assumption that $\etax\etay \SR{-BA}< \frac{1}{4}$, provided
additionally that $\teig{BA} \cup \teig{AB} \subset (-\infty, 0]$ and
that either $A$ and $B$ are square and invertible or
$\Lambda_{A,B}$ (defined in Appendix~\ref{apx:deMontbrun}) is
diagonalizable.

\paragraph{$1$-optEW for Convex-Concave Games:}  
   \citet{lei2021last} establish local last-iterate convergence of  1-optEW for general constrained convex-concave min-max problems by proving that the 1-optEW Jacobian is Schur stable for sufficiently small equal step size. This is a standard, well-established analysis approach where the two analyses share some similarities. A central technical obstacle in the setting of \citet{lei2021last} is that the payoff is convex-concave rather than bilinear. The Jacobian contains additional Hessian blocks and is no longer reducible to a skew-symmetric structure with purely imaginary eigenvalues. One of their key technical contributions relies on Ky Fan’s inequality, which relates the real parts of the eigenvalues of a matrix to the eigenvalues of its symmetrized part.  When specialized to bimatrix zero-sum games, the following result aligns with their computations.
  \begin{corollary}
  Consider a zero-sum game $\Gamma(A,-A^\top)$ with a fully mixed Nash equilibrium $Z^\star$ and a non-singular game matrix $A \in \RR^{d\times d}$.
  Then $ Z^\star$ is asymptotically stable for 1-optEW if
  \[
     0 < \etax \etay \SR{M_x(Z^\star)|_{L_x}}<  \frac{1}{ 3}  \,. 
  \] 
  \end{corollary}
  We restrict to fully mixed and non-singular game matrices for two technical reasons. First, \citet{lei2021last} requires the Hessian restricted to the equilibrium support to be invertible; in bilinear zero-sum games, this holds if and only if the two supports have the same size and the corresponding game matrix is nonsingular, excluding standard examples such as rock-paper-scissors and matching pennies.
  Second, full support avoids some of the boundary issues. 
    \citet{lei2021last} write the update as a map on a product of simplices and analyze the Jacobian at the equilibrium. This Jacobian is well defined for any point in the interior of the constraint space; however, for the product of simplices, this is empty.   One can interpret their Jacobian computation as coming from the smooth ambient formula for the exponential-weights update, locally extended around the simplex wherever the normalization denominators remain nonzero, or equivalently as a derivative along the affine/tangent directions of the constraint set (cf. Lemma~\ref{lem:jac_stable_subspace}). Our analysis specialized to $m=1$ and zero-sum games makes this rigorous for the linear case.

 \section{Experiments and Open Questions}\label{sec:experiments_open_questions}
In this section, we present empirical studies illustrating our theoretical results. In all examples, we opt for the simplest possible model for illustration. Namely, many examples use low-dimensional toy models, specifically games with two or three actions, because they allow details to be visualized without dimensionality reduction. Even in $2\times2$ games, the experiments support and illustrate our theory while raising interesting questions for future work. 

\subsection{Landscape $2\times 2$ Non-Zero-Sum Games}\label{sec:experiments2d}
The product-dependent stability criterion guarantees the existence of a neighbourhood around the equilibrium such that the dynamics converge whenever initialized within this neighborhood. However, it does not provide quantitative guarantees. For instance, it does not identify an open ball of a prescribed radius in which convergence is assured, nor does it yield an explicit convergence rate.
 In this section, we consider a non-zero-sum variant of matching pennies $\Gamma(A,B)$.  
\[ A = \begin{bmatrix} -1 & \phantom{-}1\\ \phantom{-}3&-1 \end{bmatrix} \quad \text{and}\quad B = \begin{bmatrix} \phantom{-}2 &- 1\\ -1& \phantom{-}1 \end{bmatrix}\, . \]
This game has a unique fully mixed Nash equilibrium; see Appendix~\ref{apx:nzspm} for details.
 We use the results from Theorem~\ref{thm:2x2games} for step size computation. That is, for a small $\epsilon>0$, we set \[\eta^\star(m) =  (1 - \epsilon)  \frac{1}{p(1-p)q(1-q)|\Delta_A \Delta_B|}   \frac{2m-1}{m^2(2m+1)}   \,. \]
 The step sizes are defined as a function of $m$ and the ratio $c$. That is  \begin{align} \label{eq:etas}\etax(m,c) =  \sqrt{\frac{\eta^\star(m)}{c}}\qquad \text{ and }
\qquad\etay(m,c) =  \sqrt{ c \eta^\star(m)}\, .\end{align}
 Using a $2\times 2$ game allows us to project the simplex to the interval $[0,1]$ using the variable transform $p' \triangleq [p,1-p]^\top$. We use this in the plots where the vertical axis corresponds to the $2$-dimensional simplex and the horizontal axis to a varying parameter. 

\begin{figure}[t]
\centering

 \includegraphics[width=\linewidth]{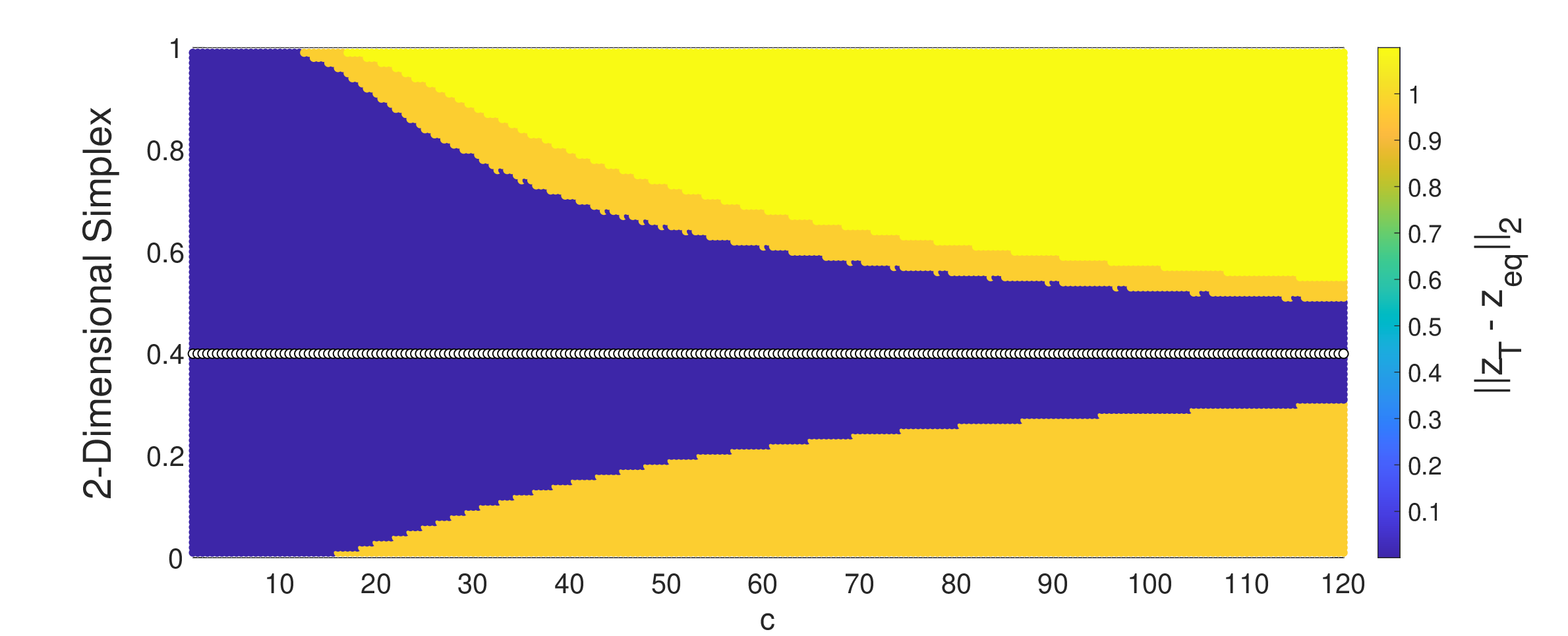}

\caption{Distance to the Nash equilibrium for the non-zero-sum matching pennies variant. Parameter $m$ is fixed to $m^\star$, the ratio between the step sizes varies depending on $c \in [1,120]$.}
 \label{fig:2dSimplex1}
\end{figure}
 \begin{figure}[t]
\centering

 \includegraphics[width=\linewidth]{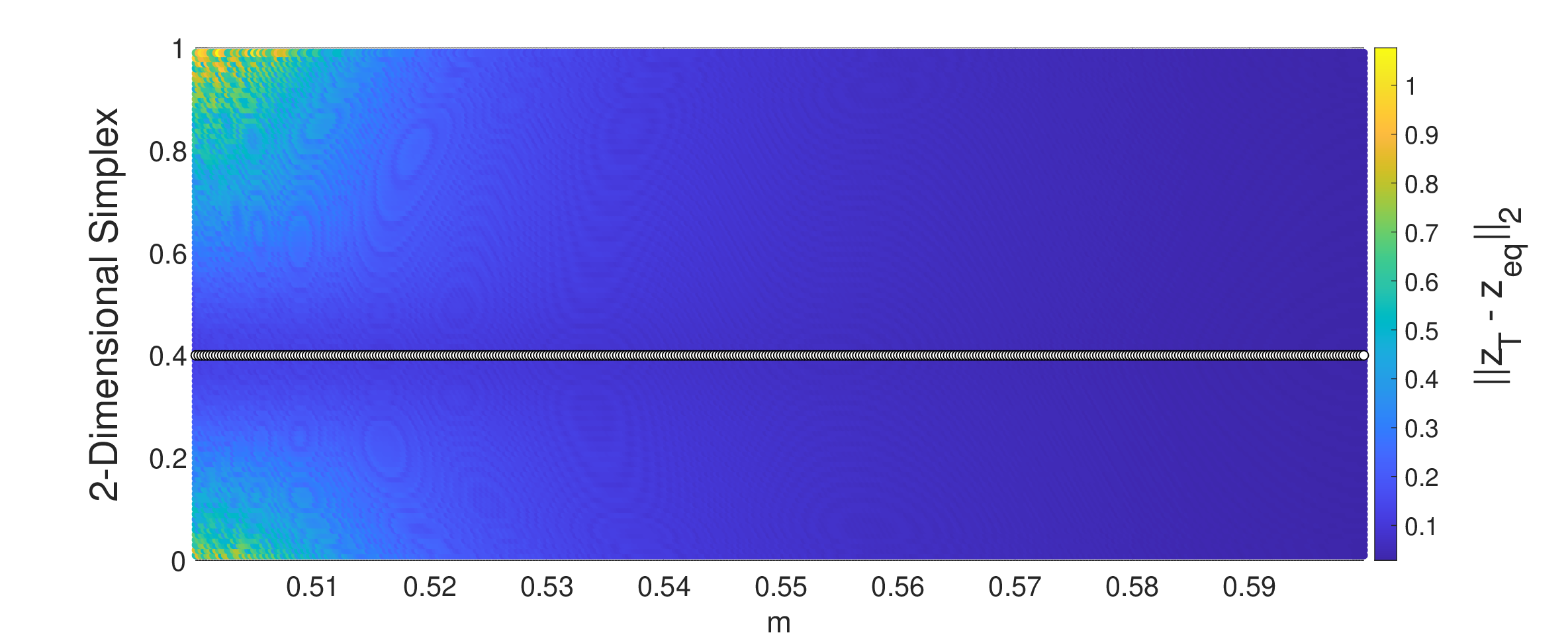}

\caption{Distance to the Nash equilibrium for the non-zero-sum matching pennies variant. The ratio between the step sizes $c$ is fixed to $1$, parameter $m$ varies.}
\label{fig:varyM_c1}
\end{figure}

\subsubsection{Dependence on the Step Size Ratio} 
In this section, we illustrate the dependence on the ratio $c$. 
 We define $m^\star = \frac{1+\sqrt{5}}{4} $. Note that  $m^\star = \argmax_{m \in (0,1]} \frac{2m - 1}{m^2(2m+1)}$ and $2 m^\star $ is the golden ratio. 
  We initialize $m^\star$-optEW with $Z^i_0 = [x^i_0, y_0,x^i_0, y_0]$ where $y_0 = [1/4,3/4]^\top$ is fixed and $x_0^i \in \Delta_2$ from an equidistant grid over the simplex. 
  
Figure~\ref{fig:2dSimplex1} shows the norm distance to the Nash equilibrium after $T = 10\,000$ iterations of $m^\star$-optEW with step sizes $\etax(m^\star,c)$ and $\etay(m^\star,c)$ for varying ratio $c$ on the horizontal axis. Specifically, for each $x_0^i$, the color indicates $  \norm{Z^i_T - Z^\star}$.
The Nash equilibrium $x^\star$ is marked in white. As can be seen in Figure~\ref{fig:2dSimplex1}, the unique Nash equilibrium appears to be asymptotically stable, hence empirically verifying our theoretical results. With a finite cut-off time $T$, we further observe that the apparent region of convergence depends on $c$ and becomes smaller as $c$ increases in this experiment. It remains unclear, however, whether this behavior is merely a consequence of the finite cut-off time $T$, together with slower convergence for larger values of $c$, or whether it reflects an actual shrinking of the attracting neighbourhood.

 \begin{figure}[h]
 \centering
\includegraphics[width=\linewidth]{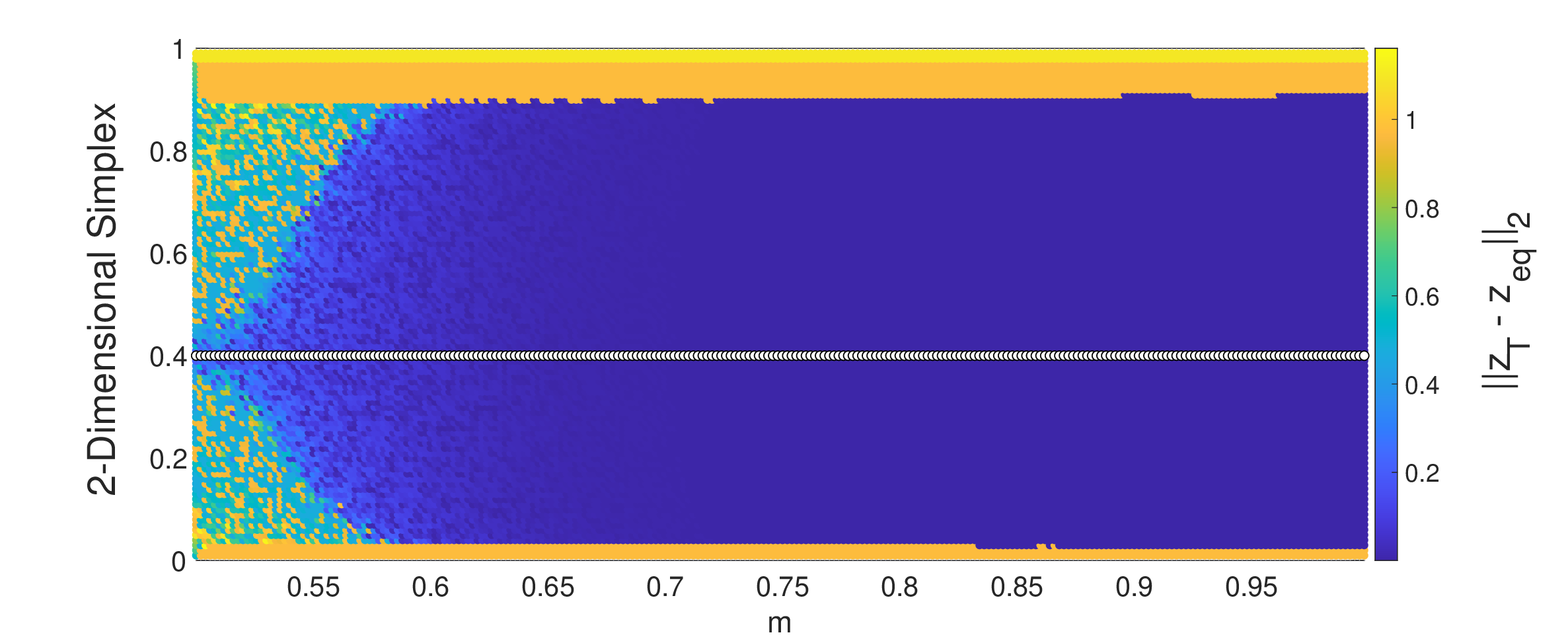}

\caption{Distance to the Nash equilibrium for the non-zero-sum matching pennies variant. The ratio between the step sizes $c$ is fixed to $20$, parameter $m$ varies.}
\label{fig:varyM_c20}
\end{figure}

\subsubsection{Dependence on $m$}
In the previous section, we fixed $m^\star$. Contrasting this, we repeat the experiments with fixed $c$ and vary over $m$. Figure~\ref{fig:varyM_c1} illustrates the dependence on $m$ with $c = 1$. As can be seen, within the finite time $T$, not all dynamics reach the neighborhood of the Nash equilibrium as $m$ approaches $\frac{1}{2}$.

In Figure~\ref{fig:varyM_c20}, we repeat the same experiments with $c = 20$. 
Recall from Figure~\ref{fig:2dSimplex1} that, when $m=m^\star$ and $c=20$, the dynamics do not converge within $T$ steps to a neighbourhood of the Nash equilibrium from every initialization. This naturally raises the question of how the set of convergent initializations changes as $m$ varies. The experiments suggest that this region remains essentially unchanged for $m\in(0.51,0.6]$.
 
As before, these experiments provide further support for and illustrate the asymptotic stability of the Nash equilibrium established in Theorem~\ref{thm:2x2games}. They also suggest that the qualitative stability guaranteed by the theorem is robust across the range of values of $m$ considered here, while pointing toward interesting directions for a more refined quantitative characterization.
 
 \begin{figure}[h]
\centering
\includegraphics[width=\linewidth]{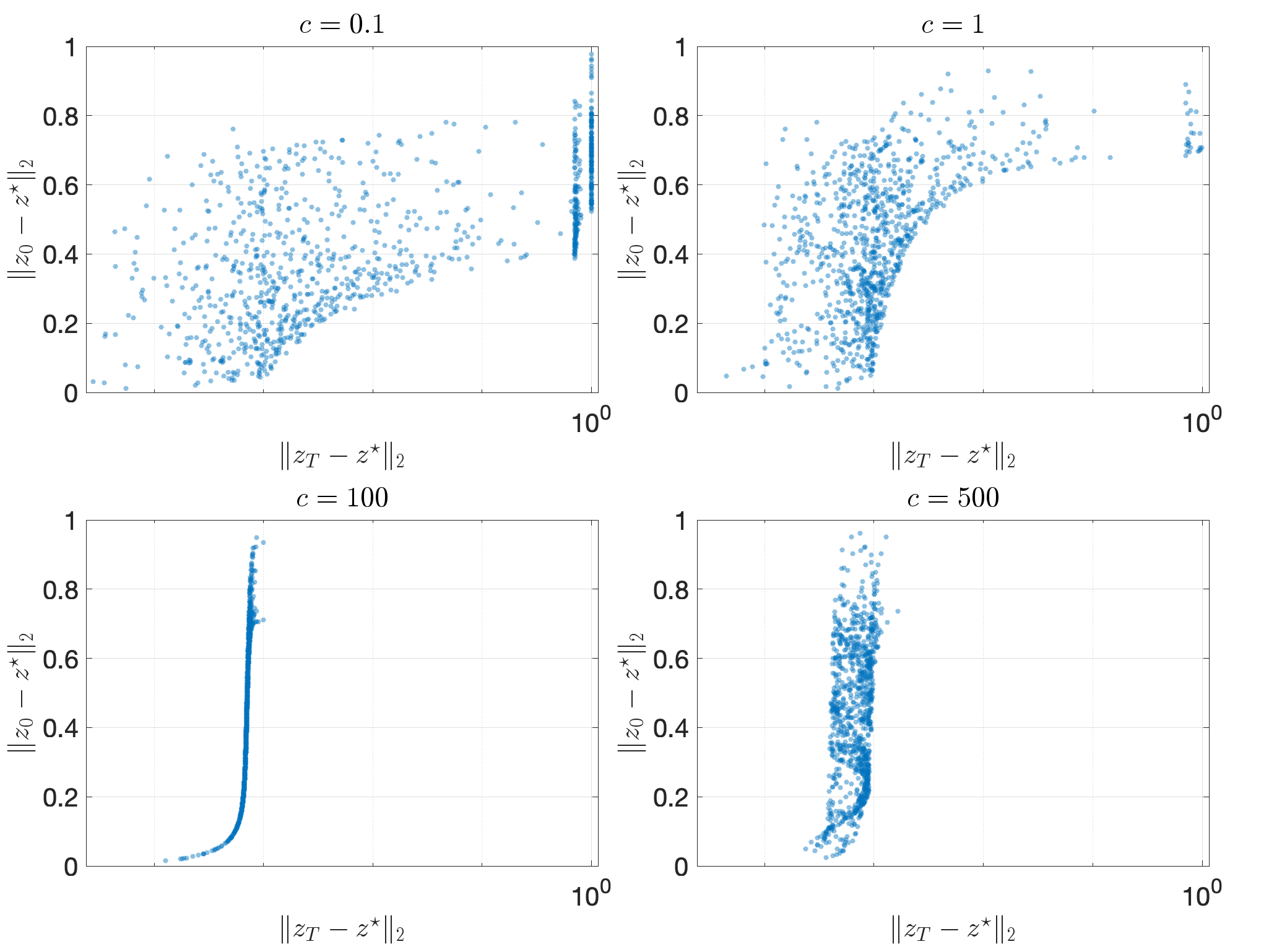}

\caption{Experiment illustrating the benefits of $\etax \ll \etay$. Clearly $c =100$ achieves better convergence than $c \in\{ 10^{-1},1\}$.}
\label{fig:basicLargeC}
\end{figure}

\subsubsection{Benefits of $\etax \ll \etay$}
While our theoretical results are primarily of fundamental interest, the preceding examples naturally raise the question: \emph{Why use asymmetric step sizes?} In addition to the motivation given in the introduction, we illustrate their benefit with a simple toy example based on a general-sum game with differing scales. A more in-depth case study is beyond the scope of this paper.

Consider the matching pennies game $\Gamma(A,-A^\top)$ (see Appendix~\ref{apx:pm}).  Our experiments  are run for the non-zero-sum variant $\Gamma(\frac{1}{\alpha} A, -\alpha A^{\top})$ with $\alpha = 10$. The total number of iterations is $T = 1\,000$, and $m = m^\star$. Figure~\ref{fig:basicLargeC} shows the initial and final norm distances of $N=1\,000$ random samples from $\Delta_2 \times \Delta_2$ with varying choices of $c$ and step sizes as defined in \eqref{eq:etas}. Unsurprisingly, $c=100$ achieves better empirical results than $c = 10^{-1}$ or $c=1$. While we note that this controlled scale mismatch is only a toy example, it illustrates potential practical benefits.

\subsection{ A $4\times 4$ Non-Zero-Sum Game}
In this section, we also study a $4\times 4$ game with a non-fully mixed equilibrium. See Appendix~\ref{apx:4x4game} for details. As a proxy for convergence, Figure~\ref{fig:4x4gameVarC} and Figure~\ref{fig:4x4gameVarM} show the norm distance after $T = 1\,000$ iterations. We sample $N = 2\,000$ points: $N/2$ are from an $\epsilon$-neighbourhood of the equilibrium, and $N/2$ are random points close to the boundaries of $\Delta_4 \times \Delta_4$.  Since the supports are $|S_x| = |S_y| = 2$, $\teig{M_x(Z^\star)|_{L_x}} = \{-\frac{280}{23}\}$ is a singleton. Set $\mu = -\frac{280}{23}$. Based on the theoretical results in Theorem~\ref{thm:jac_spectrum}, we set 
\[ \eta^\star(m) = (1-\epsilon_{\mathrm{margin}}) \frac{2m -1}{m^2(2m+1) \absv{\mu}}\, .\]
 We set the margin as $\epsilon_{\mathrm{margin}} = 0.05$ to ensure that the step size conditions are strictly satisfied. As before, the step sizes are set as in \eqref{eq:etas}.
 
 To interpret the plot, note that points above the diagonal move closer to equilibrium after $T$ iterations, while points on or below it show no progress. The $\epsilon$-neighbourhood is marked by a dashed line. 
  \begin{figure}[h]
\centering

\includegraphics[width=\linewidth]{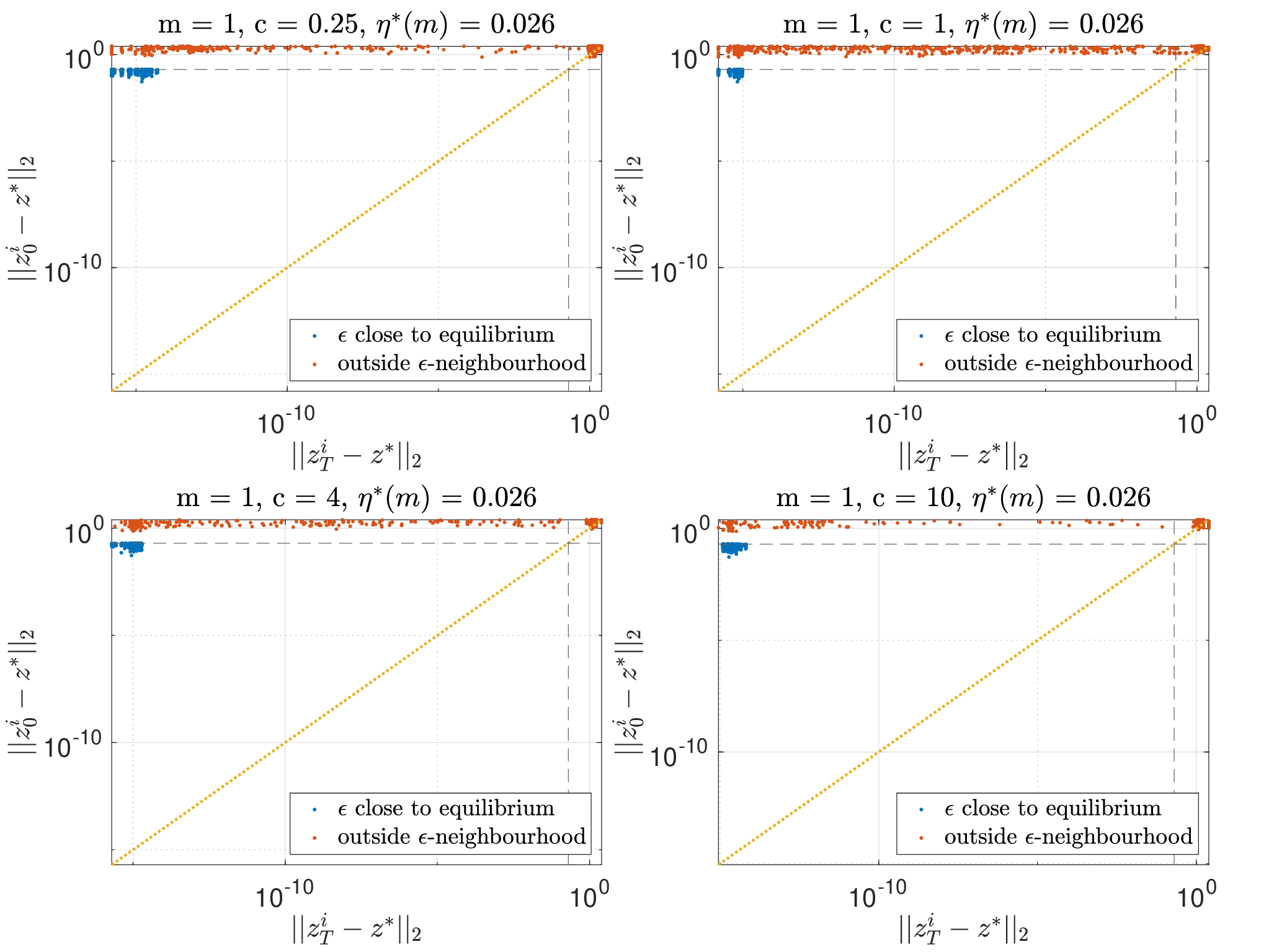}

\caption{Convergence with various choices of $c$. }
\label{fig:4x4gameVarC}
\end{figure}

\begin{figure}[h]
\centering

\includegraphics[width=\linewidth]{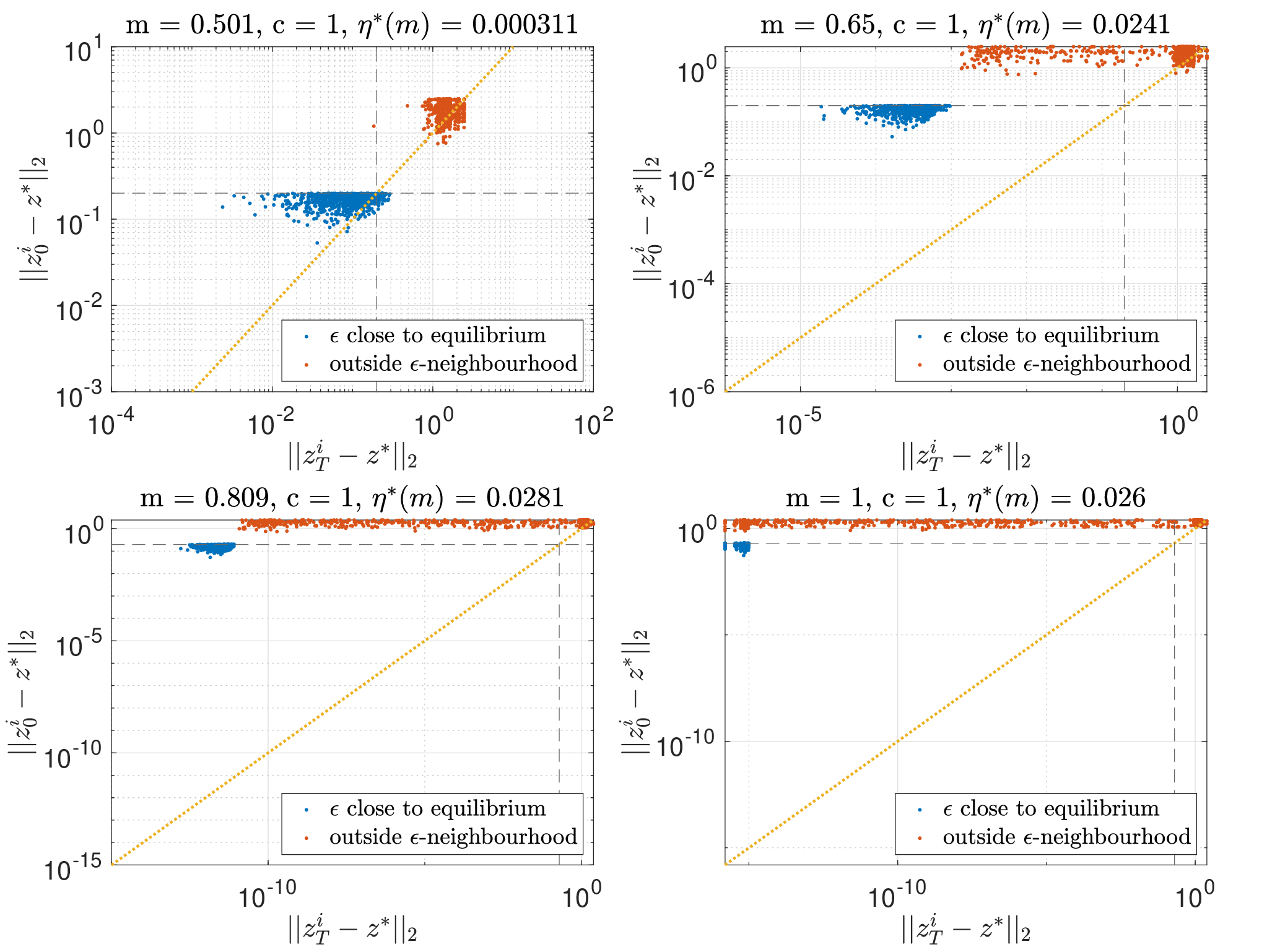}

\caption{Convergence with various choices of $m$.}
\label{fig:4x4gameVarM}
\end{figure}

\subsection{Open Questions} 
\subsubsection{Global Convergence for Zero-Sum Games and $m \in \left(\frac{1}{2}, 1 \right)$}\label{sec:open_questions}

Consider the following variant of rock-paper-scissors $\Gamma(A,-A^\top)$ with
\[
A =
\begin{bmatrix}
\phantom{-}0 & \phantom{-}a & -1 \\
-1           & \phantom{-}0 & \phantom{-} \frac{1}{a} \\
\phantom{-}1 & -1           & \phantom{-}0
\end{bmatrix}\, .
\]
We let $a = 10$. Note that this zero-sum game has a unique, fully mixed Nash equilibrium. As observed in Theorem~\ref{thm:lowDim}, this allows us to calculate the step size threshold analytically. In the following experiments we choose \[  \etax = \etay = (1 - \epsilon)\sqrt{\SR{\restr{M_x(Z^\star)}{L_x}}^{-1} \frac{ 2m^\star - 1}{(m^\star)^2(2m^\star + 1)}} \, ,\]
where $\epsilon = 0.01$ ensures asymptotic stability and $m^\star$ is set to half the golden ratio. Figure~\ref{fig:potentialNonConv} illustrates proximity and potential non-convergence. We note that this empirical observation provides some evidence that there exists $m \in (1/2, 1)$ and $\etax, \etay$ such that the equilibrium is asymptotically stable, but global convergence may not hold. However, we emphasize that numerical errors cannot be ruled out. 

 \begin{figure}[t]
\centering

\begin{minipage}[t]{.55\linewidth}
\centering
\includegraphics[width=\linewidth]{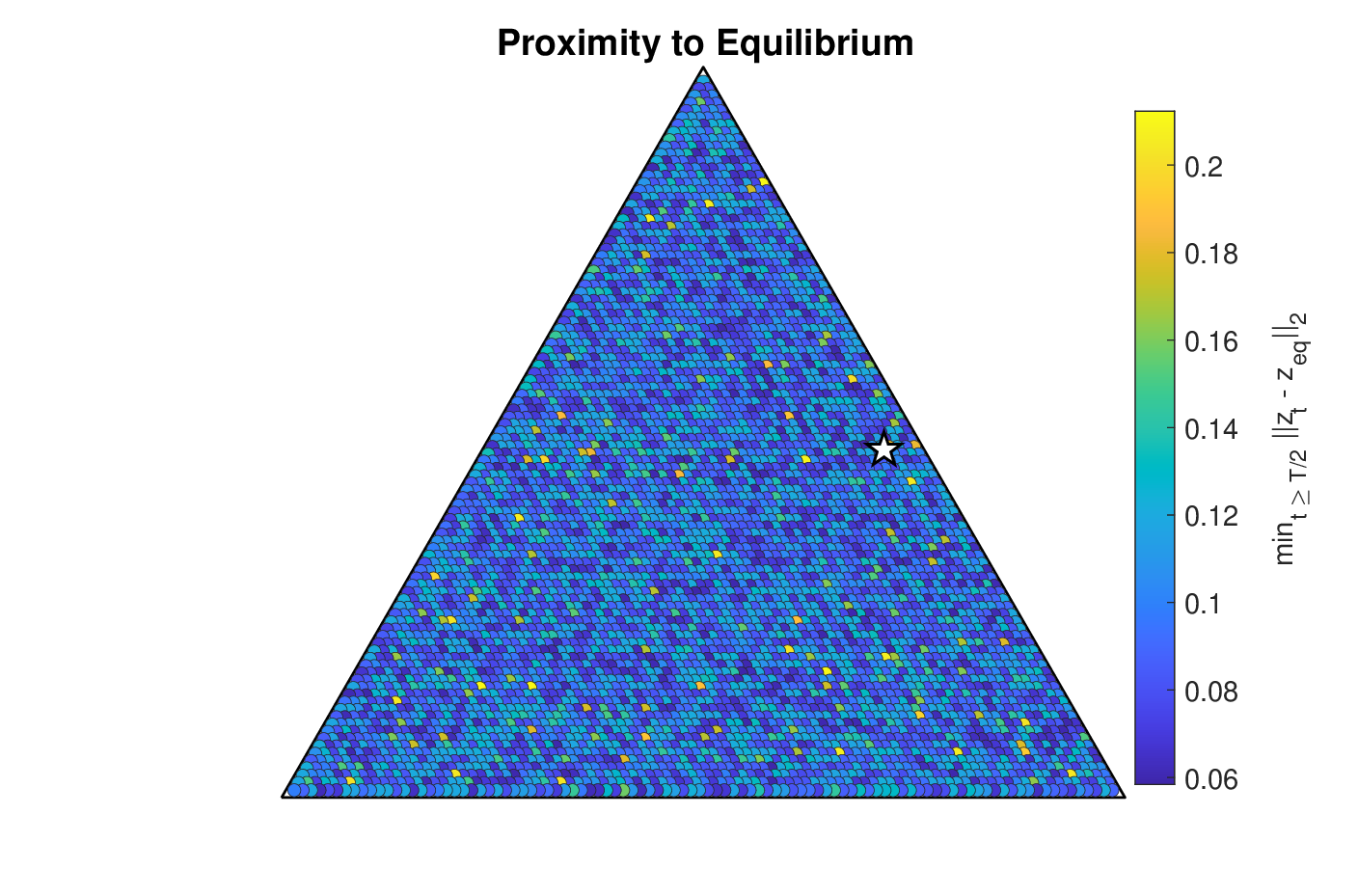}

\end{minipage}
\hfill
\begin{minipage}[t]{.44\linewidth}
\centering
\includegraphics[width=\linewidth]{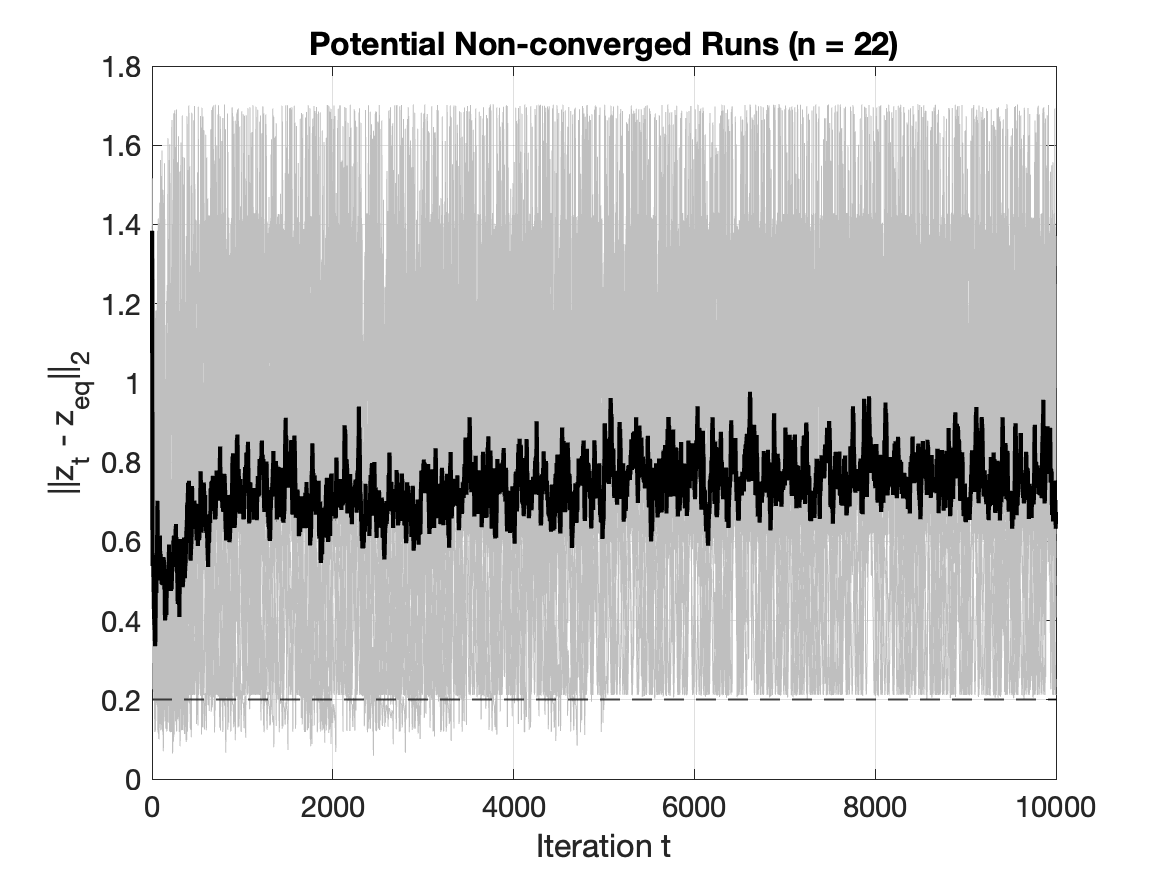}

\end{minipage}
 
\caption{We run $m^\star$-optEW for $T = 10\,000$ iterations. The initial $y_0$ is fixed at $ [0.45,0.45,0.1]^\top$. \textbf{Left plot:} We initialize $m^\star$-optEW with $(x_0^i, y_0)$ where $x_0^i$ is from a grid defined at equidistances $0.01$ over the simplex. The color indicates the minimal norm distance to the equilibrium $(z^i_t)_{t \in \NN} = (x_t^i,y_t)_{t \in \NN}$ reaches over the iterations $t \in [T/2, T]$. \textbf{Right plot:}  This graph shows the norm distances and means thereof for all $(z^i_t)_{t \in \NN}$ with $\min_{t \in [T/2, T]}\norm{z_t^i - z^\star} \geq 0.2$.  Note that for some of the sequences $\min_{t \in [1, T/2]}\norm{z_t^i - z^\star} < 0.2$.}
\label{fig:potentialNonConv}
\end{figure}

 \section{Conclusion}
  We studied optimistic exponential weights in two-player bimatrix games with constant, potentially asymmetric step sizes. Our results show that the relevant stability conditions naturally depend on the product $\eta_x\eta_y$, rather than on the two step sizes separately: in zero-sum games, this yields a global convergence guarantee to Nash equilibria, while in general games, our local spectral analysis characterizes stability through a reduced matrix at equilibrium. Via this technical result, we provide a sharper stability analysis of $m$-optEW with asymmetric step sizes. 
The experiments illustrate and support our results, but they also highlight open questions.

\paragraph{Acknowledgements}
We thank Andrea Celli for his support. 
During parts of the project, Sachs was supported by European Union – Next Generation EU funds, component M4.C2, investment 1.1. - CUP: J53D23007170001 and the London Mathematical Society, Scheme 4, grant 42508.

 \begin{appendices}
  \section{Omitted Proofs of Section \ref{sec:algo}}\label{appendix:FPPhi}
  
 We first show some basic invariance properties of exponential weights.
\begin{proposition}[Invariance] 
\label{prop:invariance}
Let $Z=[\px,\py,\ppx,\ppy]\in\cD_{{{\Phi}_m}}$ and ${{\Phi}_m}(Z)=[\pxp,\pyp,\px,\py]$. Then
for every coordinate,
\[
  \pxp_{i}=0\iff \px_{i}=0
  \qquad\text{and}\qquad
  \pyp_{j}=0\iff \py_{j}=0 .
\]
In particular, if $Z \in \relint \prodSimplex$, then $\Phi_m(Z) \in \relint \prodSimplex$.
\end{proposition}
\begin{proof}
By definition, $\pxp=P_{d_x}\!\big(\px\odot\Exp(\etax A((1+m)\py-m\ppy))\big)$. By definition of $\cD_{{{\Phi}_m}}$,
the normalizer is nonzero, and the exponential factor is strictly positive in
every coordinate, so $\pxp_{i}=0\iff \px_{i}=0$. The same argument holds for $\pyp$. If $Z \in \relint \prodSimplex$ then all coordinates of $\px, \py, \ppx,\ppy$ are strictly positive. Combined with the strict positivity of the exponential function yields the claim.
\end{proof}

\subsection{Proof of Theorem \ref{thm:FPPhi}} 

\begin{proof}[Proof of Theorem~\ref{thm:FPPhi}]
Let $Z=[\px,\py,\ppx,\ppy]\in\prodSimplex$.   By the definition of $\Phim$, $
  \Phim(Z)=Z$
is equivalent to 
\begin{align*}
\begin{cases}
\px
&=
P_{d_x}\!\left(
\px\odot \Exp\!\left(\etax A((1+m)\py-m\ppy)\right)
\right),\\
\py
&=
P_{d_y}\!\left(
\py\odot \Exp\!\left(\etay B((1+m)\px-m\ppx)\right)
\right),\\
\ppx&=\px, \qquad 
\ppy=\py.
\end{cases}  
\end{align*}
 This is equivalent to
\begin{align*}
\begin{cases}
\ppx&=\px, \qquad \ppy=\py,\\
\px
&=
P_{d_x}\!\left(
\px\odot \Exp(\etax A \py)
\right),\\
\py
&=
P_{d_y}\!\left(
\py\odot \Exp(\etay B \px)
\right).
\end{cases}
\end{align*}
Define
\[
  C_x
  :=
  \sum_{\ell=1}^{d_x}
  \px_{\ell}\exp\!\left(\etax (A\py)_\ell\right).
\]
Since $\px\in\Delta_{d_x}$, we have $\supp(\px) = S_x\neq\varnothing$ and $C_x>0$. By
Proposition~\ref{prop:invariance}, the zero pattern of the first block is preserved by
the update. Hence, the coordinates outside $S_x$ impose no additional condition.
Therefore,  \[
\begin{aligned}
\px
&=
P_{d_x}\!\left(
\px\odot \Exp(\etax A \py)
\right)
\\
&\iff
\forall i\in S_x,\quad
\px_{i}
=
\px_{i}
\frac{\exp(\etax (A\py)_i)}{C_x}
\\
&\iff
\forall i\in S_x,\quad
\exp(\etax (A\py)_i)=C_x
\\
&\iff
\forall i,j\in S_x,\quad
(A\py)_i=(A\py)_j \, .
\end{aligned}
\]
For the last step, we used that $\etax>0$.
 Combined with the observation that the argument for $\py$ is identical, we note that $Z$ is a fixed point of
$\restr{\Phim}{\prodSimplex}$ if and only if
\[
  \ppx=\px\in\Delta_{d_x},\qquad
  \ppy=\py\in\Delta_{d_y},
\]
and the payoff vectors $A\py$ and $B\px$ are constant on the supports of
$\px$ and $\py$, respectively.  \end{proof}

\section{General Technical Results}
Recall from the main part that by $\odot$ we denote the Hadamard (entry-wise) multiplication and $\Exp$ denotes the component-wise exponential function. In this section, we also use entry-wise division, denoted by $\odiv$, and the component-wise logarithm $\Log$.
Throughout the section, we let $(Z_t)_{t \in \NN}$ denote a sequence generated by
$Z_t=\Phi_m(Z_{t-1})$  with $m \in (0,1]$. We always assume that $Z_0 \in \prodSimplex$ and for some results we use the stronger assumption that $Z_0 \in \relint \prodSimplex$. This assumption is indicated for each result. For each
$t\in\NN$, write $Z_t=[x_{t},y_{t},x_{t-1},y_{t-1}]$ and $Z_{1:t} = [Z_1, \dots, Z_t]$, with the convention $Z_{1:0}=\varnothing$.
Further, we let $\bar y_{1:t} =\sum_{s=1}^t y_s$ and $\bar x_{1:t} =\sum_{s=1}^t x_s$ with the convention that $\bar x_{1:0} = \bar y_{1:0} = 0$.
Further, for $x\in \Delta_d, x' \in \relint \Delta_d$, we denote the Kullback-Leibler (KL) divergence by $\KL(x,x') = \sum_{i =1}^d x_i \log \frac{ x_i}{ x_i'}$. We use the convention that $0 \log 0 = 0$.

\subsection{Known Results} \label{sec:knownResults}
For completeness, we include several well-known results for the exponential weight updates and bounds for the KL divergence.  For a proof of these results, see, e.g., \citet{cesa-bianchi2006predicti}.  %
 \begin{proposition}[Pinsker's inequality] \label{prop:Pinsker}
For any $x,x' \in \Delta_d$, 
\begin{align*}
  \KL(x, x')
 \geq \frac{1}{2} \norm{x - x' }_1^2\, . 
\end{align*}
\end{proposition}
\begin{proposition}[Three-point identity for KL-divergence] \label{threePointKL}
 For any $x \in \Delta_d $ and $y, z \in \relint \Delta_d$, 
\[
\KL(x, z)
= \KL(x, y) + \KL(y, z)
+ \Big\langle x - y , \Log ({y}\odiv{z}) \Big\rangle\, .
\]
\end{proposition}
 Now define the log-sum-exp potentials
\begin{equation}\label{def:LSE}
\begin{array}{l}
\displaystyle
\psi^x(Z_{1:t})
:=
\frac{1}{\eta_x}
\log
\left\langle
 x_1,
\Exp\!\left(
\eta_x A\big( \bar y_{1:t} + m(y_t- y_{0}) \big)
\right)
\right\rangle ,
\\[8pt]
\displaystyle
\psi^y(Z_{1:t})
:=
\frac{1}{\eta_y}
\log
\left\langle
 y_1,
\Exp\!\left(
\eta_y B\big( \bar x_{1:t}+ m(x_{t} - x_{0}) \big)
\right)
\right\rangle .
\end{array}
\end{equation}
We set $\psi(Z_{1:t}):=\psi^x(Z_{1:t})+\psi^y(Z_{1:t})$ and use the convention that $\psi^x(Z_{1:0}) =\psi^y(Z_{1:0})= 0$.

\begin{proposition}[Logit identity]
\label{prop:KL_logit1}
Consider the game $\Gamma(A,B)$ and assume $Z_0 \in \relint \prodSimplex$. Then, for every $t\geq 1$,
  \begin{enumerate}
  \item  $ \Log({x_{t+1}}\odiv{x_{t}})
    =
    \eta_x A((1+m)y_{t}- my_{t-1})
    +
    \eta_x
    \big(
        \psi^x(Z_{1:t-1})-\psi^x(Z_{1:t})
    \big)\mathbf 1 \, ,$
    \item  $ \Log({y_{t+1}}\odiv{y_{t}})
    =
    \eta_y  B((1+m)x_{t}- m x_{t-1})
    +
    \eta_y
    \big(
        \psi^y(Z_{1:t-1})-\psi^y(Z_{1:t})
    \big)\mathbf 1 \, .$
    \end{enumerate}

 \end{proposition}

\begin{proof}
We prove the identity for the $x$-player; the proof for the $y$-player is analogous.

By iterating the $m$-optEW update, for every coordinate $i \in \{1,\dots, d_x\}$,
\[
  [x_{t+1}]_i
  =
  \frac{
    [x_1]_i
    \exp\!\left(
      \eta_x
      \left[
        A
        \sum_{s=1}^t\big((1+m)y_s-my_{s-1}\big)
      \right]_i
    \right)
  }{
    \sum_{k=1}^{d_x}
    [x_1]_k
    \exp\!\left(
      \eta_x
      \left[
        A
        \sum_{s=1}^t\big((1+m)y_s-my_{s-1}\big)
      \right]_k
    \right)
  } .
\]
The sum simplifies to
\[
\sum_{s=1}^t((1+m)y_{s}-my_{s-1})
=
m(y_{t} - y_{0})+\bar y_{1:t} =: Y_t^{(m)}\, .
\]
 Since $Z_0\in\relint\prodSimplex$, all iterates remain strictly positive, so the coordinate-wise logarithm of the ratio is well defined.   Taking the logarithm of the quotient gives
\[
  \log\frac{[x_{t+1}]_i}{[x_t]_i}
  =
  \eta_x[A(Y_t^{(m)}-Y_{t-1}^{(m)})]_i
  +
  \eta_x\Big( \psi^x(Z_{1:t-1})-\psi^x(Z_{1:t})\Big).
\]
Finally,
\[
  Y_t^{(m)}-Y_{t-1}^{(m)}
  =
  (1+m)y_t-my_{t-1}.
\]
This proves the first identity. The second follows from an analogous argument. \end{proof}

     \begin{lemma}\label{lem:delta_epsilon_tau}
  Consider a bimatrix game $\Gamma(A,B)$ and assume $Z_0 \in \prodSimplex$.
Let
\[
  \tZ=\Zfp \in \FP
\]
be a fixed point of $\Phi_m$.  
Assume that one of the following two alternatives holds:
\begin{enumerate}
\item[\textnormal{(I)}] there exist
$\hat \imath\notin\supp(\tx)$ and $j\in\supp(\tx)$ such that
\[
  (A\ty)_{\hat \imath}>(A\ty)_j\, .
\]
In this case, assume $[x_0]_j > 0$ and define
\[
  R_t:=\frac{[x_t]_{\hat \imath}}{[x_t]_j};
\]
\item[\textnormal{(II)}] there exist
$\hat k\notin\supp(\ty)$ and $\ell\in\supp(\ty)$ such that
\[
  (B\tx)_{\hat k}>(B\tx)_\ell\, .
\]
 In this case, assume $[y_0]_\ell >0$ and define
\[
  R_t:=\frac{[y_t]_{\hat k}}{[y_t]_\ell}.
\]
\end{enumerate}
 Then there exist constants
$\delta>0$, $\epsilon>0$, and $\tau>0$ such that:
\begin{enumerate}
\item \textbf{Exponential growth:} for every $t\in\NN$,
\[
  \|Z_t-\tZ\|<\delta
  \quad\Longrightarrow\quad
  R_{t+1}\ge e^\epsilon R_t .
\]
\item \textbf{Bounded ratio:} for every $Z=[x,y,x',y']\in\prodSimplex$,
\[
  \|Z-\tZ\|<\delta
  \quad\Longrightarrow\quad
  R(Z)<\tau,
\]
where
\[
  R(Z)=
  \begin{cases}
    x_{\hat \imath}/x_j, & \text{in case \textnormal{(I)},}\\
    y_{\hat k}/y_\ell, & \text{in case \textnormal{(II)}.}
  \end{cases}
\]
\end{enumerate}
\end{lemma}

\begin{proof}
We prove case \textnormal{(I)}, the case \textnormal{(II)} is analogous.
Set
\[
  g:=(A\ty)_{\hat \imath}-(A\ty)_j>0,
  \qquad
  \beta:=\frac{\tx_{j}}{2}>0.
\]
For any $Z=[x,y,x',y']$, define the continuous local payoff-difference map
\[
  \Gamma_x^{(m)}(Z)
  :=
  [A((1+m)y-my')]_{\hat \imath}
  -
  [A((1+m)y-my')]_j .
\]
Since $\tZ=[\tx,\ty,\tx,\ty]$, we have
$
  (1+m)\ty-m\ty=\ty.
$
Therefore
\[
  \Gamma_x^{(m)}(\tZ)
  =
  (A\ty)_{\hat \imath}-(A\ty)_j
  =
  g.
\]
By continuity, there exists $\delta_1>0$ such that
\[
  \|Z-\tZ\|<\delta_1
  \quad\Longrightarrow\quad
  \Gamma_x^{(m)}(Z)>\frac{g}{2}.
\]
Moreover, since $j\in\supp(\tx)$, we have $\tx_j>0$.
Thus there exists $\delta_2>0$ such that
\[
  \|Z-\tZ\|<\delta_2
  \quad\Longrightarrow\quad
  x_j>\beta .
\]
Let
\[
  \delta:=\min\{\delta_1,\delta_2\},
  \qquad
  \epsilon:=\frac{\eta_x g}{2},
  \qquad
  \tau:=\frac{\delta}{\beta}.
\]
We first show the exponential growth statement.
Since the normalization factors in the update operation cancel, we obtain 
\[
  R_{t+1}
  =
  R_t
  \exp\!\left(
    \eta_x\Gamma_x^{(m)}(Z_t)
  \right).
\]
If $\|Z_t-\tZ\|<\delta$, then $\|Z_t-\tZ\|<\delta_1$, and hence
\[
  \Gamma_x^{(m)}(Z_t)>\frac{g}{2}.
\]
Therefore
\[
  R_{t+1}
  =
  R_t
  \exp\!\left(
    \eta_x\Gamma_x^{(m)}(Z_t)
  \right)
  \ge
  R_t\exp\!\left(\frac{\eta_x g}{2}\right)
  =
  e^\epsilon R_t .
\]
This proves the exponential growth statement in case \textnormal{(I)}.

It remains to prove the bounded-ratio statement.
Let $Z=[x,y,x',y']\in\prodSimplex$ satisfy
\[
  \|Z-\tZ\|<\delta .
\]
Since $\hat \imath\notin\supp(\tx)$, we have $\tx_{\hat \imath}=0$.
Thus
\[
  x_{\hat \imath}
  =
  |x_{\hat \imath}-\tx_{\hat \imath}|
  \le
  \|Z-\tZ\|
  <
  \delta.
\]
On the other hand, since $\delta\le\delta_2$, we have
\[
  x_j>\beta.
\]
Consequently,
\[
  R(Z)
  =
  \frac{x_{\hat \imath}}{x_j}
  <
  \frac{\delta}{\beta}
  =
  \tau.
\]
This proves the bounded-ratio statement in case \textnormal{(I)}.

 \end{proof}

\section{Omitted Proofs from Section~\ref{sec:global}}
\subsection{Proof of Theorem~\ref{thm:convFP}}\label{apx:proof_convFP}

\paragraph{Key Technical Results} 
 A key ingredient is the Lyapunov function, consisting of a KL-divergence, a cross-term gap function, and the log-sum-exp $\psi$ defined in equation \eqref{def:LSE}. For
  $
    Z=[x,y,x',y'] \in \relint{ \prodSimplex}
  $, let $\varphi(Z):=\varphi^x(Z)+\varphi^y(Z)$ with
  \begin{align*}
    \varphi^x(Z):=\frac{1}{2\etax}&\KL(x,x'),
    \qquad \text{ and }\qquad
    \varphi^y(Z):=\frac{1}{2\etay}\KL(y,y')\, ;
      \intertext{and define a gap function capturing the cross terms as}
      &\gap(Z):=\langle x,Ay'\rangle-\langle y,A^\top x'\rangle\, .
  \end{align*}
  The Lyapunov function is
  \[
    V(Z_{1:t}):=\psi(Z_{1:t-1})+\varphi(Z_t)-\gap(Z_t).
  \]
 Before we show our key technical lemma, we derive an identity for $\psi$ from Proposition~\ref{prop:KL_logit1} in the special case of zero-sum games and $m = 1$.
  \begin{proposition} [$\psi$ identity] \label{prop:KLidentity_logit}
Consider the game $\Gamma(A,-A^\top)$ and denote by $(Z_t)_{t \in \NN}$ a
sequence with $Z_0\in\relint\prodSimplex$ and
$Z_{t+1}:=\Phi_1(Z_t)$. Then for every $t \geq 1$
\begin{enumerate}
\item $
    \psi^x(Z_{1:t-1})-\psi^x(Z_{1:t})
    =
    -x_{t+1}^{\top}A(2y_{t}-y_{t-1})
    +
    \frac{1}{\eta_x}\KL(x_{t+1},x_{t}) $;
 \item $
    \psi^y(Z_{1:t-1})-\psi^y(Z_{1:t})
    =
    (2x_{t}-x_{t-1})^\top A y_{t+1}
    +
    \frac{1}{\eta_y}\KL(y_{t+1},y_{t}) $ ;

\end{enumerate}
and
\[
\begin{aligned}
\psi(Z_{1:t-1})-\psi(Z_{1:t})
={}&
-x_{t+1}^{\top}A(2y_{t}-y_{t-1})
+(2x_{t}-x_{t-1})^{\top}Ay_{t+1}  \\
&\quad
+\frac{1}{\eta_x}\KL(x_{t+1},x_{t})
+\frac{1}{\eta_y}\KL(y_{t+1},y_{t}) .
\end{aligned}
\]
\end{proposition}
\begin{proof}
Again, we only show the first identity; the second identity follows the analogous
argument.
 Using the definition of the KL-divergence, Proposition~\ref{prop:KL_logit1} and that
$x_{t+1}\in\Delta_{d_x}$, we obtain
\[
\begin{aligned}
\frac{1}{\eta_x}\KL(x_{t+1},x_{t})
&=
\frac{1}{\eta_x}
\left\langle
    x_{t+1},
    \Log({x_{t+1}}\odiv{x_{t}})
\right\rangle  \\
&=
x_{t+1}^{\top}A(2y_{t}-y_{t-1})
+
\psi^x(Z_{1:t-1})-\psi^x(Z_{1:t}) .
\end{aligned}
\]
Rearranging yields
\[
    \psi^x(Z_{1:t-1})-\psi^x(Z_{1:t})
    =
    -x_{t+1}^{\top}A(2y_{t}-y_{t-1})
    +
    \frac{1}{\eta_x}\KL(x_{t+1},x_{t}) .
\]
Analogously, \[
    \psi^y(Z_{1:t-1})-\psi^y(Z_{1:t})
    =
    (2x_{t}-x_{t-1})^\top A y_{t+1}
    +
    \frac{1}{\eta_y}\KL(y_{t+1},y_{t}) .
\]
Adding the two identities gives the third claim.
\end{proof}

 \begin{lemma}[Key Technical Lemma]\label{lem:keyTech}
 Consider a zero-sum game $\Gamma(A,-A^\top)$. Assume that $\Phi_1$ is defined with 
 $\etax,\etay>0$ and $6\etax \etay  \norm{A}^2_{\oneinf} \leq 1$.
  Let $(Z_t)_{t\in \NN}$
be the sequence generated by $Z_{t+1}=\Phi_1(Z_t)$ with
$Z_0\in \relint\prodSimplex$. Using the definitions above, the sequence
\begin{enumerate}
\item   $\left(\psi(Z_{1:t})\right)_{t\in \NN}$ is bounded from below by a constant;
\item $\left(V(Z_{1:t})\right)_{t\in \NN}$ is non-increasing; and
\item $\norm{Z_t - Z_{t-1}}  \xrightarrow[t \to \infty]{}  0$. 
\end{enumerate}
 
 \end{lemma}
 \begin{proof}
 Let $Z^\star = [x^\star, y^\star, x^\star, y^\star] \in \NE$.  
  Applying the three-point identity (Proposition~\ref{threePointKL}), Proposition~\ref{prop:KL_logit1}  and Proposition~\ref{prop:KLidentity_logit} yields
\begin{align*}
    \frac{1}{\etax}\KL(x^\star, x_{t})-  \frac{1}{\etax}\KL(x^\star, x_{t+1})
    &=
    \frac{1}{\eta_x}\KL(x_{t+1},x_{t})
    +
    \frac{1}{\eta_x}
    \dprod{
        x^\star-x_{t+1},
        \Log({x_{t+1}}\odiv{x_{t}})
     }\\
     &\overset{(1)}{=}
    \frac{1}{\eta_x}\KL(x_{t+1},x_{t})
    +
    \dprod{
        x^\star-x_{t+1},
        A(2y_{t}-y_{t-1})
    }  \\
    &=
    \psi^x(Z_{1:t-1})-\psi^x(Z_{1:t})
    +
    \left\langle
        x^\star,
        A(2y_{t}-y_{t-1})
    \right\rangle .
\end{align*}
 For $(1)$ we used that $\dprod{x^\star - x_{t+1}, c \ones} = 0$ for any $c$ (here $c = \etax(\psi^x(Z_{1:t-1}) - \psi^x(Z_{1:t}))$). 
Analogously,
\[
    \frac{1}{\etay}\KL(y^\star, y_{t})-  \frac{1}{\etay}\KL(y^\star, y_{t+1})
    =
    \psi^y(Z_{1:t-1})-\psi^y(Z_{1:t})
    -
    \left\langle
        y^\star,
        A^\top(2x_{t}-x_{t-1})
    \right\rangle .
\]
Therefore, $ \frac{1}{\etax}\KL(x^\star, x_{t})-  \frac{1}{\etax}\KL(x^\star, x_{t+1}) +  \frac{1}{\etay}\KL(y^\star, y_{t})-  \frac{1}{\etay}\KL(y^\star, y_{t+1}) $ is equal to
\begin{align*}
    &\psi(Z_{1:t-1})-\psi(Z_{1:t}) 
    +
     \dprod{
        x^\star,
        A(2y_{t}-y_{t-1})
    }
    -
    \dprod{
        y^\star,
        A^\top(2x_{t}-x_{t-1})
    }\\
    &\quad=
    \psi(Z_{1:t-1})-\psi(Z_{1:t})  
    +
    \left\langle
        x^\star,
        A(y_{t}-y_{t-1})
    \right\rangle
    -
    \left\langle
        y^\star,
        A^\top(x_{t}-x_{t-1})
    \right\rangle  \\
    &\phantom{xxxxxxxxxxxxxxxx}+
    \left\langle
        x^\star,
        Ay_{t}
    \right\rangle
    -
    \left\langle
        y^\star,
        A^\top x_{t}
    \right\rangle .
\end{align*}
Summing over $t=1,\ldots,T$ yields
\[
\begin{aligned}
     &\frac{1}{\etax}\left(\KL(x^\star, x_{1})-  \KL(x^\star, x_{T+1}) \right)+  \frac{1}{\etay}\left(\KL(y^\star, y_{1})-  \KL(y^\star, y_{T+1}) \right)\\
    &\qquad \qquad=
    \psi(Z_{1:0})-\psi(Z_{1:T})  
    +
    \dprod{
        x^\star,
        A(y_{T}-y_{0})
    }    -
    \dprod{
        y^\star,
        A^\top( x_{T}-x_{0})
    }  \\
    &\qquad \qquad\qquad \qquad\quad
    +
    \sum_{t=1}^T
    \left(
         \dprod{ x^\star,Ay_{t}}
        -
        \dprod{ y^\star,A^\top x_{t}}
    \right)\\
    &\qquad \qquad\geq
    \psi(Z_{1:0})-\psi(Z_{1:T})  
    - \left(
    \dprod{
        x^\star,
        A y_{0}
    }    -
    \dprod{
        y^\star,
        A^\top x_{0}
    } \right) \, ,
\end{aligned}
\]
where the last inequality is due to $Z^\star \in \NE$. 
Rearranging
\[
    \psi(Z_{1:T})
    \geq
       \psi(Z_{1:0})
    -
  \underbrace{ \left( \frac{1}{\etax} \KL(x^\star, x_{1}) + \frac{1}{\etay}\KL(y^\star, y_{1}) \right)}_{=: k_1}
    -
    \left(
    \dprod{
        x^\star,
        A y_{0}
    }    -
    \dprod{
        y^\star,
        A^\top x_{0}
    } \right)\, .
\]
Since, by assumption, $Z_0 \in \relint \prodSimplex$ the term $k_1$ is bounded by a constant. Thus $(\psi(Z_{1:t}))_{t\in \NN}$ is bounded from below by a constant  independent of $T$.

 Next, we prove the descent inequality, that is, we show that $(V(Z_{1:t}))_{t\in \NN}$ is non-increasing.  By definition, $ V(Z_{1:t})-V(Z_{1:t+1})$ is equal to
\begin{align*}
         \psi(Z_{1:t-1})-&\psi(Z_{1:t})   -
    \gap(Z_t)
    +
    \gap(Z_{t+1})\\
    +
    \frac{1}{2\eta_x}\KL(x_{t},x_{t-1})
    +
   &\frac{1}{2\eta_y}\KL(y_{t},y_{t-1})  
    -
    \frac{1}{2\eta_x}\KL(x_{t+1},x_{t})
    -
    \frac{1}{2\eta_y}\KL(y_{t+1},y_{t})\,.
\end{align*}
 Using the entropy identity from  Proposition~\ref{prop:KLidentity_logit} and the definition of the gap function, yields
\begin{align*}
    V(Z_{1:t})-&V(Z_{1:t+1})\\
    &=
    -(x_{t+1})^\top A(2y_{t}-y_{t-1})
    +
    (2x_{t}-x_{t-1})^\top Ay_{t+1}  -
    x_{t}^\top Ay_{t-1}\\
    &\qquad+
    x_{t-1}^\top Ay_{t}
    +
    (x_{t+1})^\top Ay_{t}
    -
    x_{t}^\top Ay_{t+1}  \\
    &\qquad+
    \frac{1}{2\eta_x}\KL(x_{t},x_{t-1})
    +
    \frac{1}{2\eta_y}\KL(y_{t},y_{t-1})  
    +
    \frac{1}{2\eta_x}\KL(x_{t+1},x_{t})
    +
    \frac{1}{2\eta_y}\KL(y_{t+1},y_{t})\, ,
    \intertext{After regrouping the bilinear terms, we obtain}
    &=
    -(x_{t+1}-x_{t})^\top A(y_{t}-y_{t-1})
    +
    (x_{t}-x_{t-1})^\top A(y_{t+1}-y_{t})  \\
    &+
    \frac{1}{2\eta_x}\KL(x_{t},x_{t-1})
    +
    \frac{1}{2\eta_y}\KL(y_{t},y_{t-1})  
    +
    \frac{1}{2\eta_x}\KL(x_{t+1},x_{t})
    +
    \frac{1}{2\eta_y}\KL(y_{t+1},y_{t})\, .
\end{align*} 
By Fenchel's inequality and the definition of the operator norm,
\[
\begin{aligned}
    -(x_{t+1}-x_{t})^\top A(y_{t}-y_{t-1})
    &\geq
    -
    \frac{3\eta_y}{2}
    \|A^\top(x_{t+1}-x_{t})\|_\infty^2
    -
    \frac{1}{6\eta_y}
    \|y_{t}-y_{t-1}\|_1^2  \\
    &\geq
    -
    \frac{3\eta_y\|A^\top\|_{{\oneinf}}^2}{2}
    \|x_{t+1}-x_{t}\|_1^2
    -
    \frac{1}{6\eta_y}
    \|y_{t}-y_{t-1}\|_1^2 \, ,
\end{aligned}
\]
and analogously,
\[
\begin{aligned}
    (x_{t}-x_{t-1})^\top A(y_{t+1}-y_{t})
    &\geq
    -
    \frac{3\eta_x\|A\|_{{\oneinf}}^2}{2}
    \|y_{t+1}-y_{t}\|_1^2
    -
    \frac{1}{6\eta_x}
    \|x_{t}-x_{t-1}\|_1^2 .
\end{aligned}
\]
Therefore, by Pinsker's inequality (cf. Proposition~\ref{prop:Pinsker}),
 \[
\begin{aligned}
    V(Z_{1:t})-V(Z_{1:t+1})
    \geq&
    \left(
        \frac{1}{4\eta_x}
        -
        \frac{3\eta_y\|A^\top\|_{{\oneinf}}^2}{2}
    \right)
    \|x_{t+1}-x_{t}\|_1^2  \\
    &+
    \left(
        \frac{1}{4\eta_y}
        -
        \frac{3\eta_x\|A\|_{{\oneinf}}^2}{2}
    \right)
    \|y_{t+1}-y_{t}\|_1^2  \\
    &+
    \frac{1}{12\eta_x}\|x_{t}-x_{t-1}\|_1^2
    +
    \frac{1}{12\eta_y}\|y_{t}-y_{t-1}\|_1^2 .
\end{aligned}
\]
 By the step size condition, (note that $\|A^\top\|_{{\oneinf}} = \|A\|_{{\oneinf}}$)
\[
    \frac{1}{4\eta_x}
    -
    \frac{3\eta_y\|A^\top\|_{{\oneinf}}^2}{2}
    \geq 0,
    \qquad
    \frac{1}{4\eta_y}
    -
    \frac{3\eta_x\|A\|_{{\oneinf}}^2}{2}
    \geq 0.
\]
Hence $V(Z_{1:t+1})\leq V(Z_{1:t})$ for all $t$.

To show part 3, we note that $V(Z_{1:1})$ is a constant. Further, since $\psi(Z_{1:t})$ is bounded from below and the KL-divergence is non-negative,
while $\gap(Z_t)$ is bounded on the compact set $\prodSimplex$,
the sequence $(V(Z_{1:t}))_{t\in \NN}$ is bounded from below. Summing over $t$ gives
\[
    \sum_{t=1}^\infty
    \left(
       \frac{1}{12\etax} \|x_{t}-x_{t-1}\|_1^2
        +
       \frac{1}{12\etay} \|y_{t}-y_{t-1}\|_1^2
    \right)
    <\infty \, .
\]
Consequently,
\[
    \|x_{t+1}-x_{t}\|_1\to 0,
    \qquad
    \|y_{t+1}-y_{t}\|_1\to 0.
\]
Because
\[
    Z_t=(x_{t},y_{t},x_{t-1},y_{t-1}),
\]
we conclude that
\[
    \|Z_t-Z_{t-1}\|\to 0 \, .
\]
This completes the proof.
\end{proof}

 \begin{proof}[Proof of Theorem~\ref{thm:convFP}]
  By Proposition~\ref{prop:invariance}, the orbit remains in $\prodSimplex$.
Since $\prodSimplex$ is compact, the sequence $(Z_t)_{t\in\NN}$ has at least one
accumulation point.

We first show that every accumulation point of $(Z_t)_{t\in\NN}$ is a fixed
point of $\Phi_1$. Let $\bar Z$ be an accumulation point. Then there exists a
subsequence $(Z_{t_k})_{k\in\NN}$ such that
$
  Z_{t_k}\to \bar Z .
$
By Lemma~\ref{lem:keyTech}, Part~3, we have
$
  \norm{Z_{t+1}-Z_t}\to 0 .
$
Hence
\[
  \norm{Z_{t_k+1}-Z_{t_k}}\to 0,
\]
and therefore
\[
  Z_{t_k+1}\to \bar Z .
\]
 By continuity of $\Phi_1$ on $\prodSimplex$,
\[
\bar Z =  \lim Z_{t_k+1}
  =
 \lim \Phi_1(Z_{t_k})
  =
 \Phi_1(\bar Z)\, .
\]
 Thus, every accumulation point of the orbit belongs to the fixed-point set
$\FPone$.

It remains to show that the distance to $\FPone$ converges to zero. Suppose, for
contradiction, that
$
  d(Z_t,\FPone)
$ does not converge to $0$.
Then there exist $\epsilon>0$ and a subsequence $(Z_{t_k})_{k\in\NN}$ such that
\[
  d(Z_{t_k},\FPone)\ge \epsilon
  \qquad
  \text{for all } k\in\NN .
\]
By compactness of $\prodSimplex$, after passing to a further subsequence if
necessary, we may assume that
\[
  Z_{t_k}\to \bar Z
\]
for some $\bar Z\in\prodSimplex$. By the first part of the proof, $\bar Z$ is a
fixed point, i.e. $\bar Z\in\FPone$. Since the distance function
$Z\mapsto d(Z,\FPone)$ is continuous, we obtain
\[
  d(Z_{t_k},\FPone)\to d(\bar Z,\FPone)=0,
\]
contradicting $d(Z_{t_k},\FPone)\ge \epsilon$ for all $k$.
Therefore
\[
  d(Z_t,\FPone)\to 0.
\]
Since this holds for any $Z_0 \in \relint \prodSimplex$, $\FPone$ is
globally attracting with respect to $\relint\prodSimplex$.

 \end{proof}
 
 \subsection{Proof of Theorem~\ref{thm:NEconvergenceZSG}}\label{apx:proofglobal}

\subsubsection{Convergence to a single fixed point}

The key ingredient is that the \emph{whole} sequence converges to a single
fixed point. We obtain this from the structure of the $\omega$-limit set
\[
  \Omega_{\omega} \;:=\; \bigcap_{N\ge 1}\overline{\{Z_t : t\ge N\}}\,,
\]
combined with two facts that the consecutive increments vanish
(Lemma~\ref{lem:keyTech}, Part~3) and $\prodSimplex$ is compact. 
We add the following well-known results for completeness.

\begin{proposition}[Connectedness of the limit set]\label{lem:omega_connected}
  Let $(U_t)_{t\in\NN}$ be a sequence in a compact metric space with
  $\|U_{t+1}-U_t\|\to 0$, and let
  $\Omega_{\omega}=\bigcap_{N\ge1}\overline{\{U_t:t\ge N\}}$ be its $\omega$-limit set.
  Then $\Omega_{\omega}$ is nonempty, compact, and connected.
\end{proposition}
\begin{proof}
  $\Omega_{\omega}$ is a decreasing intersection of nonempty compact sets in a compact
  space, hence nonempty and compact.

  Suppose, for contradiction, that $\Omega_{\omega}$ is disconnected, say
  $\Omega_{\omega}=K_1\sqcup K_2$ (disjoint union) with $K_1,K_2$ nonempty compact and
  $2\delta:=d(K_1,K_2)>0$. Define the $1$-Lipschitz function
  $g(t):=d(U_t,K_1)$, so that
  \[
    |g(t+1)-g(t)|\le \|U_{t+1}-U_t\|\xrightarrow[t\to\infty]{}0 .
  \]
  Because $K_1\subseteq\Omega_{\omega}$, we have $g(t)<\delta/2$ for infinitely many $t$.
  Since $K_2\subseteq\Omega_{\omega}$ and $d(K_2,K_1)\ge 2\delta$, we have $g(t)>3\delta/2$
  for infinitely many $t$. A real sequence with vanishing increments that lies
  below $\delta/2$ and above $3\delta/2$ infinitely often must take values in the
  band $[\delta/2,3\delta/2]$ at arbitrarily large times; in particular, there are
  times $s_k\to\infty$ with $|g(s_k)-\delta|\to 0$. By compactness, pass to a
  convergent subsequence $U_{s_k}\to U^\star$; then $U^\star\in\Omega_{\omega}$ and
  $d(U^\star,K_1)=\lim_k g(s_k)=\delta$. But $U^\star\in\Omega_{\omega}=K_1\sqcup K_2$ forces
  $d(U^\star,K_1)=0$ (if $U^\star\in K_1$) or $d(U^\star,K_1)\ge 2\delta$ (if
  $U^\star\in K_2$), contradicting $d(U^\star,K_1)=\delta\in(0,2\delta)$. Hence
  $\Omega_{\omega}$ is connected.
\end{proof}

\begin{proposition}[Connected subsets of countable sets are singletons]
  \label{lem:connected_countable}
  A connected subset $C$ of a countable metric space has at most one point.
\end{proposition}
\begin{proof}
  This is the standard fact that every countable metric space is totally
  disconnected; see, e.g., \citet[Sections~23--24]{Munkres2000Topology}.
\end{proof}

 Lemma~\ref{lem:keyTech} and Propositions~\ref{lem:omega_connected} and \ref{lem:connected_countable} imply the following corollary. 
\begin{corollary}\label{cor:single_fp}
  Assume $\FPone$ is countable. Then there exists $Z_\infty\in\FPone$ such that
  $Z_t\to Z_\infty$. Writing $Z_\infty=(x_\infty,y_\infty,x_\infty,y_\infty)$, the
  payoff $Ay_\infty$ is constant on $S_x:=\supp(x_\infty)$ and $A^\top x_\infty$
  is constant on $S_y:=\supp(y_\infty)$.
\end{corollary}
\begin{proof}
  By Lemma~\ref{lem:keyTech}, Part~3, $\|Z_{t+1}-Z_t\|\to 0$, and
  $\prodSimplex$ is compact, so by Proposition~\ref{lem:omega_connected} the
  $\omega$-limit set $\Omega_{\omega}$ is nonempty, compact, and connected.

  Note that every point of $\Omega_{\omega}$ is a fixed point of $\Phi_1$. If $Z_{t_k}\to Z'$ then,
  using continuity of $\Phi_1$ and $\|Z_{t_k+1}-Z_{t_k}\|\to 0$,
  \[
    Z'=\lim_k Z_{t_k+1}=\lim_k \Phi_1(Z_{t_k})=\Phi_1(Z') .
  \]
  Hence $\Omega_{\omega}\subseteq\FPone$. Since $\FPone$ is countable, $\Omega_{\omega}$ is a
  connected subset of a countable set, so by
  Proposition~\ref{lem:connected_countable}, $\Omega_{\omega}$ is a single point,
  $\Omega_{\omega}=\{Z_\infty\}$ with $Z_\infty\in\FPone$.

  Finally, we observe that a sequence in a compact space whose $\omega$-limit set is the single
  point $Z_\infty$ converges to $Z_\infty$. If not, there would exist
  $\varepsilon>0$ and a subsequence with $\|Z_{t_k}-Z_\infty\|\ge\varepsilon$,
  which by compactness has a further subsequence converging to some
  $Z''\in\Omega_{\omega}$ with $Z''\neq Z_\infty$, contradicting $\Omega_{\omega}=\{Z_\infty\}$.
  Thus $Z_t\to Z_\infty$. The stated form and support conditions are
  Theorem~\ref{thm:FPPhi}.
\end{proof}

\begin{remark}
  We only use the total disconnectedness of $\FPone$.
  Proposition~\ref{lem:connected_countable} is applied to the connected set $\Omega_{\omega}$,
  and the conclusion $|\Omega_{\omega}|\leq 1$ holds whenever $\FPone$ contains no connected
  subset of positive diameter. Countability and finiteness imply total disconnectedness. While the converse generally does not hold, it does in our setting since $\FPone$ is a semialgebraic set. For transparency, we therefore use the strongest assumption, that is, the finiteness of the fixed point set.   \end{remark}
  
  The following lemma is an essential technical tool to show convergence to a Nash equilibrium. We 
  do this by showing that there is no off-support action with strictly better payoff. Recall the characterization of
   fixed points from Theorem~\ref{thm:FPPhi} and the example from Remark~\ref{rem:strictInclusion}. The following lemma rules out exactly these cases.
   
  \subsubsection{The limit is a Nash equilibrium}

By Corollary~\ref{cor:single_fp} it remains to show that the limit $Z_\infty$
corresponds to a Nash equilibrium.

\begin{proof}[Proof of Theorem~\ref{thm:NEconvergenceZSG}]
  By Corollary~\ref{cor:single_fp}, $Z_t\to Z_\infty$ with
  $Z_\infty=(x_\infty,y_\infty,x_\infty,y_\infty)\in\FPone$, and $Ay_\infty$ is
  constant on $S_x=\supp(x_\infty)$ while $A^\top x_\infty$ is constant on
  $S_y=\supp(y_\infty)$. It remains to show $(x_\infty,y_\infty)$ is a Nash
  equilibrium.

  \paragraph{Reduction to a profitable deviation.}
  Let ${\vx}:=(Ay_\infty)_i$ for any $i\in S_x$; then $x_\infty^\top Ay_\infty={\vx}$.
  Since $\max_{x\in\Delta_{d_x}}x^\top Ay_\infty=\max_i (Ay_\infty)_i$, the
  strategy $x_\infty$ is a best response to $y_\infty$ iff $(Ay_\infty)_i\le {\vx}$
  for all $i$. Symmetrically, with ${\vy}:=(A^\top x_\infty)_l$ for $l\in S_y$ (note
  $B=-A^\top$, so the $y$-player maximizes $y^\top Bx=-y^\top A^\top x$,
  equivalently minimizes $y^\top A^\top x$), $y_\infty$ is a best response to
  $x_\infty$ iff $(A^\top x_\infty)_k\ge {\vy}$ for all $k$. Consequently, if
  $(x_\infty,y_\infty)$ is \emph{not} a Nash equilibrium, then at least one of the
  following holds:
  \begin{itemize}
    \item[\textnormal{(I)}] there exist $\hat\imath$ and $j\in S_x$ with
      $(Ay_\infty)_{\hat\imath}>(Ay_\infty)_j={\vx}$; since $(Ay_\infty)_i={\vx}$ for
      $i\in S_x$, necessarily $\hat\imath\notin S_x$;
    \item[\textnormal{(II)}] there exist $\hat\jmath$ and $l\in S_y$ with
      $(A^\top x_\infty)_{\hat\jmath}<(A^\top x_\infty)_l={\vy}$; necessarily
      $\hat\jmath\notin S_y$.
  \end{itemize}
 We only show case \textnormal{(I)} since case \textnormal{(II)} follows exactly the same argument.

     Define $R_t:=x_{t,\hat\imath}/x_{t,j}$. This ratio is well defined and strictly positive for
  every $t$ by Proposition~\ref{prop:invariance} and since $Z_0 \in \relint \prodSimplex$ by assumption. Since $\hat\imath\notin S_x$, $j\in
  S_x$, and $(Ay_\infty)_{\hat\imath}>(Ay_\infty)_j$,
  Lemma~\ref{lem:delta_epsilon_tau}, applied with ${\widetilde Z = Z_\infty}$, provides $\delta,\epsilon,\tau>0$ such that,
  for all $t\in\NN$,
  \begin{equation}\label{eq:F4growth}
    \|Z_t-Z_\infty\|<\delta
    \quad\Longrightarrow\quad
    R_{t+1}\ge e^{\epsilon}R_t,
  \end{equation}
  and, for all $Z=[x,y,x',y']$,
  \begin{equation}\label{eq:F4bound}
    \|Z-Z_\infty\|<\delta
    \quad\Longrightarrow\quad
    \frac{x_{\hat\imath}}{x_{j}}<\tau .
  \end{equation}
  By Corollary~\ref{cor:single_fp}, $Z_t\to Z_\infty$. Thus, there exists
  $T_\star\in\NN$ with
  \begin{equation}\label{eq:tail_in_ball}
    \|Z_t-Z_\infty\|<\delta\qquad\text{for all }t\ge T_\star .
  \end{equation}
  Convergence of the whole sequence ensures that \emph{every} iterate with
  $t\ge T_\star$ is in $B_\delta(Z_\infty)$. Hence
  \eqref{eq:F4growth} applies at every $t\ge T_\star$, and by induction
  \begin{equation}\label{eq:Rt_lower}
    R_t\ \ge\ e^{\epsilon(t-T_\star)}\,R_{T_\star}
    \qquad(t\ge T_\star).
  \end{equation}
  Since $R_{T_\star}>0$ and $\epsilon>0$, the right-hand side of
  \eqref{eq:Rt_lower} diverges, so $R_t\to+\infty$. On the other hand,
  \eqref{eq:tail_in_ball} and \eqref{eq:F4bound} give the uniform bound
  \begin{equation}\label{eq:Rt_upper}
    R_t=\frac{x_{t,\hat\imath}}{x_{t,j}}<\tau
    \qquad(t\ge T_\star).
  \end{equation}
  Inequalities \eqref{eq:Rt_lower} and \eqref{eq:Rt_upper} are incompatible. This excludes Case~\textnormal{(I)}.

  We conclude that neither deviation can occur, so $x_\infty$ is a best response to $y_\infty$ and
  $y_\infty$ is a best response to $x_\infty$. With the support-equalization of
  Theorem~\ref{thm:FPPhi}, $(x_\infty,y_\infty)$ is a Nash equilibrium, i.e.\
  $Z_\infty\in\NE$. By Corollary~\ref{cor:single_fp}, $Z_t\to Z_\infty\in\NE$.
\end{proof}

 \section{Omitted Proofs of Section \ref{sec:local}}\label{appendix:local}
 Our results in this section build on known techniques from discrete dynamical systems theory. See, e.g.,  \citet{Kuznetsov:2026aa}.
 
 \subsection{The Jacobian}\label{app:jac_spectrum}
We provide some standard computations for the convenience of the reader.
Let $\tZ = \Zfp \in \FP$ be a fixed point of ${\Phi_m}$.
 First recall that
  \[
  \V_1 = \Exp\left( \eta_x ( A \ty - \vx \ones ) \right), \quad
  \V_2 = \Exp\left( \eta_y ( B \tx - \vy \ones ) \right)\, .
\]
(Recall the definition: $\vx := (A\ty)_i$ for any $i\in S_x$ and $\vy := (B\tx)_j$ for any $j\in S_y$.)
\begin{lemma}
For a game $\Gamma(A,B)$, the Jacobian matrix $J_{\Phi_m}(\tZ) \in \RR^{(2(d_x+d_y)) \times (2(d_x+d_y))}$ with $\tZ \in \FP$ is:
\[
J_{\Phi_m}(\tZ) = 
\begin{bmatrix}
  \diag(\V_1) - \tx \V_1^\top & (1+m) \etax H(\tx) A & 0 & - m\etax H(\tx) A
  \\
  (1+m) \etay H(\ty)B & \diag(\V_2) - \ty \V_2^\top & - m \etay H(\ty) B & 0 \\
  I_{d_x} & 0 & 0 & 0 \\
  0 & I_{d_y} & 0 & 0
\end{bmatrix}\, .
\]
\end{lemma}
\begin{proof} 
 Using the block structure of $\Phi_m = [(\Phi_m)_1,(\Phi_m)_2,(\Phi_m)_3,(\Phi_m)_4]$, where $(\Phi_m)_k:\RR^{2(d_x + d_y)} \rightarrow \RR^{d_x}$ for $k = 1,3$ and $(\Phi_m)_k:\RR^{2(d_x + d_y)} \rightarrow \RR^{d_y}$ for $k= 2,4$,
  we compute the Jacobian block-wise $J_{1,1} = \frac{\partial
{(\Phi_m})_1}{\partial \px}$, $J_{1,2} = \frac{\partial {(\Phi_m})_1}{\partial \py}$ and $J_{1,4} =
\frac{\partial {(\Phi_m})_1}{\partial \ppy}$. The blocks $J_{2,1},J_{2,2}$ and $J_{2,3}$ follow analogously.
Since for any $i, j$ and for $u \in \RR^d \setminus \{ {v} \in \RR^d : \ones^\top {v} = 0\}$  , we have
\[
\frac{\partial (P_d)_i}{\partial u_j}
= 
\left\{
\begin{split}
    & - \frac{u_i}{\dprod{\ones,u}^2} & \text{if } i \neq j \\  
    & \frac{1}{ \dprod{\ones,u}} -  \frac{u_i}{ \dprod{\ones,u}^2} & \text{if } i = j 
\end{split}
\right.
\,,
\]
the Jacobian for $P_d$ is
 \begin{align}\label{id:Jacob}
  J_{P_d}(u) = \frac{1}{ \dprod{\ones,u}} \big( I_d - P_d(u) \mathbf 1^\top  \big)
\end{align}
 Now let $ ({\Phi^{\mathrm{u}}_m})_1(\px, \py, \ppx, \ppy) 
 = \px \odot \Exp( \etax  A((1+m)\py - m\ppy))$ denote the first component of the unnormalized exponential weight updates.
Then,  
\[
  J_{1, 1} = \frac{\partial {(\Phi^{\mathrm{u}}_m})_1}{\partial \px} = J_{P_{d_x}}( ({\Phi^{\mathrm{u}}_m})_1(\px, \py, \ppx, \ppy))\diag\left(   \Exp( \etax  A((1+m)\py - m\ppy)) \right)\, .
\]
 Evaluating at the fixed point $\tZ$, where $\py = \ppy = \ty$, we have $ \Exp( \etax  A((1+m)\py - m\ppy)) 
 =  \Exp(\etax A \ty)$. 
 Thus, by \eqref{id:Jacob}
\begin{align*}
  J_{1, 1}
  & = 
 \frac{1}{\dprod{\tx , \Exp(\etax A \ty)}} 
  \big(
    \id_{d_x}
    - P_{d_x}\Big(\tx \odot \Exp(\etax A \ty)\Big) \mathbf{1}^\top
    \big)
  \diag(\Exp(\etax A \ty)) \\
  &\overset{(\star)}{ = }
 \frac{1}{\dprod{\tx, \Exp(\etax A \ty)}} 
  \big(
    \id_{d_x}
    - \tx \mathbf{1}^\top
    \big)
  \diag(\Exp(\etax A \ty)) \\
  & = 
 \frac{1}{\dprod{\tx ,\Exp(\etax A \ty)}} 
  \big(
    \diag(\Exp(\etax A \ty))
    - \tx \Exp(\etax A \ty)^\top 
    \big) \\
  & = 
    \diag(\Exp(\etax (A \ty - \vx \mathbf 1)))
    - \tx \Exp(\etax (A \ty - \vx \mathbf 1))^\top 
    \,.
\end{align*}
We used the fixed-point equation in the final line
\[
 \frac{1}{\dprod{\tx , \Exp(\etax A \ty)}} 
 = e^{-\etax\vx}
 \,,
\]
 and consequently also $P_{d_x}(\tx \odot \Exp(\etax A \ty)) = \tx$ in $(\star)$. 
 The computation of the block $J_{1,2}$ follows similar arguments. We have,
 \begin{align*}
    J_{1,2}
    & =  J_{P_{d_x}}(\tx \odot \Exp(\etax A \ty)) 
    \diag(\tx)  \diag(\Exp ( \etax  A \ty))
    (1+m)\etax A
     \\
  & = 
 \frac{1+m}{\dprod{\tx , \Exp(\etax A \ty)}} 
  \big(
    \id_{d_x}
    - P_{d_x}(\tx \odot \Exp(\etax A \ty)) \mathbf{1}^\top
    \big)
    \diag(\tx)
    \diag(\Exp ( \etax  A \ty ))
    \etax A
     \\
  & = 
 \frac{1+m}{\dprod{\tx , \Exp(\etax A \ty)}} 
  \big(
    \id_{d_x}
    - \tx \mathbf{1}^\top
    \big)
    \diag\big(\tx \odot \Exp(\etax  A \ty )\big)
    \etax A 
     \\
  & = 
  (1+m)\big(
    \id_{d_x}
    - \tx \mathbf{1}^\top
    \big)
    \diag\big(P_{d_x}(\tx \odot \Exp(\etax  A \ty ))\big)
    \etax A  \\
  & = 
  (1+m)
  \big(
    \id_{d_x}
    - \tx \mathbf{1}^\top
    \big)
    \diag\big(\tx)
    \etax A  \\
  & = 
  (1+m)
  \big(
    \diag\big(\tx)
    - \tx \tx^\top
    \big)
    \etax A 
     \,.
\end{align*}
The proof is completed by noting that $J_{1,4}$ follows by the same arguments.  
\end{proof}

 \subsection{Eigenvalues of the Jacobian and Proof of Theorem~\ref{thm:jac_spectrum}}
 
 \paragraph{Notation: Restrictions, Complexification and Eigenvalues}
 Recall that for a real linear map $T:E\to E$ we denote the spectrum, the spectrum without multiplicities, and the spectral radius of the complexified restriction by $\eig{T|_V},\eigTilde{T|_V}$ and $\SR{T|_V}$. To make the complexification transparent in this section, we introduce some notation. Thus, let $E$ be a real vector space and denote its complexification by $E_{\CC} := E \otimes_{\RR} \CC$.
Given a linear map $M : E \to E$ we denote by $M_{\CC} : E_{\CC} \to E_{\CC}$ the complex(-linear)
extension defined by $M_{\CC}(u + iv) = Mu + iMv$ for all $u, v \in E$.
Let $V \subset E$ be a linear subspace such that $M (V) \subseteq V $.
Then $\restr{M}{V}$ induces an endomorphism on $V$
and $ \left(\restr{M}{V}\right)_{\CC} = \restr{(M_{\CC})}{V_{\CC}} $.
 We denote both by ${M}_{V, \CC}$.
Moreover, the eigenvectors of $M_{V, \CC}$ are exactly the eigenvectors $z$ of $M_{\CC}$
such that $z \in V_{\CC}$.
That is, complex eigenvalues of $\restr{J_{\Phi_m}(\tZ)}{L}$ are complex numbers such 
that there exists $\mathbf h \in L_{\CC} (=  L + i L )$ for which $ J_{\Phi_m}(\tZ) \mathbf h = \lambda
\mathbf h$. 
The proof of Theorem~\ref{thm:jac_spectrum} builds on the following two results.
\begin{lemma}[Spectrum of $J_{\Phi_m}(\tZ)|_{L}$]
  \label{lem:detailed_spectrum_1}
   The complex $\lambda \in \CC^\star$ is an eigenvalue of $J_{\Phi_m}(\tZ)|_L$, 
    if and only if one of the following holds
    \[
          \lambda 
          \in
          \{ 
            \V_{1, i} \mid i \in \{ 1,\dots, d_x \}
          \setminus S_x
          \}
          \cup
          \{ 
            \V_{2, j} \mid j \in \{ 1, \dots, d_y \} 
          \setminus S_y
          \}
    \]
    or $\lambda \neq \frac{m}{m+1}$ and
    \begin{align}\label{eq:linSys}
      \exists [h_1, h_2, \lambda^{-1} h_1, \lambda^{-1} h_2] \in L_{\CC}\setminus\{0\}
      \quad \text{s.t.} 
        \begin{cases}
          &\etax H(\tx) A h_2 = \frac{\lambda(\lambda - 1)}{(1+m)\lambda - m} h_1 \\[1em]
          &\etay H(\ty) B h_1 = \frac{\lambda(\lambda - 1)}{(1+m)\lambda - m} h_2  \\[1em]
          &\forall i \notin S_x , \;
      \forall j \notin S_y , \; 
      h_{1, i} = h_{2, j} = 0\, . \tag{LS}
        \end{cases}  
    \end{align}
\end{lemma}

\begin{proof}[Proof of Lemma~\ref{lem:detailed_spectrum_1}]
  The complex $\lambda \in \CC^\star$ is an eigenvalue of $J_{\Phi_m}(\tZ)|_L$, if and only if
  there is some $\mathbf h = [h_1, h_2, h_3, h_4] \in L_{\CC}$ such that $J_{\Phi_m}(\tZ) \mathbf
  h = \lambda \mathbf h$.

  Assuming $\lambda \neq 0$, the last two block rows imply $h_3 = \lambda^{-1} h_1$ and
  $h_4 = \lambda^{-1} h_2$. Substituting these into the first two rows yields the reduced
  system:
  \begin{align}
    (\diag(\V_1) - \tx \V_1^\top) h_1 + \etax H(\tx) A (1 + m - m\lambda^{-1}) h_2 
    &= \lambda h_1 \label{eq:sys1} \\
    (\diag(\V_2) - \ty \V_2^\top) h_2 + \etay H(\ty) B (1 + m - m\lambda^{-1}) h_1 
    &= \lambda h_2 \label{eq:sys2}
  \end{align}
  We then proceed by a case disjunction: 
  \paragraph{Case 1:
  $h_{1, k} \neq 0$ for some $k \notin S_x$
  }
  If there exists $k \notin S_x$ such that $h_{1, k} \neq 0$, then for such $k$, we have
  $\tx_{k} = 0$, so the $k$-th row of $H(\tx)$ is zero, and the $k$-th row of $\tx
  \V_1^\top$ is zero. Equation~\eqref{eq:sys1} for the $k$-th component becomes:
  \[
    \V_{1, k} \, h_{1,k} + 0 = \lambda\, h_{1,k} \,,
  \]
  therefore $\lambda = \V_{1, k}$, since $h_{1, k} \neq 0$.

  Conversely, for any $k \notin S_x$, consider the left-eigenvector
  $
    \mathbf v  = (e_k, 0 , 0, 0)
  $
  then 
  \[
    \mathbf v^\top J_{\Phi_m}(\tZ)
    = (e_k^\top, 0 , 0, 0)J_{\Phi_m}(\tZ)  
    = (\V_{1,k}e_k^\top, 0 , 0 , 0)
    = \V_{1,k} \mathbf v^\top
    \,.
  \]
  So $\V_{1, k}$ is a non-zero eigenvalue of $J_{\Phi_m}(\tZ)$.  It remains to show that the corresponding right eigenvector lies in $L_{\CC}$. If $[h_1, h_2, \lambda^{-1} h_1, \lambda^{-1}h_2]$ is a (right-)eigenvector
  with eigenvalue $\V_{1, k} > 0$, then necessarily, due to Equation~\eqref{eq:sys1},
  \[
    \lambda \mathbf 1^\top h_1 
    = 
    \mathbf 1^\top (\diag(\V_1) - \tx \V_1^\top) h_1 
    + \etax \mathbf 1^\top H(\tx) A (1 + m - m\lambda^{-1}) h_2 
    \, .
  \]
  Now
  \[
    \mathbf 1^\top (\diag(\V_1) - \tx \V_1^\top) 
    = 
    \mathbf 1^\top \diag(\V_1) - \mathbf 1^\top \tx \V_1^\top
    =\V_1^\top  - \V_1^\top = 0  
    \, ,
  \]
  and similarly
  \[
    \mathbf 1^\top H(\tx) = 0 \,.
  \]
  So $\mathbf 1^\top h_1 = 0$. Symmetrically, we also have $\mathbf 1^\top h_2 = 0$, 
  so $[h_1, h_2, \lambda^{-1}h_1, \lambda^{-1}h_2] \in L_{\CC}$.

\paragraph{Case 2: 
  $h_{2, \ell} \neq 0$ for some $\ell \notin S_y$
  }
  Similarly, if there exists $\ell \notin S_y$ such that $h_{2, \ell} \neq 0$, then
  necessarily $\lambda = \V_{2, \ell}$, and conversely these are indeed 
  eigenvalues of $J_{\Phi_m}(\tZ)|_{L}$.

  \paragraph{Case 3}

  If neither Case 1 nor Case 2 holds, then $h_1$ is supported on $S_x$ and $h_2$ on $S_y$.
  On the support $S_x$ (resp. $S_y$), all coordinates of the vector $\V_1$ (resp.
  $\V_2$) are $1$.
  The term $(\diag(\V_1) - \tx \V_1^\top) h_1$ becomes $(I - \tx \ones^\top) h_1$.
  Since $\ones^\top h_1 = 0$, this simplifies to $h_1$.
   Equations~\eqref{eq:sys1} and~\eqref{eq:sys2} become:
  \begin{align*}
    h_1 + (1 + m - m\lambda^{-1}) \etax H(\tx) A h_2 
    &= \lambda h_1 
    \quad \text{i.e.}\quad 
    (\lambda - 1) h_1 
    = (1 + m - m\lambda^{-1}) \etax H(\tx) A h_2 \\
    h_2 + (1 + m  - m\lambda^{-1}) \etay H(\ty) B h_1 
    &= \lambda h_2 
    \quad \text{i.e.} \quad 
    (\lambda - 1) h_2 = (1 + m - m\lambda^{-1}) \etay H(\ty) B h_1
    \,.
  \end{align*}
  To rewrite these in the form \eqref{eq:linSys}, divide by
$1+m-m\lambda^{-1} = ((1+m)\lambda-m)/\lambda$, which is nonzero if $\lambda \neq m/(m+1)$. If instead $\lambda = m/(m+1)$, then
the equations simplify to
$(\lambda-1)h_1 = 0$, $(\lambda-1)h_2 = 0$. Since $\lambda=m/(m+1)\neq 1$,
this forces $h_1=h_2=0$. Hence $m/(m+1)$ never arises as an eigenvalue
through Case~3, and can occur only via Cases~1 or 2, i.e.\ as an element of $W$.

  Conversely, we check that any pair $(\lambda, \mathbf h) \in \CC^\star \times L_{\CC}\setminus\{0\}$
  satisfying these conditions is indeed an eigenvalue-eigenvector pair for
  $J_{\Phi_m}(\tZ)|_L$.
  Assume first that $\lambda \neq \frac{m}{m+1}$ and that  \eqref{eq:linSys} holds. Define
\[
    \mathbf h
    :=
    \bigl[h_1,h_2,\lambda^{-1}h_1,\lambda^{-1}h_2\bigr]
    \in L_{\mathbb C}\, .
\]
Since $[h_1,h_2]\neq [0,0]$, we have $\mathbf h\neq 0$. Moreover, $h_1,h_2$ are
supported on $S_x,S_y$, respectively, and $\V_1=\ones$ on $S_x$,
$\V_2=\ones$ on $S_y$. Hence, using
$\ones^\top h_1=\ones^\top h_2=0$,
\[
    \bigl(\diag(\V_1)-\tx\V_1^\top\bigr)h_1=h_1\, ,
    \qquad
    \bigl(\diag(\V_2)-\ty\V_2^\top\bigr)h_2=h_2 \,.
\]
Also,
\[
    1+m-m\lambda^{-1}
    =
    \frac{(1+m)\lambda-m}{\lambda}\,.
\]
Therefore, by  \eqref{eq:linSys},
\[
    \bigl(1+m-m\lambda^{-1}\bigr)\eta_x H(\tx)Ah_2
    =
    (\lambda-1)h_1\,,
\]
and similarly
\[
    \bigl(1+m-m\lambda^{-1}\bigr)\eta_y H(\ty)Bh_1
    =
    (\lambda-1)h_2\,.
\]
Substituting these identities into the first two block rows of
$J_{\Phi_m}(\tZ)\mathbf h$, and using
$h_3=\lambda^{-1}h_1$, $h_4=\lambda^{-1}h_2$ for the last two block rows, gives
\[
    J_{\Phi_m}(\tZ)\mathbf h=\lambda \mathbf h\, .
\]
Thus $\lambda\in \eig{J_{\Phi_m}(\tZ)|_L}$.

It remains to consider the case $\lambda\in W$. This is exactly the off-support
case treated above: the corresponding off-support coordinate gives a nonzero
eigenvector of $J_{\Phi_m}(\tZ)$, and the summation argument shows that this
eigenvector lies in $L_{\CC}$. Hence
\[
    \lambda\in \eig{J_{\Phi_m}(\tZ)|_L}\,.
\]
\end{proof}

\begin{lemma}\label{lem:equiv_double_equation}
Suppose $n_x \geq n_y$.
  For any $\lambda \in \CC^\star \setminus \{m / (m+1)\}$, we have the following equivalence
   \begin{align} \label{eq:linSys1}
      \exists [h_1, h_2, \lambda^{-1} h_1, \lambda^{-1} h_2] \in L_{\CC}\setminus\{0\}
      \quad \text{s.t.} 
        \begin{cases}
          &\etax H(\tx) A h_2 = \frac{\lambda(\lambda - 1)}{(1+m)\lambda - m} h_1 \\[1em]
          &\etay H(\ty) B h_1 = \frac{\lambda(\lambda - 1)}{(1+m)\lambda - m} h_2  \\[1em]
          &\forall i \notin S_x , \;
      \forall j \notin S_y , \; 
      h_{1, i} = h_{2, j} = 0\, . \tag{LS}
        \end{cases}  
    \end{align}
  if and only if 
  \[
      Q_m(\lambda) \in 
    \eig{
    \etax \etay 
    \Pi_{S_x} (H(\tx) A H(\ty) B) \Pi_{S_x}^{\top} |_{L_x}
    }
    \, .
  \]
\end{lemma}

\begin{proof}[Proof of Lemma~\ref{lem:equiv_double_equation}]
We prove the two implications separately and isolate the case $\lambda=1$.

\medskip
\noindent
$(i) \Rightarrow (ii)$, first assume $\lambda\neq 1$.
Let $[h_1, h_2, \lambda^{-1}h_1, \lambda^{-1} h_2] \in L_{\CC}\setminus\{0\}$ 
satisfying the condition of the statement. 
Then since $\mathbf 1^\top h_1 = 0$,
and since $h_{1, i} = 0$ for all $i \notin S_x$, we have
\[
  \sum_{i \in S_x} h_{1, i} 
  = 
  \sum_{i =1}^{d_x} h_{1, i} 
  = 0 \,.
\]
so $\Pi_{S_x} h_1 \in L_{x, \CC}$. Moreover, the eigenvalue equations imply that 
$h_1 \neq 0$, so $\Pi_{S_x} h_1 \neq 0$.
Finally, since $h_1$ is supported on $S_x$, we have
$
  \Pi_{S_x}^{\top} \Pi_{S_x} h_1 = h_1
$
so
\begin{multline*}
  \etax \etay \Pi_{S_x} (H(\tx) A H(\ty) B) \Pi_{S_x}^{\top} \Pi_{S_x} h_1
  = 
  \etax \etay \Pi_{S_x} (H(\tx) A H(\ty) B) h_1 \\
  = 
  \etax \Pi_{S_x} H(\tx) A 
  \bigg(
    \frac{\lambda(\lambda - 1)}{(1+m)\lambda - m} h_2
  \bigg)
 = 
  \bigg(\frac{\lambda(\lambda - 1)}{(1+m)\lambda - m} \bigg)^2
  \Pi_{S_x}h_1 \,.
\end{multline*}
Therefore $\Pi_{S_x} h_1\in L_{x,\CC}$ is an eigenvector of $\etax \etay \restr{M_x(\tZ)}{L_x}$ with eigenvalue
$Q_m(\lambda)$.

\medskip
\noindent
$(i) \Rightarrow (ii)$, now consider $\lambda=1$.
Then $\mu=0$ and $Q_m(1)=0$.
If $h_1\neq0$, then $\Pi_{S_x}h_1\in L_{x,\CC}\setminus\{0\}$ and
\[
H(\ty)Bh_1=0,
\]
so
\[
\Pi_{S_x}H(\tx)AH(\ty)B\Pi_{S_x}^{\top}\Pi_{S_x}h_1=0.
\]
Thus $0=Q_m(1)$ is an eigenvalue of 
$\Pi_{S_x}H(\tx)AH(\ty)B\Pi_{S_x}^{\top}|_{L_x}$.

Suppose instead that $h_1=0$. Since $[h_1,h_2]\neq [0,0]$, we have
$h_2\neq0$. Moreover $h_2$ is supported on $S_y$ and belongs to
$L_{y}$, and the first equation gives
\[
H(\tx)Ah_2=0.
\]
Define
\[
T:=\Pi_{S_x}H(\tx)A\Pi_{S_y}^{\top}:L_{y,\CC}\to L_{x,\CC}.
\]
Then $\ker T\neq\{0\}$, so
\[
\rank T<\dim L_{y}=n_y-1.
\]
Now write
\[
\Pi_{S_x}H(\tx)AH(\ty)B\Pi_{S_x}^{\top}\big|_{L_{x}}
=
T\circ R
\]
where
\[
R:=\Pi_{S_y}H(\ty)B\Pi_{S_x}^{\top}:L_{x,\CC}\to L_{y,\CC}.
\]
Hence
\[
\rank(T\circ R)\le \rank T < n_y-1.
\]
Using the assumption $n_x\ge n_y$, we get
\[
\rank(T\circ R)< n_y-1\le n_x-1=\dim L_{x}.
\]
Therefore
\[
\Pi_{S_x}H(\tx)AH(\ty)B\Pi_{S_x}^{\top}\big|_{L_{x}}
\]
is singular, and so $0=Q_m(1)$ is an eigenvalue.

\medskip
\noindent
$(ii) \Rightarrow (i)$, first assume $\lambda\neq 1$.
Let
\[
 \mu = \frac{\lambda(\lambda -1)}{(1+m) \lambda - m} 
 \neq 0 
 \,,
\]
then the assumption is that there exist $v_1 \in L_{x, \CC} \setminus \{0\}$ such that 
\[
  \etax \etay \Pi_{S_x} (H(\tx) A H(\ty) B) \Pi_{S_x}^{\top} v_1
  = \mu^2 v_1 \, .
\]
Define the vectors
\[
  h_1 = \Pi_{S_x}^\top v_1 
  \quad \text{and} \quad
  h_2 = \frac{\etay}{\mu} H(\ty) B \Pi_{S_x}^\top v_1 
\]
then 
\[
  \etay H(\ty) B h_1
  = \mu h_2
\]
and
\[
  \etax \Pi_{S_x} H(\tx) A h_2
  = 
  \frac{\etax \etay}{\mu} \Pi_{S_x} H(\tx) AH(\ty) B \Pi_{S_x}^\top v_1 
  = 
  \frac{1}{\mu} \mu^2 v_1  = \mu v_1 \, ,
\]
so 
\[
  \etax H(\tx) A h_2 
  = \etax \Pi_{S_x}^\top \Pi_{S_x} H(\tx) A h_2 
  = \mu \Pi_{S_x}^\top v_1 
  = \mu h_1 \,.
\]

\medskip
\noindent
$(ii) \Rightarrow (i)$, now consider $\lambda=1$.
Since $Q_m(1)=0$, assume that there exists
$v_1\in L_{x,\CC}\setminus\{0\}$ such that
\[
  \Pi_{S_x}H(\tx)AH(\ty)B\Pi_{S_x}^{\top}v_1=0.
\]
Set
\[
  \widehat h_1:=\Pi_{S_x}^{\top}v_1,
  \qquad
  r:=H(\ty)B\widehat h_1 .
\]
Then $r$ is supported on $S_y$ and satisfies $\mathbf 1^\top r=0$, hence
$r\in L_{y,\CC}$. Moreover,
\[
  \Pi_{S_x}H(\tx)Ar
  =
  \Pi_{S_x}H(\tx)AH(\ty)B\Pi_{S_x}^{\top}v_1
  =
  0.
\]
Since $H(\tx)Ar$ is supported on $S_x$, this implies
\[
  H(\tx)Ar=0.
\]
If $r=0$, choose
\[
  h_1=\widehat h_1,
  \qquad
  h_2=0.
\]
Then
\[
  H(\ty)Bh_1=0,
  \qquad
  H(\tx)Ah_2=0.
\]
If $r\neq0$, choose
\[
  h_1=0,
  \qquad
  h_2=r.
\]
Then
\[
  H(\ty)Bh_1=0,
  \qquad
  H(\tx)Ah_2=H(\tx)Ar=0.
\]
In both cases, $[h_1,h_2]\neq[0,0]$, and the vectors have the required support
and tangent-space properties. Hence \eqref{eq:linSys1} holds for $\lambda=1$.
\end{proof}

 We are now ready to prove Theorem~\ref{thm:jac_spectrum}.
\begin{proof}[Proof of Theorem~\ref{thm:jac_spectrum}]
From Lemma~\ref{lem:equiv_double_equation}, for every
$\lambda \in \CC\setminus\{0,\tfrac{m}{m+1}\}$,
\[
  \lambda \in Q_m^{-1}\Big(\eigTilde{\etax \etay \restr{M_x(\tZ)}{L_x}}\Big)
  \iff \eqref{eq:linSys}.
\]
By Lemma~\ref{lem:detailed_spectrum_1}, every
$\lambda \in \eig{J_{\Phi_m}(\tZ)|_L}\setminus\{0,\tfrac{m}{m+1}\}$
either satisfies \eqref{eq:linSys} or lies in $W$. It remains to treat
$\lambda = m/(m+1)$. As the pole of $Q_m$, it is never an element of
$Q_m^{-1}(\mu)$ for any finite $\mu$, so it never belongs to the right-hand
set; and by the Case~3 analysis above, it is an eigenvalue of
$J_{\Phi_m}(\tZ)|_L$ if and only if $m/(m+1)\in W$. Thus adjoining or removing
$\lambda=m/(m+1)$ changes neither side, and
\[
  \eigTilde{J_{\Phi_m}(\tZ)|_{L}}\setminus\{0\}
  = \left(W \cup Q_m^{-1}\Big(\eigTilde{\etax \etay \restr{M_x(\tZ)}{L_x}}\Big)\right)\setminus\{0\},
\]
as claimed.
 
\end{proof}

 \subsection{Stability}\label{app:jac_stab}
 
It is well known that when a fixed point lies in the interior of the domain, the
stability and attraction of the fixed point can be characterized by the spectrum of the Jacobian at the fixed point.
In our case, the state space is a product of simplices, and its interior is empty.
Our techniques rely on stable subsets with respect to a non-empty relative interior.

 \begin{lemma}[Fréchet-Differentials, Jacobians and Restrictions]\label{lem:jac_stable_subspace}
  Let $\cO \subset \RR^d$ be an open subset and $\gz \in \cO$. Let $f : \cO \to \RR^d$ be a $\cC^1$ map.
  Let $V \subset \RR^d$ be a linear subspace and $\delta>0$.
  If $f(\gz + (B_\delta(0) \cap V)) \subset \gz + V$, then
  \[
    D f(\gz)|_{V} = D (f|_{\cO \cap (\gz +V)})(\gz) \,.
  \]
  In particular, $Df(\gz) V \subset V$.
\end{lemma}
\begin{proof}
 We first show that we can assume without loss of generality that $\gz$ is the origin. Define
$
g(u)=f(\gz+u)-\gz,
$
for $u\in \cO-\gz$. Then $g$ is $\cC^1$, $Dg(0)=Df(\gz)$, and
$
g\bigl((B_\delta(0)\cap V)\bigr)\subset V
$.
Thus, it suffices to prove the claim for $\gz=0$.

The set $\cO\cap V$ is open in $V$, so $\restr{f}{\cO\cap V}$ is differentiable
as a map defined on $V$.
For any $v\in V$ and any $t$ small enough, we have $tv\in \cO\cap V$, hence
\[
Df(0)v
=
\lim_{t\to 0}\frac{f(tv)-f(0)}{t}
=
\lim_{t\to 0}\frac{f|_{\cO\cap V}(tv)-f|_{\cO\cap V}(0)}{t}
=
D(f|_{\cO\cap V})(0)v.
\]
Moreover, due to the local invariance assumption, for $t$ small enough, both $f(tv)$ and $f(0)$ lie in $V$, so
\[
\frac{f(tv)-f(0)}{t}\in V.
\]
Passing to the limit gives $Df(0)v\in V$.
\end{proof}

\subsection{Stability and Instability}

Recall that we denote by $\SR{A}$ the spectral radius of the linear map $A$, i.e., the largest
modulus of its complex eigenvalues.
The following results will be an essential tool. Many of them are well-known results, which we include for completeness.

\begin{proposition}[Stability and Linearization of Maps] \label{prop:linearization_stability}
Consider a $\cC^1$ map $g:\RR^n \to \RR^n$ with $g(0) = 0$. Then $0$ is
\begin{enumerate}
\item  asymptotically stable if $\SR{D g(0)} <1$; and
\item unstable if $\SR{D g(0)} >1$.
\end{enumerate}
\end{proposition}
For a proof, see, e.g., Theorem 3.2 in \citet{Kuznetsov:2026aa}.

  \begin{theorem}\label{thm:stable_linear}
  Let $V \subset \RR^d$ be a linear subspace, and $\cO$ an open subset of $\RR^d$.
  Let $F : \cO \to \cO$ be a $\cC^1$ map and
  $\gz \in \cO$ be a fixed point of $F$. If there exists $\delta
  > 0$ such that
  \[
    F\big( \gz + (V \cap B_\delta(0))\big) \subset  \gz + V \,,
  \]
   then $\gz$ is asymptotically stable relative to $\gz + V$
   if $\SR{\restr{DF(\gz)}{V} } < 1$
  and unstable relative to $\gz + V$ if $\SR{\restr{DF(\gz)}{V}} > 1$.
\end{theorem}
\begin{remark}
The restriction to \emph{relative to $\gz+V$} is essential. Note that if a fixed point is stable relative to a set $A$, then it is stable relative to every (forward invariant) subset $B \subset A$ containing the fixed point. Conversely, if it is unstable relative to such a subset  $B$, then it is unstable relative to $A$.
\end{remark}

\begin{proof}
  We assume again without loss of generality that $\gz = 0$ (otherwise consider $G(u)
  = F(\gz + u) - \gz$). We also denote $B_\delta = B_\delta(0)$.

  Since
  $
    F(V\cap B_\delta)\subset V,
  $
  the restriction
  $
    f:=\restr{F}{V\cap B_\delta} : V\cap B_\delta \to V
  $
  is well defined and by Lemma~\ref{lem:jac_stable_subspace},
  \[
    Df(0)=\restr{DF(0)}{V}.
  \]
  Let $k=\dim V$, and let $T_c:\RR^k\to V$ be a linear isomorphism. Define
  \[
    \widetilde F := T_c^{-1}\circ f\circ T_c
  \]
  on the open set $T_c^{-1}(V\cap B_\delta)\subset \RR^k$. Then $\widetilde F$
  is of class $\cC^1$, satisfies $\widetilde F(0)=0$, and
  \[
    D\widetilde F(0)
    =T_c^{-1}\circ Df(0)\circ T_c
    =T_c^{-1}\circ \restr{DF(0)}{V}\circ T_c.
  \]
  Hence $D\widetilde F(0)$ is similar to $\restr{DF(0)}{V}$, so they have the
  same spectrum and therefore the same spectral radius.

  Since $T_c$ is a homeomorphism, Lyapunov stability, attraction, and
instability are preserved under this change of coordinates (see, e.g., Section 3 \cite{Kuznetsov:2026aa}).
  Therefore, by Proposition~\ref{prop:linearization_stability}, the fixed point $0$ is asymptotically stable relative to $V$ when
  $\SR{\restr{DF(0)}{V}}<1$, and unstable relative to $V$ when
  $\SR{\restr{DF(0)}{V}}>1$.
\end{proof}

We note that Theorem~\ref{thm:stable_linear} does not imply instability in the presence
of constraint sets (only instability in $x + V$).
Since our problem is inherently constrained, instability must be transferred to the
constraint set itself. 

 For a fixed point $\tZ \in \FP$ with supports
$S_x,S_y$, define the face
\[
  {\cF_{\tZ}}:=\left\{[\px,\py,\ppx,\ppy]\in\prodSimplex: \begin{cases} \supp(\px)\subseteq S_x\\\supp(\ppx)\subseteq S_x \end{cases}\
 \begin{cases} \supp(\py)\subseteq S_y\\ \supp(\ppy)\subseteq S_y \end{cases}\right\} 
\]
and the face-tangent (lineality) subspace
\[
  \Lpar:=\{\,h=[h_1,h_2,h_3,h_4]\in L:\ h_{k,i}=0\ \text{for }i\notin S_x\,(k{=}1,3),\
  \text{and }j\notin S_y\,(k{=}2,4)\,\}\, .
\]
Note that $\tZ \in \relint {\cF_{\tZ}}$.
\begin{lemma}[Ambient invariance of the face-tangent affine space]\label{cor:face_inv}
There exists a $\delta>0$ with $\tZ+(\Lpar\cap B_\delta(0))\subseteq\cD_{{{\Phi}_m}}$ and
\[
  {{\Phi}_m}\big(\tZ+(\Lpar\cap B_\delta(0))\big)\subseteq \tZ+\Lpar .
\]
Consequently ${J_{\Phi_m}}(\tZ)\,\Lpar\subseteq\Lpar$ and
$D\big(\restr{{{\Phi}_m}}{\tZ+\Lpar}\big)(\tZ)=\restr{{J_{\Phi_m}}(\tZ)}{\Lpar}$.
\end{lemma}
\begin{proof}
Choose $\delta>0$ small enough that, for every $h\in \Lpar\cap B_\delta(0)$,
all on-support coordinates of $\tZ+h$ are positive. Since $h\in\Lpar$, the
off-support coordinates of $\tZ+h$ are zero, and the four block sums are still
equal to $1$. Thus $\tZ+h\in \cD_{\Phi_m}$ for $\delta$ small enough.

Let $U:=\tZ+h$. By Proposition~\ref{prop:invariance}, the first two blocks of
$\Phi_m(U)$ have no mass outside $S_x$ and $S_y$, respectively. The last
two blocks of $\Phi_m(U)$ are the first two blocks of $U$, and hence are also
supported on $S_x$ and $S_y$. Moreover, the first two blocks are normalized
by $P_{d_x}$ and $P_{d_y}$, while the last two blocks already have coordinate
sum $1$. Hence, all four blocks of $\Phi_m(U)$ have the same support
constraints and affine-sum constraints as $\tZ+\Lpar$. Therefore
\[
  \Phi_m\bigl(\tZ+(\Lpar\cap B_\delta(0))\bigr)\subseteq \tZ+\Lpar .
\]

The final assertion follows from Lemma~\ref{lem:jac_stable_subspace} applied
with $\gz=\tZ$, $V=\Lpar$, and $f=\Phi_m$.
\end{proof}
 
\begin{theorem}[Constrained instability of optEW]\label{thm:constraint_unstable}
Let $\tZ\in\prodSimplex$ be a fixed point of ${{\Phi}_m}$. If
$\SR{{J_{\Phi_m}}(\tZ)|_L}>1$, then $\tZ$ is not stable for
$\restr{{{\Phi}_m}}{\prodSimplex}$.
\end{theorem}
\begin{proof}
 
By Theorem~\ref{thm:jac_spectrum}, the nonzero spectrum of
${J_{\Phi_m}}(\tZ)|_L$ is $W\cup Q_m^{-1}\big(\eig{\etax\etay M_x(\tZ)|_{L_x}}\big)$ if $n_x \geq n_y$ and by Remark~\ref{rem:yIdentity}, we have the analogous identity with $M_y(\tZ)|_{L_y}$ if $n_y \geq n_x$. Since the proof in the second case follows the same argument, we assume without loss of generality that $n_x \geq n_y$.
 In this case, $\SR{{J_{\Phi_m}}(\tZ)|_L}>1$ forces (a) some $ \V\in W$ with $ \V>1$, or
(b) some $\mu\in\eig{\etax\etay M_x(\tZ)|_{L_x}}$ with a preimage
$\lambda\in Q_m^{-1}(\mu)$, $|\lambda|>1$. We treat the two cases separately.

\smallskip  
\noindent\textbf{(a) Off-support instability.}
Suppose first that $\V_{1,i}>1$ for some $i\notin S_x$.
Then
\[
  (A\ty)_i>{\vx}=(A\ty)_j
  \qquad
  \text{for every } j\in S_x .
\]
Fix such a $j\in S_x$. We apply
Lemma~\ref{lem:delta_epsilon_tau}, case \textnormal{(I)}, to the fixed point
\[
  \tZ=\Zfp.
\]
Now perturb $\tZ$ by giving the off-support coordinate $i$ a small positive mass.
For $\varepsilon>0$, let
$  x^\varepsilon:=(1-\varepsilon)\tx+\varepsilon e_i
$
and define
$
  Z^\varepsilon_0:=(x^\varepsilon,\ty,x^\varepsilon,\ty).
$ 
Note that this perturbation ensures that $[x_t]_i > 0$. This allows us to apply Lemma~\ref{lem:delta_epsilon_tau}.
Hence there exist constants $\delta>0$, $\epsilon_0>0$, and $\tau>0$ such that any orbit
satisfying $\|Z_t-\tZ\|<\delta$ also satisfies the exponential growth, while every point $Z'\in\prodSimplex$ with $\|Z'-\tZ\|<\delta$ satisfies
 the ratio bound.

Then $Z^\varepsilon_0\in\prodSimplex$,
$
  \|Z^\varepsilon_0-\tZ\|=O(\varepsilon),
$
and
$
  R_0=\frac{x^\varepsilon_i}{x^\varepsilon_j}>0 .
$
By Proposition~\ref{prop:invariance}, coordinates $i$ and $j$ remain positive along the
orbit, so $R_t$ is well defined for all $t$.
Assume for contradiction that the orbit initialized at $Z^\varepsilon_0$ remains in $B_\delta(\tZ)$ for all $t\ge 0$. Then Lemma~\ref{lem:delta_epsilon_tau} gives
\[
  R_t\ge e^{\epsilon_0 t}R_0 \, .
\]
Since $R_0>0$, this implies $R_t\to+\infty$. On the other hand, since the orbit
remains in $B_\delta(\tZ)$, the bounded-ratio part of the same lemma gives
\[
  R_t<\tau
  \qquad
  \text{for all } t\ge 0,
\]
a contradiction. Therefore, the orbit initialized at $Z^\varepsilon_0$ must leave
$B_\delta(\tZ)$.

Since $Z^\varepsilon_0\to \tZ$ as $\varepsilon\downarrow0$, we have constructed
arbitrarily small perturbations of $\tZ$ whose orbits leave the fixed  neighbourhood
$B_\delta(\tZ)$. Hence $\tZ$ is not stable.

The case $\V_{2,k}>1$ for some $k\notin S_y$ is identical, using
Lemma~\ref{lem:delta_epsilon_tau}, case \textnormal{(II)}.   
  
\smallskip
\noindent\textbf{(b) Support (in-face) instability.}
Here the unstable eigenvector lies in $(\Lpar)_{\CC}$ (it is supported on $S_x,S_y$;
Lemma~\ref{lem:detailed_spectrum_1}, Case~3), and
$\SR{{J_{\Phi_m}}(\tZ)|_{\Lpar}}>1$. By Lemma~\ref{cor:face_inv}, ${{\Phi}_m}$
maps a (ambient) neighbourhood of $\tZ$ in $\tZ+\Lpar$ into
$\tZ+\Lpar$ and $\restr{{J_{\Phi_m}}(\tZ)}{\Lpar}$ is the corresponding
differential, so Theorem~\ref{thm:stable_linear} applies with
$V=\Lpar$. Hence, $\tZ$ is unstable relative to $\tZ+\Lpar$. Because
$\tZ\in\relint{\cF_{\tZ}}$, a relative neighbourhood of $\tZ$ in
$\tZ+\Lpar$ is contained in ${\cF_{\tZ}}\subseteq\prodSimplex$; instability
relative to $\tZ+\Lpar$ therefore gives instability relative to
$\prodSimplex$. \end{proof}

\begin{lemma}[Invariance by ${\Phi_m}$ of the local space]
  \label{lem:local_invariance}
 For any fixed point $ \tZ $, there exists $\delta > 0$ such that
 $ {\Phi_m} \bigl( \tZ + (L\cap B_{\delta}(0)) \bigr) \subseteq  \tZ + L $.
\end{lemma}
\begin{proof}
  Let $\tZ=\Zfp$ and $H=[h_1,h_2,h_3,h_4]\in L$. We use the block structure of the operator $\Phi_m = [(\Phi_m)_1,(\Phi_m)_2,(\Phi_m)_3,(\Phi_m)_4]$, where $(\Phi_m)_k:\RR^{2(d_x + d_y)} \rightarrow \RR^{d_x}$ for $k = 1,3$ and $(\Phi_m)_k:\RR^{2(d_x + d_y)} \rightarrow \RR^{d_y}$ for $k= 2,4$.
  Since
  \[
    \bigl(\tx\odot \Exp(\etax A\ty)\bigr)^\top \mathbf 1>0
    \qquad\text{and}\qquad
    \bigl(\ty\odot \Exp(\etay B\tx)\bigr)^\top \mathbf 1>0,
  \]
  continuity implies that there exists $\delta>0$ such that for every
  $H\in L\cap B_\delta(0)$ the first two coordinates of ${\Phi_m}(\tZ+H)$ are well-defined.
  By definition of the projections $P_{d_x}$ and $P_{d_y}$, we then have
  \[
    ({\Phi_m})_1(\tZ+H)^\top \mathbf 1 = 1
    \qquad\text{and}\qquad
    ({\Phi_m})_2(\tZ+H)^\top \mathbf 1 = 1.
  \]
  Moreover, the last two coordinates of ${\Phi_m}$ are just copies of the first two
  coordinates of the input, so
  \[
    ({\Phi_m})_3(\tZ+H) = \tx+h_1
    \qquad\text{and}\qquad
    ({\Phi_m})_4(\tZ+H) = \ty+h_2.
  \]
  Since $H\in L$, we have $h_1^\top \mathbf 1 = h_2^\top \mathbf 1=0$, hence
  \[
    ({\Phi_m})_3(\tZ+H)^\top \mathbf 1 = (\tx+h_1)^\top \mathbf 1 = 1
    \qquad\text{and}\qquad
    ({\Phi_m})_4(\tZ+H)^\top \mathbf 1 = (\ty+h_2)^\top \mathbf 1 = 1.
  \]
  Therefore all four blocks of ${\Phi_m}(\tZ+H)$ satisfy the affine constraints defining
  $\tZ+L$, so ${\Phi_m}(\tZ+H)\in \tZ+L$ for every $H\in L\cap B_\delta$.
\end{proof}

We are now ready to prove Theorem~\ref{thm:stab_from_jac}.

\begin{proof}[Proof of Theorem~\ref{thm:stab_from_jac}]
By Lemma~\ref{lem:local_invariance}, there exists $\delta>0$ such that
\[
  {\Phi_m}\bigl(\tZ + (L\cap B_\delta(0))\bigr)\subset \tZ+L.
\]
Applying Lemma~\ref{lem:jac_stable_subspace} with $\gz=\tZ$, $V=L$, and
$f={\Phi_m}$, we obtain
\[
  D{\Phi_m}(\tZ)L\subset L
  \qquad\text{and}\qquad
  D\bigl({\Phi_m}|_{(\tZ+L)\cap \cD_{\Phi_m}}\bigr)(\tZ)=D{\Phi_m}(\tZ)|_L.
\]
This proves the first claim.

Assume first that $\SR{J_{\Phi_m}(\tZ)|_L}<1$.
Applying Theorem~\ref{thm:stable_linear} with $V=L$, $\gz=\tZ$, and $F={\Phi_m}$
shows that $\tZ$ is asymptotically stable relative to $\tZ+L$. Since
$\prodSimplex \subset \tZ+L$, this implies that $\tZ$ is asymptotically stable for ${\Phi_m}|_{\prodSimplex}$.

Assume now that $\SR{J_{\Phi_m}(\tZ)|_L}>1$.
Applying Theorem~\ref{thm:constraint_unstable} gives the result.
\end{proof}

\section{Proofs of Section \ref{sec:consequences}}\label{apx:consequences}

\subsection{Proof of Corollary \ref{cor:instability}}\label{apx:corMsmall}
\begin{proof}
For a mixed equilibrium, $\dim (L_x) \geq 1$, hence $\eig{\restr{M_x(\Zs)}{L_x}} \neq \varnothing$.  Furthermore, since $\restr{M_x(\Zs)}{L_x}$ is non-singular, $0\notin \eig{\restr{M_x(\Zs)}{L_x}}$. Hence, there exists a $\mu \in \eig{\restr{M_x(\Zs)}{L_x}}$ with $\mu \neq 0$.  
Solving $Q_m(\lambda) = \mu$ reduces to finding the roots of the polynomials
\begin{align*}
\begin{cases}
p_+(\lambda) &= \lambda^2 - (1+(m+1)  s) \lambda + ms\, ,\\
p_-(\lambda) &= \lambda^2 - (1-(m+1)  s) \lambda - m s\, ,
\end{cases}
\end{align*}
where $s^2 = \mu$.
We will show that not all roots can simultaneously lie in the closed unit disk $\DD$. To simplify notation, let $a_+ := (1+(m+1) s)$ and $ b_+ := ms$, similarly $a_- := (1-(m+1) s) $ and $b_- := -m s$. Assume, for contradiction, that all roots lie in $\DD$. By the  Schur--Cohn criterion for quadratic polynomials (cf. Lemma~\ref{lem:scNecCond}), $\absv{a_+ - \bar a_+ b_+} \leq 1 - \absv{b_+}^2$ and $\absv{a_- - \bar a_- b_-} \leq 1 - \absv{b_-}^2$. Hence
\begin{align*}
\begin{cases}
&\absv{1 +  s - m(m+1) \absv{ s}^2} \leq 1 - m^2\absv{ s}^2\, ,\\[1ex]
&\absv{1 - s - m(m+1) \absv{ s}^2} \leq 1 - m^2\absv{ s}^2\, .
\end{cases}
\end{align*}
Squaring and adding gives
\begin{align}\label{ineq:roots3727} 
\absv{1 +  s - m(m+1) \absv{ s}^2}^2 + \absv{1 - s - m(m+1) \absv{ s}^2}^2 \leq 2(1 - m^2\absv{ s}^2)^2\, .
 \end{align}
Note that $1- m(m+1) \absv{ s}^2 \in \RR$. Hence, we use that for $r \in \RR$ and $c \in \CC$, $\absv{r+c}^2 + \absv{r-c}^2 = 2r^2 + r(c+\bar c) - r(c+\bar c) + 2\absv{c}^2$.    Thus, \eqref{ineq:roots3727} is equivalent to
\begin{align*} 
&2(1 - m(m+1) \absv{ s}^2)^2 + 2\absv{ s }^2 \leq 2(1 - m^2\absv{ s}^2)^2\\
  \Leftrightarrow \qquad&\absv{ s}^2\left((1-2m) + m^2(2m+1) \absv{s}^2 \right) \leq 0\, .
 \end{align*}
 However, since $\absv{s}^2 > 0$ and $m \in (0,1/2]$, we have $(1-2m) \geq 0$ and $m^2(2m+1) \absv{s}^2 > 0$. Thus
 \begin{align*}
 \absv{ s}^2\left((1-2m) + m^2(2m+1) \absv{s}^2 \right) > 0\, ,
 \end{align*}
 which gives a contradiction. Thus, there exists a $\lambda \in Q_m^{-1}(\mu)$, such that $\lambda \not \in \DD$ and therefore $\absv{\lambda} > 1$.  Hence $\SR{J_{\Phi_m}(\Zs)|_L}>1$, and Theorem~\ref{thm:stab_from_jac} gives the claim. %
\end{proof}
We provide the (specialized version) of the Schur--Cohn criterion for the convenience of the reader. 
\begin{lemma}[Schur--Cohn criterion]\label{lem:scNecCond}
Consider a quadratic polynomial
$p(\lambda)=\lambda^2-a\lambda+b$
with $a,b\in\CC$.  
Then both roots lie in the open unit disk if and only if
\[
\absv{b}<1
\qquad\text{and}\qquad
\absv{a-\bar a b}<1-\absv{b}^2\, .
\]
 If both roots lie in the closed unit disk, then
\[
\absv{b}\leq 1
\qquad\text{and}\qquad
\absv{a-\bar a b}\leq 1-\absv{b}^2\, .
\]

  \end{lemma}

\begin{proof}
Denote the two roots by $\lambda_+,\lambda_-.$ Then
\[
a=\lambda_++\lambda_-,
\qquad
b=\lambda_+\lambda_- .
\]
We first prove the necessary direction. Assume that both roots lie in the closed unit
disk. Then
\[
\absv{b}
=
\absv{\lambda_+\lambda_-}
\leq 1 .
\]
Moreover,
\begin{align*}
a-\bar a b
&=
\lambda_+ + \lambda_-
-
(\bar\lambda_+ + \bar\lambda_-)\lambda_+\lambda_-        \\
&=
\lambda_+(1-\absv{\lambda_-}^2)
+
\lambda_-(1-\absv{\lambda_+}^2) .
\end{align*}
Thus, using the triangle inequality,
\begin{align*}
\absv{a-\bar a b}
&\leq
\absv{\lambda_+}(1-\absv{\lambda_-}^2)
+
\absv{\lambda_-}(1-\absv{\lambda_+}^2).
\end{align*}
Writing $r=\absv{\lambda_+}$ and $s=\absv{\lambda_-}$, we have
\begin{align*}
1-r^2s^2
-
r(1-s^2)
-
s(1-r^2)
=
(1-r)(1-s)(1+r+s+rs)
\geq 0 .
\end{align*}
Hence
\[
\absv{a-\bar a b}
\leq
1-r^2s^2
=
1-\absv{b}^2 .
\]
For the sufficient direction, assume
\[
\absv{b}<1
\qquad\text{and}\qquad
\absv{a-\bar a b}<1-\absv{b}^2 .
\]
Suppose, for contradiction, that at least one root does not lie in the open unit disk.
Without loss of generality, let
\[
r=\absv{\lambda_+}\geq 1,
\qquad
s=\absv{\lambda_-}.
\]
Since $\absv{b}=rs<1$, we have $s<1$. Using the same identity as above and the
reverse triangle inequality,
\begin{align*}
\absv{a-\bar a b}
&=
\absv{
\lambda_+(1-s^2)
+
\lambda_-(1-r^2)
}        \\
&\geq
r(1-s^2)-s(r^2-1)        \\
&=
(r+s)(1-rs).
\end{align*}
Since $r\geq 1$ and $s<1$,
\[
r+s-(1+rs)
=
(r-1)(1-s)
\geq 0 .
\]
Therefore
\[
\absv{a-\bar a b}
\geq
(r+s)(1-rs)
\geq
(1+rs)(1-rs)
=
1-r^2s^2
=
1-\absv{b}^2 ,
\]
which contradicts the assumed strict inequality. Hence, both roots must lie in the open
unit disk.
 \end{proof}
  
\subsection{Technical Results for the Proofs of Theorem~\ref{thm:2x2games} and Theorem~\ref{thm:stabilityZSG}}
To show Theorem~\ref{thm:2x2games} and Theorem~\ref{thm:stabilityZSG}, we need the following technical results. For convenience, define $\sigma := \etax \etay \SR{M_x(\tZ)|_{L_x}}$, and denote $\sigma^\star(m) := \frac{2m -1}{m^2(2 m+ 1)} $. 
   Define the polynomial in $\lambda \in \CC\setminus\{m/(m+1) \}$ with coefficients $\sigma, m \in \RR$: 
\begin{align}
  P_{\sigma,m}(\lambda) := \lambda^2(1-\lambda)^2 + \sigma((m+1)\lambda -m)^2\, .
\end{align}
Note that $\lambda$ is a root of $P_{\sigma,m}$ if and only if $Q_m(\lambda) = -\sigma$.  Using the continuity of the roots of a polynomial with respect to its coefficients, we will show that
\begin{enumerate}
\item the roots cross the boundaries of the unit disk at $\sigma = 0$ and $\sigma = \sigma^\star(m)$ (see Lemma~\ref{lem:uniteCycle}); and
\item the roots lie in the unit disk for $\sigma \in (0,\sigma^\star(m) )$ and outside or on the boundary otherwise (see Lemma~\ref{lem:roots}). 
\end{enumerate}
See Figure~\ref{fig:rootsD3D4} for illustration. 
 \begin{figure}[h]
\centering

\begin{minipage}{.45\linewidth}
\centering
\includegraphics[width=\linewidth]{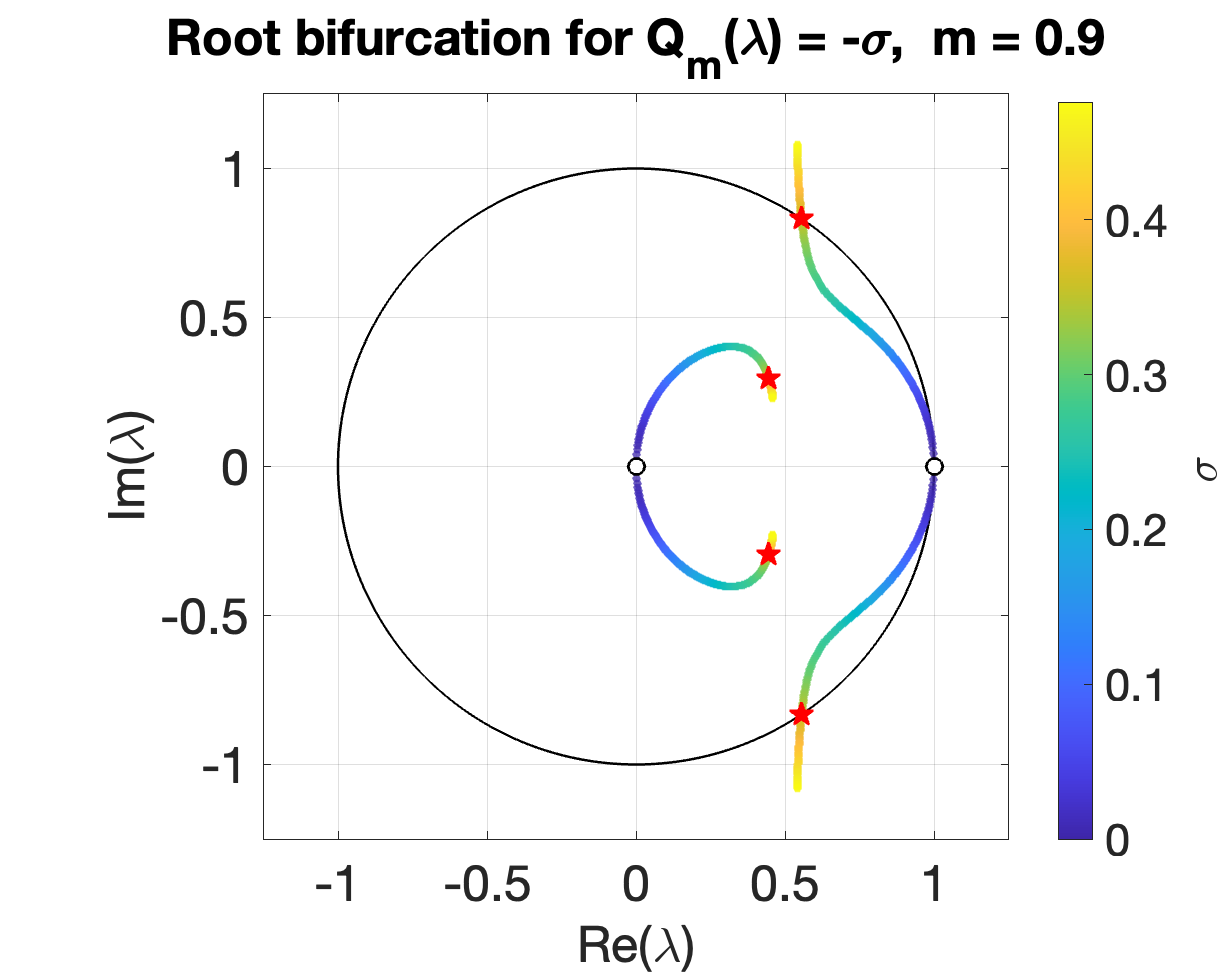}

\end{minipage}
\hfill
\begin{minipage}{.45\linewidth}
\centering
\includegraphics[width=\linewidth]{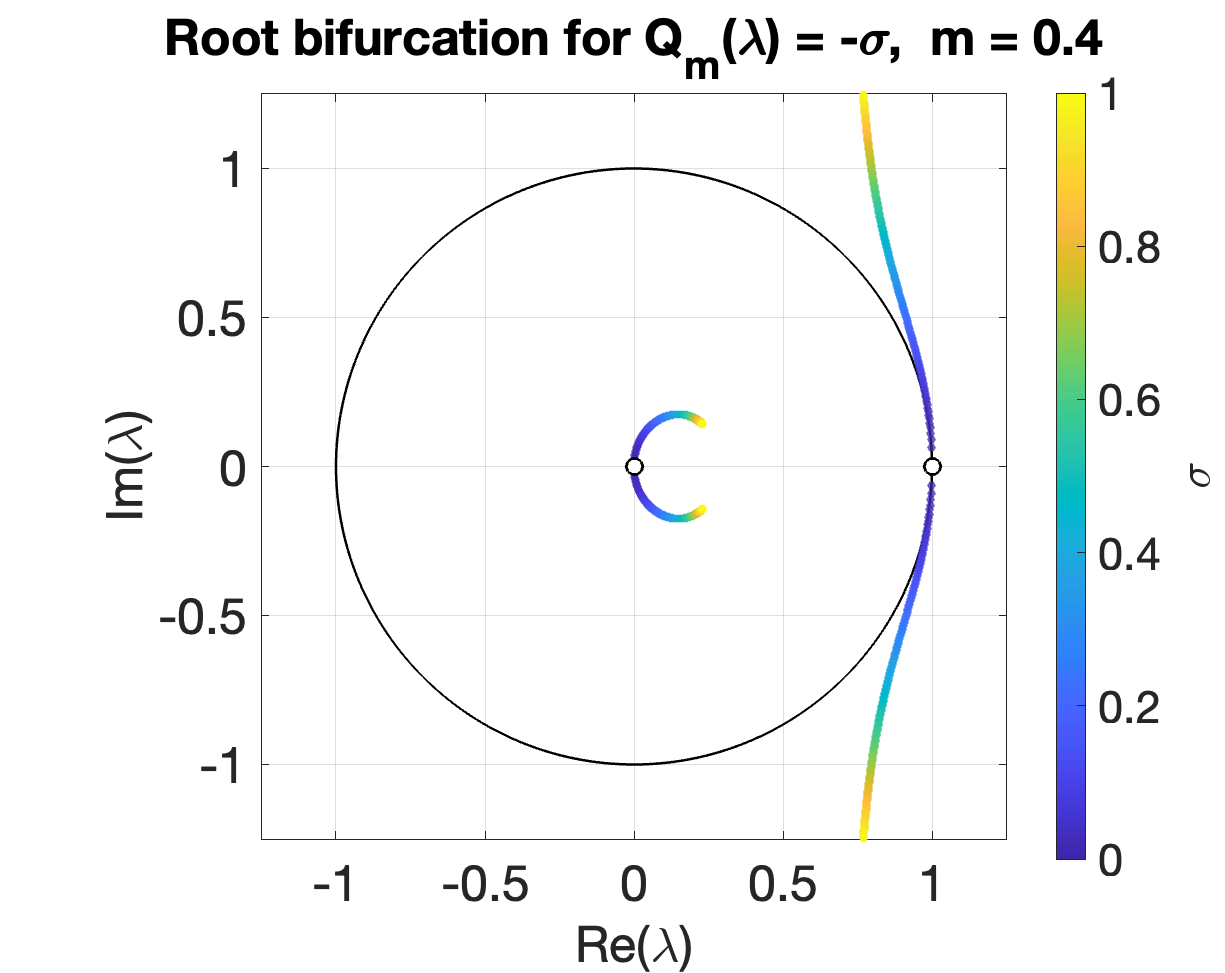}

\end{minipage}
\caption{Illustration of the proof of Lemma~\ref{lem:uniteCycle} and Lemma~\ref{lem:roots}. The double roots for $\sigma = 0$ are marked with a white circle, and the roots for $\sigma = \sigma^\star(m)$ with a red star.}\label{fig:rootsD3D4}
\end{figure}

\paragraph{Technical Details: }
\begin{lemma}\label{lem:cont}  
The multi-set of roots $\lambda \in \CC \setminus \{ m/(m+1) \}$ of $P_{\sigma,m}(\lambda)$ depends continuously on $\sigma$.
\end{lemma}
\begin{proof}
 This follows from the continuity of a polynomial's roots with respect to its coefficients.
\end{proof}
 To understand the existence of a root with modulus strictly less than or greater than $1$, i.e., when there exists a $\lambda$ inside or outside the unit circle, we need to understand when the roots pass through the boundary of the unit circle.

\begin{lemma}\label{lem:uniteCycle}
 Let $Q_m(\lambda) = -\sigma$ where $\sigma \in \RR_+$.
Suppose $m \neq \frac{1}{2}$ and $\absv{\lambda} = 1$.    Then  $\sigma \in \left\{ \sigma^\star(m) , 0\right\}$.
Further, if $m > \frac{1}{2}$ and $\sigma = \sigma^\star(m) $, then there exists a root $\lambda$ with $\absv{\lambda} = 1$.
 \end{lemma}
\begin{proof}
First note that  $Q_m(\lambda)$ must be in $\RR$ by assumption. 
  Since $\absv{\lambda} = 1$ we write $ \lambda = e^{i \theta}$ and for easier notation define $c_m := 2m +1$. By Euler's identity, $e^{i \theta} -1 = e^{i \frac{\theta}{2}}(2 i \sin(\theta/2))$ and $(m+1)e^{i \theta/2} - m e^{- i  \theta/2} = \cos(\theta/2) +  i c_m \sin(\theta/2) $. Hence, using trigonometric identities, we find that $Q_m(e^{i\theta})$ is equal to
 \begin{align*} 
   &- \frac{4 \sin^2\left( \frac{\theta}{2}  \right) e^{2 i \theta}}{ \left(\cos\left( \frac{\theta}{2}\right)+ i c_m\sin\left( \frac{\theta}{2}\right)\right)^2}\nonumber \\ 
 &\overset{(1)}{=}   - 4 \sin^2\left( \frac{\theta}{2}  \right) \cos^2 \left( \frac{ \theta}{2}\right)\left( \frac{\left( 1+ i  \tan \left( \frac{\theta}{2}\right)\right)^2}{  1+ i c_m \tan \left( \frac{\theta}{2}\right)}\right)^2\nonumber\\
 &= - 4 \sin^2\left( \frac{\theta}{2}  \right) \cos^2 \left( \frac{ \theta}{2}\right)\left(  \underbrace{ \frac{1 -  \tan^2 \left(\frac{\theta}{2} \right) + 2 c_m  \tan^2 \left(\frac{\theta}{2} \right)}{1+c_m^2 \tan^2 \left(\frac{\theta}{2} \right)} }_{:= \xi}+ i \underbrace{\frac{ \tan \left(\frac{\theta}{2} \right) \left(c_m  \tan^2 \left(\frac{\theta}{2} \right) + 2 - c_m \right)}{1+c_m^2 \tan^2 \left(\frac{\theta}{2} \right)}}_{=: \rho} \right)^2\, .
 \end{align*}
 Since $\tan(\pi/2)$ is undefined, the equality in $(1)$ is invalid  if $\theta = \pi$, that is $\lambda = -1$. However, we note that $Q_m(-1) = \frac{4}{(2m +1)^2} > 0$, hence, this case is excluded by the assumption that $Q_m(\lambda)$ is negative. 
Since $Q_m(e^{i \theta}) = Q_m(\lambda)$ must be real, $\rho \xi $ must be zero. 
Since \[1 -  \tan^2 \left(\frac{\theta}{2} \right) + 2 c_m  \tan^2 \left(\frac{\theta}{2} \right) = 1 + (4m +1 )\tan^2 \left(\frac{\theta}{2} \right) > 0\, , \]
$\rho$ must be zero. This implies that  one of the following conditions must hold:
\begin{align*}
(1) \qquad \theta = 0 \qquad \text{ or, } \qquad (2) \qquad \tan^2 \left(\frac{\theta}{2} \right) = \frac{c_m -2}{ c_m}\,.
\end{align*}
 In Case $(1)$,   $  Q_m(e^{i 0})  = 0$, hence $\sigma = 0$, and in Case (2), we simplify $Q_m(e^{i \theta}) = -\frac{2m-1}{m^2(2m+1)}$ which implies that $\sigma = \sigma^\star(m)$. 
 
   It remains to show that for $m > \frac{1}{2}$,
$\sigma=\sigma^\star(m)$ is attained by a root on the
unit circle.  
 Choose
\[
  \theta := 2\arctan\left(\sqrt{\frac{c_m-2}{c_m}}\right),
  \qquad
  \lambda := e^{i\theta}.
\]
Due to the assumption $m > \frac{1}{2}$, $\frac{c_m-2}{c_m}=\frac{2m-1}{2m+1}>0$.
 Using $\sigma \in \RR$, hence $\rho = 0$, by the same computations as before,
 \[
  Q_m(\lambda)
  =
  -\frac{4(c_m-2)}{c_m(c_m-1)^2}
  =
  -\frac{2m-1}{m^2(2m+1)}
  =
  -\sigma^\star(m)\, 
\]
 and $\absv{\lambda} = 1$. 

\end{proof}

\begin{lemma}\label{lem:roots}
 If $m  > 1/2 $ and
\begin{enumerate}
\item $\sigma \in \left(0,\sigma^\star(m) \right)$, then every root of $P_{\sigma, m}(\lambda)$ satisfies that the modulus is bounded as $\absv{\lambda} < 1$;
\item $\sigma = \sigma^\star(m) $,  then at least one root of $P_{\sigma, m}(\lambda)$ lies  exactly on the boundary of the unit circle, i.e., $\absv{\lambda} = 1$;
\item $\sigma  > \sigma^\star(m) $, then at least one root of $P_{\sigma, m}(\lambda)$ lies outside the unit circle, i.e.,  $\absv{\lambda} > 1$;

\end{enumerate}
If $m \leq 1/2$ and $\sigma \geq 0$, at least one root of $P_{\sigma, m}(\lambda)$ has modulus greater than or equal to $1$. 
\end{lemma}
Note that this technical lemma does not require any assumption on the game.

\begin{proof}  
Let $s=\sqrt{\sigma}$. If $\sigma=0$, then
\[
  P_{0,m}(\lambda)=\lambda^2(1-\lambda)^2,
\]
so the roots are $0,0,1,1$.

Assume now that $\sigma>0$. We factor
\[
P_{\sigma,m}(\lambda)
=
\prod_{\delta\in\{-1,1\}}
\left(
\lambda(1-\lambda)+\delta i s((m+1)\lambda-m)
\right).
\]
Equivalently,
\[
  P_{\sigma,m}(\lambda)=\prod_{\delta\in\{-1,1\}}p_\delta(\lambda),
\]
where
\[
  p_\delta(\lambda)
  =
  \lambda^2-a_\delta\lambda+b_\delta,
  \qquad
  a_\delta=1+\delta i(m+1)s,
  \qquad
  b_\delta=\delta i m s .
\]
For each $\delta$,
\[
  \absv{b_\delta}^2=m^2\sigma,
\]
and
\[
  a_\delta-\bar a_\delta b_\delta
  =
  1-m(m+1)\sigma+\delta i\sqrt{\sigma}.
\]
Hence
\[
  \absv{a_\delta-\bar a_\delta b_\delta}^2
  =
  (1-m(m+1)\sigma)^2+\sigma,
\]
and therefore
\begin{align}\label{eq:SC-difference}
 (1-\absv{b_\delta}^2)^2
 -\absv{a_\delta-\bar a_\delta b_\delta}^2 
=
 \sigma\Big((2m-1)-m^2(2m+1)\sigma\Big).
\end{align}

Suppose first that $m>1/2$ and
\[
  0<\sigma<\sigma^\star(m)
  =
  \frac{2m-1}{m^2(2m+1)}.
\]
Then the right-hand side of \eqref{eq:SC-difference} is positive. Moreover,
\[
  m^2\sigma<m^2\sigma^\star(m)
  =
  \frac{2m-1}{2m+1}<1.
\]
Thus, for both $\delta=\pm1$,
\[
  \absv{b_\delta}<1
  \qquad\text{and}\qquad
  \absv{a_\delta-\bar a_\delta b_\delta}<1-\absv{b_\delta}^2.
\]
By Lemma~\ref{lem:scNecCond}, both roots of each quadratic factor
$p_\delta$ lie in the open unit disk. Hence every root of
$P_{\sigma,m}$ satisfies $\absv{\lambda}<1$.

Next let $m>1/2$ and $\sigma=\sigma^\star(m)$. By the unit-circle calculation
in Lemma~\ref{lem:uniteCycle}, choosing
\[
  \theta
  =
  2\arctan\sqrt{\frac{2m-1}{2m+1}},
  \qquad
  \lambda=e^{i\theta},
\]
gives
\[
  Q_m(\lambda)
  =
  -\frac{2m-1}{m^2(2m+1)}
  =
  -\sigma^\star(m).
\]
Equivalently, $P_{\sigma^\star(m),m}(\lambda)=0$, and by construction
$\absv{\lambda}=1$. Thus at least one root lies on the unit circle.

Now let $m>1/2$ and $\sigma>\sigma^\star(m)$. Suppose, for contradiction, that
all roots of $P_{\sigma,m}$ lie in the closed unit disk. Then the roots of each
quadratic factor $p_\delta$ lie in the closed unit disk. By the necessary part
of Lemma~\ref{lem:scNecCond},
\[
  \absv{a_\delta-\bar a_\delta b_\delta}
  \le
  1-\absv{b_\delta}^2.
\]
Squaring gives
\[
  (1-\absv{b_\delta}^2)^2
  -\absv{a_\delta-\bar a_\delta b_\delta}^2
  \ge0.
\]
This contradicts \eqref{eq:SC-difference}, whose right-hand side is negative
when $\sigma>\sigma^\star(m)$. Hence, at least one root satisfies
$\absv{\lambda}>1$.

Finally, suppose that $m\le1/2$. If $\sigma=0$, then $\lambda=1$ is a double root,
so the desired conclusion holds. If $\sigma>0$ and all roots lay in the open
unit disk, then Lemma~\ref{lem:scNecCond} applied to each $p_\delta$ would imply
\[
  (1-\absv{b_\delta}^2)^2
  -\absv{a_\delta-\bar a_\delta b_\delta}^2
  >0.
\]
But by \eqref{eq:SC-difference},
\[
  (1-\absv{b_\delta}^2)^2
  -\absv{a_\delta-\bar a_\delta b_\delta}^2
  =
  \sigma\Big((2m-1)-m^2(2m+1)\sigma\Big)<0
\]
for every $m\le1/2$ and every $\sigma>0$. This contradiction shows that not all
roots lie in the open unit disk. Therefore, at least one root satisfies
$\absv{\lambda}\ge1$.

 \end{proof}
 
 \begin{lemma}\label{lem:nonSingularM}
 Consider a zero-sum game $\Gamma(A, -A^\top)$ and assume the Nash equilibrium is unique and fully mixed. Then $\restr{M_x(Z^\star)}{L_x}$ is non-singular. 
 \end{lemma}
  
 \begin{proof}
 Write $Z^\star = [x^\star,y^\star,x^\star,y^\star]$.
Since the equilibrium is fully mixed,  \[
  \restr{M_x(Z^\star)}{L_x}
  =
  - H(x^\star) A H(y^\star) A^\top .
\]
Write $H_x:=H(x^\star)$ and $H_y:=H(y^\star)$.
  Suppose, for contradiction, that
$\restr{M_x(Z^\star)}{L_x}$ is singular. Then there exists
$0\neq h\in L_x$ such that
\[
  M_x(Z^\star)h=0 \quad \Leftrightarrow \quad
   H_x A H_y A^\top h=0 \,.
\]
Thus, $A H_y A^\top h \in \ker H_x$. But since $\ker H_x=\Span(\ones_{d_x})$, there exists
$c\in\RR$ such that
\[
  A H_y A^\top h = c\ones_{d_x}.
\]
Taking the inner product with $h$, and using $h\in L_x$, gives
\[
  0
  =
  c\, h^\top \mathbf 1_{d_x}
  =
  h^\top A H_y A^\top h
  =
  (A^\top h)^\top H_y (A^\top h).
\]
Since $H_y$ is positive semidefinite with kernel
$\Span (\mathbf 1_{d_y}) $, it follows that
$  A^\top h \in \Span(\mathbf 1_{d_y}).$
Thus, for some $c'\in\RR$,
\[
  A^\top h = c' \mathbf 1_{d_y}.
\]
Because $h\in L_x$, we have $\mathbf 1_{d_x}^\top h=0$. Since
$x^\star\in\relint \Delta_{d_x}$, for all sufficiently small
$\epsilon\neq 0$,
\[
  x(\epsilon) := x^\star+\epsilon h
\]
still belongs to  $\relint\Delta_{d_x}$, and
$x(\epsilon)\neq x^\star$.

We show that $(x(\epsilon),y^\star)$ is again a Nash equilibrium:
Since $y^\star$ is fully mixed at the zero-sum equilibrium, $  A y^\star = v\ones_{d_x}$
for some $v\in\RR$.  
Similarly, since $x^\star$ is fully mixed at equilibrium,
$  A^\top x^\star = v\ones_{d_y}.$
Using $A^\top h=c' \ones_{d_y}$, we obtain
\[
  A^\top x(\epsilon)
  =
  A^\top x^\star+\epsilon A^\top h
  =
  (v+\epsilon c')\ones_{d_y}\,.
\]
 Hence, $x(\epsilon)$ is a best response to $y^\star$ and vice versa. Consequently $(x(\epsilon),y^\star)$ is a fully mixed Nash equilibrium distinct from
$[x^\star,y^\star]$, contradicting uniqueness. Therefore, no such nonzero
$h\in L_x$ exists, and $\restr{M_x(Z^\star)}{L_x}$ is non-singular.
\end{proof}

 \subsection{Proof of Theorem \ref{thm:stabilityZSG}} \label{apx:consequences_ZSG}
\begin{proof}[Proof of Theorem \ref{thm:stabilityZSG}] We show that the modulus of the eigenvalues of the Jacobian is strictly less than or strictly greater than one. This implies stability under the assumed step size conditions. 
 Let $Z^\star$ denote the unique Nash equilibrium and $M_x(Z^\star)$ the matrix corresponding to it.
 We first note that $M_x(Z^\star)$ has a real, non-positive spectrum. This is due to $H({x^\star})$ and $A H({y^\star})A^\top$ being symmetric positive semidefinite, and the non-zero spectrum of $H({x^\star})A H({y^\star})A^\top$ coincides with the non-zero spectrum of the symmetric positive semidefinite matrix $H({x^\star})^{\frac12}A H({y^\star})A^\top H({x^\star})^{\frac12}$. Thus, the spectrum of $M_x(Z^\star) = - H({x^\star})A H({y^\star})A^\top$ is real and non-positive. Furthermore, since $Z^\star$ is by assumption full support, $\restr{M_x(Z^\star)}{L_x}$ is non-singular (cf. Lemma~\ref{lem:nonSingularM}). Hence, $\eig{\restr{M_x(Z^\star)}{L_x}} \subset (-\infty, 0)$. Since the entries in $M_x(Z^\star)$ are finite, $\delta := \max(\absv{\lambda} : \lambda \in  \eig{\restr{M_x(Z^\star)}{L_x}})$, is finite and  $\eig{\etax\etay \restr{M_x(Z^\star)}{L_x}} \subset [-\etax\etay \delta, 0)$. Hence, by Theorem \ref{thm:jac_spectrum}, the spectrum of $J_{\Phi_m}(Z^\star)|_{L}$ are either the values in $W$ or the solutions to
 \[ Q_m(\lambda) = -\sigma \qquad \sigma \in (0,\etax\etay \delta]\,.\]
  Due to the assumption that the Nash equilibrium has full support, all values in $W$  are strictly less than $1$. Hence, we can ignore them for this proof. 
  Under the assumption that $m \in ( 1/2,1]$, the claim follows by noting that $Q_m(\lambda) = - \sigma $ if and only if $P_{\sigma, m}(\lambda) = 0$. Combining this with Lemma \ref{lem:roots} implies the claim.
   The case $m \in (0, 1/2]$, follows from Lemma~\ref{lem:nonSingularM} and Corollary~\ref{cor:instability}.
   \end{proof}
 
 \subsection{Proof of Theorem \ref{thm:2x2games}}\label{apx:theorem2x2}
 \begin{proof}[Proof of Theorem \ref{thm:2x2games}]
 Since the Nash equilibrium is, by assumption, fully supported, $W = \varnothing$.
In the $2$-dimensional case, $L_x$ is the one-dimensional space of vectors proportional 
to $(-1, 1)$.
 Moreover, 
\[
  H(x^\star)
  =
  p(1-p)
  \begin{bmatrix}
    \phantom{-}1 & -1\\
    -1 & \phantom{-}1
  \end{bmatrix},
  \qquad
  H(y^\star)
  =
  q(1-q)
  \begin{bmatrix}
    \phantom{-}1 & -1\\
    -1 & \phantom{-}1
  \end{bmatrix}
\]
Let $r=(1,-1)^\top$. Then
$H(x^\star)=p(1-p)rr^\top$ and $H(y^\star)=q(1-q)rr^\top$.  Hence

\[
  M_x(Z^\star)
  =
  H(x^\star) A H(y^\star) B
  =
  p(1-p)q(1-q)\, r r^\top A r r^\top B .
\]
Since
\[
  r^\top A r=\Delta_A,
  \qquad
  r^\top B r=\Delta_B,
\]
we obtain
\[
  M_x(Z^\star)
  =
  p(1-p)q(1-q)\Delta_A
  \begin{bmatrix}
    e-g &\phantom{-} f-h\\
    g-e &\phantom{-} h-f
  \end{bmatrix}.
\]
Restricting to $L_x= \Span \{(-1,1)^\top\}$,
 \[
  \eig{M_x(Z^\star)|_{L_x}}
  =
  \left\{
    p(1-p)q(1-q)\Delta_A\Delta_B
  \right\}.
\]
We write this eigenvalue as
\[
  \mu =
  p(1-p)q(1-q)\Delta_A\Delta_B.
\]
Thus, the local criterion becomes 
\[
  Q_m(\lambda)=\etax\etay\,\mu.
\]
If $\Delta_A\Delta_B>0$, then $\mu>0$, and by Corollary~\ref{cor:unstable_positive_mu}, the mixed equilibrium is locally unstable for 
every $\etax\etay>0$, and every $m$. 

If $\Delta_A\Delta_B<0$ and $m \leq \frac12$, instability follows from Corollary~\ref{cor:instability}. Note that in this case $\etax \etay > E$ for all strictly positive step sizes. 
 
Thus, assume $\Delta_A\Delta_B<0$ and $m > \frac12$, then we have local stability if all solutions to
\[
  Q_m(\lambda)=-\etax\etay|\mu| \, ,
\]
are of modulus strictly less than $1$. Applying Lemma~\ref{lem:roots} and Theorem~\ref{thm:jac_spectrum} gives that the full-support mixed equilibrium is locally asymptotically stable when
\begin{align*}
\etax\etay\absv{\mu} \in \left(0,\sigma^\star(m) \right)\,.
\end{align*}
 When $\etax\etay\absv{\mu} > \sigma^\star(m)$, by Lemma~\ref{lem:roots}~(3) it is unstable.

 The degenerate case $\Delta_A\Delta_B=0$ gives $\mu=0$ and therefore a unit eigenvalue;
in this case, the linear criterion is inconclusive.
\end{proof}

\subsection{Proof of Theorem \ref{thm:lowDim}}\label{apx:lowDim}
\begin{proof}[Proof of Theorem~\ref{thm:lowDim}]
We denote the $i^{\mathrm{th}}$ row of the matrices $A,B$ by $A_{i,:}$ and $B_{i,:}$.  Define
\begin{align*}
  M_A=
  \begin{bmatrix}
     A_{1,:} - A_{2,:}\\
     \vdots\\
      A_{1,:} - A_{d_x,:}\\
       \ones^\top    \\
  \end{bmatrix},
  \qquad
  M_B=
     \begin{bmatrix}
     B_{1,:} - B_{2,:}\\
     \vdots\\
      B_{1,:} - B_{d_y,:}\\
       \ones^\top    \\
  \end{bmatrix}\, .
 \end{align*}
  Furthermore, we denote by $M_A^{(i)}$ (respectively $M_B^{(i)}$) the matrix $M_A$ where the $i^{\mathrm{th}}$ column is substituted by $e_{d_x}$ (respectively $e_{d_y}$).
We first show the following:
 \begin{enumerate}
 \item  $\det(M_A) \neq 0$ and $\det(M_B)\neq 0$;
 \item Under the assumption that the Nash equilibrium is fully mixed and unique the Nash equilibrium is the solutions to the linear equations: $M_A y = e_{d_x}$ and $M_B x = e_{d_y}$. Hence \[y^\star_i = \frac{\det(M_A^{(i)})}{\det(M_A)} \qquad \text{ and } \qquad x^\star_i = \frac{\det(M_B^{(i)})}{\det(M_B)}\, .\] 
 \end{enumerate}
 Due to Theorem~\ref{thm:FPPhi}, 
 $ [Ay^\star]_1= [Ay^\star]_i$ for $i \in \{2, \dots, d_x\}$ and $y^\star \in \Delta_{d_y}$
Equivalently,
\[ 
  [Ay^\star]_1-[Ay^\star]_i=0  \text{ for } i \in \{2, \dots, d_x\} \text{, and }
  \sum_{i=1}^{d_y} y_i^\star=1.
\]
By the definition of $M_A$, these equations are exactly
\[
  M_Ay^\star=e_{d_x}.
\]
Now, assume for contradiction that $M_A$ is singular. Then $M_A z = 0$ has a non-zero solution and, due to the last row in $M_A$, $ \dprod{z,\ones} = 0$.   We define $\hat y(\alpha) = y^\star + \alpha z  $. Then $M_A\hat y(\alpha) = M_A y^\star + \alpha M_A z = e_{d_x}$ and for a sufficiently small $\alpha>0$, $\hat y(\alpha) \in \Delta_{d_y}$ (recall $y^\star \in \relint \Delta_{d_y}$ by assumption). Thus, $[A\hat y(\alpha)]_1 = [A\hat y(\alpha)]_i$ for all $i \in \{2, \dots, d_x\}$, that is, the payoff is constant. Conversely, $B x^\star$ is constant by definition of $x^\star$.
Thus, $[x^\star, \hat y(\alpha)]$ is a Nash equilibrium since $x^\star \in \Delta_{d_x}$ is a best response to  $\hat y(\alpha)$ and vice versa.  This contradicts the uniqueness of the Nash equilibrium.

 Applying Cramer's rule gives $x^\star_i$ and $y^\star_i$ as closed-form solutions in the entries of the game matrices, and consequently, the linear operator $\restr{M_x(Z^\star)}{L_x}$ can be expressed in closed-form. So far, we do not require any assumptions on the dimensions besides $d_x = d_y$. 
 
 The claim follows from Theorem~\ref{thm:jac_spectrum} combined with the observation that the characteristic polynomial of $  \restr{M_x(Z^\star)}{L_x}$ has degree at most $4$ and thus roots are expressible by radicals.

   \end{proof}
   
     \paragraph{Closed Form Spectrum Computation:}  For illustration of Theorem~\ref{thm:lowDim}, we provide the closed-form formulas for the spectrum.
   Fix any real linear basis for $L_x$ and write $d:=d_x=d_y$. Let
\[
  N:=\etax\etay\,M_x(Z^\star)|_{L_x}\ \in\ \RR^{(d-1)\times(d-1)}\, .
\]
 By Theorem~\ref{thm:jac_spectrum} 
the nonzero spectrum of $J_{\Phi_m}(Z^\star)|_{L}$ equals $Q_m^{-1}\big(\eig{N}\big)$ since  $W=\varnothing$ at a fully mixed equilibrium.
It is convenient to read off $\eig{N}$ from the full $d\times d$ matrix
\[
  M:=\etax\etay\,H(x^\star)\,A\,H(y^\star)\,B\ \in\ \RR^{d\times d}.
\]
Since $\ones^\top H(x^\star)=0$ we have $\ones^\top M=0$ and therefore $0\in\eig{M}$. The
remaining $d-1$ eigenvalues of $M$ are precisely $\eig{N}=\{\fS_1,\dots,\fS_{d-1}\}$.
In particular, the power sums are
\[
  p_j:=\Tr\!\big(M^{j}\big)=\sum_{i=1}^{d-1}\fS_i^{\,j}\qquad(j\ge 1),
\]
and the characteristic polynomial of $N$ is
\[
  \chi_N(\fS)=\prod_{i=1}^{d-1}(\fS-\fS_i)
  =\sum_{j=0}^{d-1}(-1)^j E_j\,\fS^{\,d-1-j},\qquad E_0=1,
\]
where, by Newton's identities,
\begin{align*}
  E_1=&p_1,\quad
  E_2=\tfrac12\big(p_1^2-p_2\big),\quad
  E_3=\tfrac16\big(p_1^3-3p_1p_2+2p_3\big),\\
  &E_4=\tfrac1{24}\big(p_1^4-6p_1^2p_2+3p_2^2+8p_1p_3-6p_4\big),
\end{align*}
and $E_{d-1}=\det N$. Each $p_j$ is a trace of a product of $H(x^\star),A,H(y^\star),B$, hence a
polynomial in their entries; with the closed forms for $x^\star,y^\star$, every $E_j$ is an
explicit rational function of the entries of $A$ and $B$.

\textbf{Eigenvalues of $N$} We use that $\deg\chi_N=d-1\le 4$, hence the roots can be computed in closed form. 
\begin{itemize} 
\item $d=2$: Here $N \in \RR^{1 \times 1}$. Thus, $\fS_1=E_1= \det N=\Tr M$.
\item  $d=3$: Set $E_2=\det N=\tfrac12\big((\Tr M)^2-\Tr M^2\big)$. Then 
\[ \fS_{1,2}=\frac{E_1\pm\sqrt{E_1^2-4E_2}}{2}\]
\item $d=4$:  We use Cardano's method. The substitution $\fS=t+\tfrac{E_1}{3}$ yields the depressed cubic
  $t^3+pt+q$ with $p=E_2-\tfrac{E_1^2}{3}$, $q=-\tfrac{2E_1^3}{27}+\tfrac{E_1E_2}{3}-E_3$. Set
  \[ U^3 = -\tfrac q2+\sqrt{\tfrac{q^2}{4}+\tfrac{p^3}{27}}\qquad V^3 = -\tfrac q2-\sqrt{\tfrac{q^2}{4}+\tfrac{p^3}{27}} \,\]
  with $UV = -\frac{p}{3}$.
  Then for $k \in \{0,1,2\}$ and $\omega=e^{2\pi i/3}$ 
  \[
    \fS_k=\frac{E_1}{3}
      +\omega^{k} U
      +\omega^{2k} V
       \, .
  \]
\item $d=5$: We use Ferrari's method. Again, the substitution $\fS=t+\tfrac{E_1}{4}$ yields the depressed
  quartic $t^4+pt^2+qt+r$ with
  \[
    p=E_2-\tfrac{3E_1^2}{8},\quad
    q=-E_3+\tfrac{E_1E_2}{2}-\tfrac{E_1^3}{8},\quad
    r=E_4-\tfrac{E_1E_3}{4}+\tfrac{E_1^2E_2}{16}-\tfrac{3E_1^4}{256}.
  \]
  If $q=0$, the depressed quartic is
    biquadratic and its roots are
    $t=\pm\sqrt{(-p\pm\sqrt{p^2-4r})/2}$. Otherwise, the resolvent has a
    nonzero root since its value at $0$ is $-q^2<0$.
  Thus, let $u$ be any non-zero root of the resolvent cubic $8u^3+8pu^2+(2p^2-8r)u-q^2=0$ (solved by the
  case $d=4$). Then $t^4+pt^2+qt+r$ factors as
  \[
    \Big(t^2-\sqrt{2u}\,t+\big(\tfrac p2+u+\tfrac{q}{2\sqrt{2u}}\big)\Big)
    \Big(t^2+\sqrt{2u}\,t+\big(\tfrac p2+u-\tfrac{q}{2\sqrt{2u}}\big)\Big),
  \]
  and the four roots $\fS_i=t_i+\tfrac{E_1}{4}$ follow from the quadratic formula.
\end{itemize}

\medskip\noindent\textbf{Assembling the Jacobian spectrum.}
As in the proof of Corollary~\ref{cor:instability}, each $\fS_i$ contributes the (at most
four) eigenvalues $\lambda$ of $J_{\Phi_m}(Z^\star)|_{L}$ solving $Q_m(\lambda)=\fS_i$, i.e.\
the roots of
\[
  p_\pm(\lambda)=\lambda^2-\big(1\pm(m+1)s_i\big)\lambda\pm m\,s_i,
  \qquad s_i^2=\fS_i,
\]
namely
\[
  \lambda=\frac{\big(1\pm(m+1)s_i\big)\pm'\sqrt{\big(1\pm(m+1)s_i\big)^2\mp 4m\,s_i}}{2}.
\]
With $W=\varnothing$, this gives the full spectrum of $J_{\Phi_m}(Z^\star)$ in closed form in the
entries of $A$ and $B$.

  \subsection{Details on Related Work}
  \label{apx:deMontbrun}
  
Throughout this subsection, we keep the conventions of the main text:
$\Gamma(A,B)$ with $A \in \RR^{d_x \times d_y}$,
$B \in \RR^{d_y \times d_x}$, played over the unconstrained strategy
spaces $\RR^{d_x}$ and $\RR^{d_y}$. The general-sum bilinear games
$(x^\top \tilde A y, x^\top \tilde B y)$,
$\tilde A, \tilde B \in \RR^{d_x \times d_y}$, of
\citet{de2025optimistic} correspond to $\Gamma(A, B)$ via
$\tilde A = A$ and $\tilde B = B^\top$; their matrices
$\tilde B^\top \tilde A$ and $\tilde A \tilde B^\top$ are $BA$ and
$AB$ in our notation, and their OGDA (Definition 3.1 in
\citet{de2025optimistic}) is the OGM of
Section~\ref{sec:CompareResultsJacobian} with
$\etax = \etay = \eta$.

The OGM iteration is linear: writing
$Z_t = (x_t, y_t, x_{t-1}, y_{t-1})$, we have
$Z_{t+1} = \Lambda_{A,B} Z_t$ with
\[
  \Lambda_{A,B} =
  \begin{bmatrix}
    I_{d_x} & 2\etax A & 0 & -\etax A\\
    2\etay B & I_{d_y} & -\etay B & 0 \\
    I_{d_x} & 0 & 0 & 0\\
    0 & I_{d_y} & 0 & 0
  \end{bmatrix}
  \in \RR^{2(d_x + d_y) \times 2(d_x + d_y)}
\]
(cf. Section 3.2 in \cite{de2025optimistic}, where
$\etax = \etay$). The fixed points of the dynamics are exactly the
lifted Nash equilibria $[x,y,x,y]$ with
$[x,y] \in \Ker(B) \times \Ker(A)$.

Let $\lambda \in \CC$ and $H = [h_1, h_2, h_3, h_4] \neq 0$ with
$\Lambda_{A,B} H = \lambda H$. The last two block rows give
$h_1 = \lambda h_3$ and $h_2 = \lambda h_4$; substituting into the first two
rows yields
\begin{align}\label{eq:OGMreduced}
  \lambda(\lambda - 1)\, h_3 = \etax (2\lambda - 1) A h_4\,,
  \qquad
  \lambda(\lambda - 1)\, h_4 = \etay (2\lambda - 1) B h_3\,.
\end{align}
For $\lambda = \tfrac12$, \eqref{eq:OGMreduced} forces
$h_3 = h_4 = 0$, hence $H = 0$. Thus, under the assumption that $H \neq 0$, $\tfrac12$ is never an eigenvalue.
 Combining these equations gives
 \[ \begin{cases}
\lambda^2 (1-\lambda)^2 h_4 = (2\lambda -1)^2 \etax\etay\, B A\, h_4\\
\lambda^2 (1-\lambda)^2 h_3= (2\lambda -1)^2 \etax \etay\, A B\, h_3\,,
\end{cases}\]
with $[h_3, h_4] \neq [0,0]$. Hence, for $\lambda \neq \frac{1}{2}$,
$\frac{\lambda^2 (1-\lambda)^2}{\etax \etay(2\lambda -1)^2}$ is an
eigenvalue of $BA$ if $h_4 \neq 0$ or of $AB$ if $h_3 \neq 0$.
Conversely, following the proof of Proposition 3.7 in
\citet{de2025optimistic} verbatim, with the only modification that
the products $\etax\etay$ replace $\eta^2$ throughout, every such
$\lambda$ is an eigenvalue of $\Lambda_{A,B}$. This gives
\[ \eig{\Lambda_{A,B}} = \bigcup_{\mu \in \teig{A B} \cup \teig{B A}}
\hat\cS^\star(\mu)\,,\qquad
\hat\cS^\star(\mu) := \{\lambda \in \CC : \lambda^2 (1-\lambda)^2 =
\mu\, \etax \etay\, (1 - 2\lambda)^2  \}\, . \]

We now prove Corollary~\ref{cor:deMontbrun}. Since $d_x = d_y$ and
$\teig{BA} \subset (-\infty, 0)$, the matrix $BA$ is non-singular,
hence $\Ker(A) = \Ker(B) = \{0\}$ and the origin is the unique Nash
equilibrium. Moreover, $AB$ and $BA$
share their spectrum, so $\teig{AB} \cup \teig{BA} = \teig{BA}$. For
$\mu \in \teig{BA}$ write $\mu = -\absv{\mu}$. Then
$\lambda \in \hat\cS^\star(\mu)$ if and only if
$P_{\sigma, 1}(\lambda) = 0$ with
$\sigma = \etax\etay \absv{\mu} > 0$. If
$\etax\etay \SR{-BA} < \frac13 = \sigma^\star(1)$, Lemma~\ref{lem:roots}
(with $m = 1$) shows that every eigenvalue of $\Lambda_{A,B}$ has
modulus strictly less than $1$; since the dynamics are linear, the
fixed point is (globally) asymptotically stable. If
$\etax\etay \SR{-BA} > \frac13$, applying Lemma~\ref{lem:roots}~(3) to
an eigenvalue $\mu$ with $\absv{\mu} = \SR{BA}$ produces
$\lambda \in \eig{\Lambda_{A,B}}$ with $\absv{\lambda} > 1$, and the
fixed point is unstable.

Finally, we note why $d_x = d_y$ cannot be dropped: if, say,
$d_x > d_y$ and $\teig{BA} \subset (-\infty,0)$, then
$AB \in \RR^{d_x \times d_x}$ has rank at most $d_y < d_x$, so
$0 \in \teig{AB}$ and $\hat\cS^\star(0) = \{0, 1\}$ contributes the
eigenvalue $1$ to $\eig{\Lambda_{A,B}}$; correspondingly
$\Ker(B) \neq \{0\}$, the Nash equilibria form a non-trivial subspace,
and no single equilibrium is asymptotically stable. In this situation,
Theorem 3.8 in \citet{de2025optimistic} still yields global exponential
convergence to \emph{some} Nash equilibrium under the assumptions
$\teig{BA} \cup \teig{AB} \subset (-\infty, 0]$,
$\etax\etay\SR{-BA} < \frac14$, and $\Lambda_{A,B}$ diagonalizable or
$A, B$ are square matrices and invertible. The modification for
$\etax \neq \etay$ is again the substitution $\eta^2 \to \etax\etay$
in the spectral computation.

   \section{Special Games}\label{apx:games}
 
 \paragraph{Matching Pennies}\label{apx:pm}
  Define the zero-sum game $\Gamma(A,-A^\top)$ with
 \[ A = \begin{bmatrix}\phantom{-} 1 & -1\\ -1&\phantom{-}1 \end{bmatrix}\, .\]
 The unique Nash equilibrium is at the uniform distribution for both players. Applying Theorem~\ref{thm:2x2games} gives that $\Delta_A = -\Delta_B = 4$. Hence, the Nash equilibrium is asymptotically stable if $\etax \etay < \frac{2m -1}{m^2(2m+1)}$. 
 
 \paragraph{Non-Zero-Sum Matching Pennies}\label{apx:nzspm}
 Define a non-zero-sum game $\Gamma(A, B)$ with
 \[ A = \begin{bmatrix} -1 & \phantom{-}1\\ \phantom{-}3&-1 \end{bmatrix} \quad \text{and}\quad B = \begin{bmatrix} \phantom{-}2 &- 1\\ -1& \phantom{-}1 \end{bmatrix}\, . \]
The unique Nash equilibrium is at $x^\star = [2/5 , 3/5]^\top$, $y^\star = [1/3, 2/3]^\top$. When applying Theorem~\ref{thm:2x2games} we observe that $\Delta_A = -6$ and $\Delta_B = 5$, hence the Nash equilibrium is asymptotically stable if $\etax \etay <   \frac{5}{8} \frac{2m -1}{m^2(2m+1)}$.
 
 \paragraph{Rock-Paper-Scissors and Variants}\label{apx:rpsVar}
 This is a zero-sum game $\Gamma(A,-A^\top)$ with 
 \begin{align*}
A = \begin{bmatrix}  
    \phantom{-}0     & \phantom{-}  a   & -1\\
     -1    &\phantom{-} 0    &\phantom{-}\frac{1}{a}\\
     \phantom{-}1   &-1    &\phantom{-}0    \end{bmatrix} \, ,
    \end{align*} 
    and $a >0$. The equilibrium is at    $x^\star = \frac{1}{2a+1} [1, a,a]^\top$ and $y^\star = \frac{1}{a+2}[1,1,a]^\top$.

  \paragraph{A $4\times 4$ Game }\label{apx:4x4game}
 The following non-zero-sum game has a unique non-fully supported Nash equilibrium.
 \begin{align*}
A = \begin{bmatrix}  
    -6     &\phantom{-}5    &-6     &\phantom{-}7\\
     \phantom{-}1    &-1    &-9    &10\\
    -7    &-7    &-1     &\phantom{-}6\\
    10    &-1     &\phantom{-}4    &-4 \end{bmatrix}
    \qquad
B = \begin{bmatrix}
     \phantom{-}7    &-2     &\phantom{-}6     &\phantom{-}2\\
     \phantom{-}1     &\phantom{-}7     &3   &-10\\
    -6     &\phantom{-}4     &\phantom{-}7    &-7\\
     \phantom{-}9    &-7    &-9     &\phantom{-}6\end{bmatrix}\, .
     \end{align*}
 The Nash equilibrium is at $x^\star = [0,4/9, 0, 5/9]^\top$ and $y^\star = [14/23,0,0, 9/23]^\top$.

 \end{appendices}

 \bibliographystyle{plainnat}
 \bibliography{literatureTTS}
 
\end{document}